\documentclass{birkjour}
\numberwithin{equation}{section}
\newtheorem{theorem}{Theorem}[section]
\newtheorem{lemma}[theorem]{Lemma}
\theoremstyle{definition}

\newtheorem{example}[theorem]{Example}

\newtheorem{corollary}[theorem]{Corollary}

\theoremstyle{remark}
\newtheorem{remark}[theorem]{Remark}
\numberwithin{equation}{section}

\newcommand{\bC}{{\mathbb C}}

\newcommand{\bZ}{{\mathbb Z}}

\newcommand{\cA}{{\mathcal A}}
\newcommand{\cB}{{\mathcal B}}

\newcommand{\cH}{{\mathcal H}}
\newcommand{\cC}{{\mathcal C}}

\newcommand{\cL}{{\mathcal L}}
\newcommand{\cS}{{\mathcal S}}

\renewcommand{\Im}{\,\mathrm{Im}\,}   
\def\idty{{\mathchoice {\mathrm{1\mskip-4mu l}} {\mathrm{1\mskip-4mu l}} %
{\mathrm{1\mskip-4.5mu l}} {\mathrm{1\mskip-5mu l}}}}

\newcommand{\Tr}{\mbox{Tr}}

\newcommand{\be}{\begin{equation}}
\newcommand{\ee}{\end{equation}}
\newcommand{\bea}{\begin{eqnarray}}
\newcommand{\eea}{\end{eqnarray}}
\newcommand{\beann}{\begin{eqnarray*}}
\newcommand{\eeann}{\end{eqnarray*}}

\newcommand{\e}{{\epsilon}}
\newcommand{\lb}{{\lambda}}

\usepackage{color}

\newcommand{\thedate}

\begin{document}

%
%
%
%
%
%
%
%
%


\title[Non-equilibrium states of a photon cavity]
{Non-equilibrium states of a photon cavity pumped by an atomic beam}

\author[B. Nachtergaele]{Bruno Nachtergaele}
\address{Department of Mathematics\\
University of California, Davis\\
Davis, CA 95616\\
USA}
\email{bxn@math.ucdavis.edu}

\author[A. Vershynina]{Anna Vershynina}
\address{Department of Mathematics\\
University of California, Davis\\
Davis, CA 95616\\
USA}
\email{annavershynina@gmail.com}

\author[V. A. Zagrebnov]{Valentin A. Zagrebnov}
\address{Laboratoire d'Analyse, Topologie, Probabilit\'{e}s - UMR 7353\\
CMI - Technop\^{o}le Ch\^{a}teau-Gombert \\
39, rue F. Joliot Curie, 13453 Marseille Cedex 13, France\\
and\\
D\'{e}partement de Math\'{e}matiques, Aix-Marseille Universit\'{e}  \\
Luminy - Case 901, Marseille 13288, Cedex 09, France}
\email{Valentin.Zagrebnov@latp.univ-mrs.fr \\
Valentin.Zagrebnov@univ-amu.fr}

\thanks{This work was completed with the support of the France-Berkeley Fund and the National Science
Foundation under grants VIGRE DMS-0636297 and DMS-10009502}

\subjclass{Primary 82C10; Secondary 47D06, 46L60, 81S22.}

\keywords{Non-equilibrium quantum theory, open quantum systems, repeated
interactions, time dependent interactions, $W^*$-dynamical systems, energy-entropy variations,
photon cavity.}

\date{\today}

\begin{abstract}
We consider a beam of two-level randomly excited atoms that pass
one-by-one through a one-mode cavity. We show that in the case of an ideal
cavity, i.e. no leaking of photons from the cavity, the pumping
by the beam leads to an unlimited increase in the photon number in the
cavity. We derive an expression for the mean photon number for all times.
Taking into account leaking of the cavity, we prove that
the mean photon number in the cavity stabilizes in time. The limiting
state of the cavity in this case exists and it is independent of the initial
state. We calculate the characteristic functional of this non-quasi-free
non-equilibrium state. We also calculate the total energy variation in both the ideal
and the open cavities as well as the entropy production in the ideal cavity.

\end{abstract}

\maketitle
\tableofcontents

\section{Model and Results}

\subsection{The Model}

Our model consists of a beam $\mathcal{A}$ of two-level randomly excited atoms
that pass one-by-one through a photon cavity $\mathcal{C}$. During the passage time
$\tau$ the atom in the cavity interacts with the cavity field.

The cavity is a one-mode resonator described by a quantum harmonic oscillator with Hamiltonian
$H_{\mathcal{C}} = \epsilon \, b^*b\otimes\idty$ acting on the Hilbert space $\cH_{\mathcal{C}}$.
Here $b^*$ and $b$ stand for boson creation and annihilation operators with canonical commutation
relations (CCR): $[b \, , b^*]=\idty$, $[b\, , b]=[b^*, b^*]=0$.

The beam of two-level atoms with energy levels $0$ and $E>0$, is described by
a chain $H_{\mathcal{A}}= \sum_{k\geq1} H_{\mathcal{A}_k}$ of
individual atoms $\mathcal{A}_k$ with Hamiltonian
$H_{\mathcal{A}_k}=\idty \otimes E \, \eta_k$ in the Hilbert space $\cH_{\mathcal{A}}=
\otimes_{k\geq 1}\cH_{\mathcal{A}_k}$. Here for any $k\geq 1$,
$\cH_{\mathcal{A}_k} = \mathbb{C}^2$ and the individual atomic operator
$\eta_k := (\sigma^{z} + I)/2$, where $\sigma^{z}$ is the third Pauli matrix. The eigenvectors
$\psi_{k}^{\pm}$: $\eta_k \psi_{k}^{+} = \psi_{k}^{+}$ and
$\eta_k \psi_{k}^{-} =0$, are interpreted as the {excited} and the {ground} states of the atom,
respectively.

The initial state of the system is the product state of the cavity and the states of each individual atom:
\begin{equation}\label{In-State}
\rho_{\mathcal{S}}: =\rho_{\mathcal{C}}\otimes\bigotimes_{k\geq 1}\rho_{k} \ .
\end{equation}
Here $\rho_{\mathcal{C}}$ is the initial state of the cavity, which we assume to be normal,
i.e., given by a density matrix $\rho_{\mathcal{C}}\in\mathfrak{C}_{1}(\cH_{\mathcal{C}})$,
the space of the trace-class operators  on $\cH_{\mathcal{C}}$,
and $\rho_{\mathcal{A}}:= \bigotimes_{k\geq 1}\rho_{k}$ is the state of the beam.

In this paper, we suppose that the atomic states $\{\rho_{k}\}_{k\geq1}$ on the algebras
$\{\mathfrak{C}_1(\cH_{\mathcal{A}_k})\}_{k\geq1}$ are diagonal and identical (\textit{homogeneous} beam), hence of
the form
\begin{equation}\label{atom-state}
\rho_k=\left( \begin{array}{cc}
p & 0\\
0& 1-p
\end{array} \right)\ , \ \  p:=\Tr_{\cA_k}(\eta_k \, \rho_k) \ .
\end{equation}
The parameter $0 \leq p \leq 1$ denotes probability that atom in beam is in its excited state $\psi_{k}^{+}$.

For simplicity we consider  our model in the regime when at any moment only {\textit{one}} atom is
present in the cavity. Physically, this corresponds to a special \textit{tuning} of the cavity size
$l$ and the interatomic distance $d =l$. Then the cavity-atom interaction has piecewise constant
time-dependence, and we take it of the form:
\begin{equation}\label{W-int}
K_k(t)=\chi_{[(k-1)\tau, k\tau)}(t) \, \lambda \, (b^*+b)\otimes \eta_k \ .
\end{equation}
Here $\chi_{\Delta}(x)$ is characteristic function of the set $\Delta$.

Since the non-excited atom is not ``visible" by the cavity
($W_k (t) (u \otimes \psi_{k}^{-}) =0$ for any $u \in {\rm{dom}}(H_{\mathcal{C}})$),
the \textit{detuned} case $d > l$, when there is {at most one} atom
inside the cavity at the same time, is also described by a piecewise constant
interaction. This situation can be handled by a small modification of our arguments, but
we will not consider this situation further in this paper.

\begin{remark}\label{Rem-Q-El}
It is convenient, albeit not compulsory, to stick with a quantum description of the atomic beam.
If one restricts oneself to the atomic observables $\eta_k$, which generate a commutative algebra,
one can also consider $H_{\mathcal{A}}$ and interaction
(\ref{W-int}) as matrix representation of the continuous-time Bernoulli process with time-unity $\tau$. Then the
its piecewise constant random realisations
\begin{equation}\label{Bern-proc1}
\widehat{\eta}(t)=
\sum_{k\geq1}(\psi_{k}^{\pm},\chi_{[(k-1)\tau, k\tau)}(t)\eta_k \psi_{k}^{\pm})_{\cH_{\mathcal{A}_k}} \ ,
\end{equation}
are generated by random sequences of atomic operator $\eta_k := (\sigma^{z} + I)/2$ eigenvectors
$\{\psi_{k}^{\pm}\}_{k\geq 1}$ with probabilities $p$ and $1-p$ for eigenvalues $1$ and $0$ respectively.
Here $(\cdot\, , \,\cdot)_{\cH_{\mathcal{A}_k}}$ is the scalar product in the space $\cH_{\mathcal{A}_k}=\mathbb{C}^2$.

Returning back to the quantum electrodynamic origin of the interaction (\ref{W-int}) one observes that
it is completely elastic, since the atomic system does not evolve. The atom remains in the same
state throughout its interaction with the photon field. This may be interpreted as the limit of  {rigid}
atoms that ``kick'' the cavity mode, see \cite{FJMa}.
\end{remark}

The Hamiltonian for the entire system $\mathcal{S}$ acts on the space $\cH_{\mathcal{S}}:= \cH_{\mathcal{C}} \otimes
\cH_{\mathcal{A}}$, and is given by the sum of the Hamiltonians of the cavity and the atoms,
and the interaction between them:
\begin{align}\label{Ham-Model}
H(t)=& \ H_{\mathcal{C}}+\sum_{k\geq 1}(H_{\mathcal{A}_k}+K_k(t))\\
=& \ \epsilon \, b^*b\otimes\idty+\sum_{k\geq1} \idty\otimes E \, \eta_k + \sum_{k\geq 1}\chi_{[(k-1)\tau, k\tau)}(t)
(\lambda \, (b^*+b)\otimes \eta_k) \, . \nonumber
\end{align}
Notice that for the time $t\in[(n-1)\tau, n\tau)$, $n\geq 1$, only the $n$-th atom interacts with the cavity and the
system is \textit{autonomous}.

Projected on the time invariant sequences of $\{\psi_{k}^{\pm}\}_{k\geq 1}$ the system (\ref{Ham-Model}) is a one-mode
cavity Hamiltonian with random interaction driven by the $1$-$0$ Bernoulli process (\ref{Bern-proc1}).

\subsection{Hamiltonian Dynamics of the Ideal Cavity}

Let $t\in[(n-1)\tau, n\tau)$. Then the Hamiltonian (\ref{Ham-Model}) for the $n$-th atom in the cavity takes the form
$H(t) = H_n$, where
\begin{equation}\label{Ham-n}
H_n:=\epsilon \, b^*b\otimes \idty + \sum_{k\geq1}\idty\otimes E \, \eta_k + \lambda \, (b^*+b)\otimes \eta_n \  \ .
\end{equation}

Although the atomic beam is {infinite}, since we will assume that the initial state is normal
and since up to any finite time $t$ only a finite number of atoms have interacted with the cavity,
we can describe the evolved system by \textit{normal} states
$\omega_{\mathcal{S}}(\cdot):= {\rm{Tr}} (\, \cdot \ \rho_{\mathcal{S}})$, which are defined by density matrices
$\rho_{\mathcal{S}}$ from the space of the \textit{trace-class} operators
$\mathfrak{C}_{1}(\cH_{\mathcal{C}} \otimes \cH_{\mathcal{A}})$. Then the partial traces over $\cH_{\mathcal{A}}$
and over $\cH_{\mathcal{C}}$:
\begin{eqnarray}\label{C-A-states}
&&\omega_{\mathcal{C}}(\cdot):= \omega_{\mathcal{S}}(\cdot \otimes \idty) =
{\Tr}_{\mathcal{C}} (\, \cdot \ {\Tr}_{\mathcal{A}}\rho_{\mathcal{S}}) \ , \\
&&\omega_{\mathcal{A}}(\cdot):= \omega_{\mathcal{S}}(\idty \otimes \cdot)=
{\rm{Tr}}_{\mathcal{A}} (\, \cdot \ {\rm{Tr}}_{\mathcal{C}}\rho_{\mathcal{S}}) \ , \nonumber
\end{eqnarray}
define respectively the cavity and the beam states.

Since below we mostly deal with normal states, we will use the terms `state' and `density matrix'
interchangeably, if this does not cause any confusion.

We suppose that initially our system is in a \textit{product-state}
$\rho_{\mathcal{S}}(t)\big|_{t=0}= \rho_{\mathcal{C}} \otimes \rho_{\mathcal{A}}$ of the sub-systems $\mathcal{C}$
and $\mathcal{A}$:
\begin{equation}\label{normal-states}
\omega^{0}_{\mathcal{S}}(\cdot)= {\rm{Tr}} (\, \cdot \ \rho_{\mathcal{C}} \otimes \rho_{\mathcal{A}}) \ \ .
\end{equation}

For any states $\rho_{\mathcal{C}}$ on $\mathfrak{A}(\cH_{\mathcal{C}})$ and $\rho_{\mathcal{A}}$ on
$\mathfrak{A}(\cH_{\mathcal{A}})$
the Hamiltonian dynamics of the system is defined by (\ref{Ham-Model}), or by the quantum time-dependent
Liouvillian generator:
\begin{equation}\label{L-Gen-t}
L(t)(\rho_{\mathcal{S}}(t)):= -i \, [H(t),\rho_{\mathcal{S}}(t)] \ .
\end{equation}
Then the state $\rho_{\mathcal{S}}(t) := (\rho_{\mathcal{C}} \otimes \rho_{\mathcal{A}})(t)$ of the total system
at the time $t$ is a solution of the non-autonomous Cauchy problem corresponding to Liouville differential equation
\begin{equation}\label{Liouville's-DE}
\frac{d}{dt}\rho_{\mathcal{S}}(t) = L(t)(\rho_{\mathcal{S}}(t)) \ , \ \rho_{\mathcal{S}}(t)\big|_{t=0}=
\rho_{\mathcal{C}} \otimes \rho_{\mathcal{A}} \ .
\end{equation}
We denote by $\omega_{\mathcal{S}}^{t}(\cdot):= {\rm{Tr}} (\, \cdot \ \rho_{\mathcal{S}} (t))$ the system time
evolution due to (\ref{Liouville's-DE}) for the initial product state $\omega_S(\cdot)$ (\ref{normal-states}).

Notice that in general (see e.g. \cite{BrRo1}) the Hamiltonian evolution (\ref{Liouville's-DE}) with time-dependent
generator is a family of unitary mappings $\{U_{t,s}\}_{s \leq t}$ (called also evolution operators or propagators)
\begin{equation}\label{propagator}
i \, \frac{d}{dt} U_{t,s} =  H(t)\, U_{t,s} \ , \ s < t \ ,
\end{equation}
with the composition rule:
\begin{equation}\label{cocycle}
U_{t,s}= U_{t,r} \, U_{r,s} \ ,  \ \ \ {\rm{for}} \ \ \  s\leq r \leq  t  \ .
\end{equation}
Here $s,r,t \in \mathbb{R}$ and $U_{t,t} = \idty$. Then solution of (\ref{Liouville's-DE}) is
\begin{equation}\label{solution}
\rho_{\mathcal{S}}(t) = U_{t,s} \, \rho_{\mathcal{S}}(s)\, U_{t,s}^{-1} =: T_{t,s}(\rho_{\mathcal{S}}(s)) \ ,
\end{equation}
where by (\ref{cocycle}) the mapping $T_{t,s} (\cdot) = T_{t,r}(T_{r,s} (\cdot))$.

For our model with the tuned repeated interactions the form of the evolution operator (\ref{cocycle}) and the solution
(\ref{solution}) are considerably simplified. Indeed, our system (\ref{Ham-Model}) is \textit{autonomous} for each
interval $[(k-1)\tau, k\tau)$. Then by virtue of (\ref{Ham-n}) and (\ref{L-Gen-t}) the Liouvillian generator
\begin{equation}\label{L-Gen-n}
L(t)(\cdot)= L_{k}(\cdot) = -i[H_k, \, \cdot \, ] \ ,  \ \ {\rm{for}} \ \ t\in[(k-1)\tau, k\tau) \ ,  \  k \geq 1 \ ,
\end{equation}
is piecewise constant (time-independent). For any $t \geq 0$ one has the representation
\begin{equation}\label{t}
t:=n(t)\tau + \nu(t) \ , \  n(t) = [t/\tau] \ \ {\rm{and}} \ \ \nu(t)\in[0, \tau) \ ,
\end{equation}
where $[x]$ denotes the integer part of $x \in \mathbb{R}$. Then by (\ref{solution}) and (\ref{L-Gen-n}) the solution of
the Cauchy problem (\ref{Liouville's-DE}) for $t\in [(n-1)\tau, n\tau)$ takes the form:
\begin{equation}\label{Sol-Liouv-Eq}
\rho_{\mathcal{S}}(t)=T_{t,0}(\rho_{\mathcal{C}} \otimes \rho_{\mathcal{A}}):
= e^{\nu(t) L_{n}}e^{\tau L_{n-1}} \, ... \
e^{\tau L_{2}}e^{\tau L_{1}}(\rho_{\mathcal{C}} \otimes \rho_{\mathcal{A}}) \ .
\end{equation}
By (\ref{cocycle}) and (\ref{L-Gen-n}) the mapping $T_{t,0}$ is the composition
\begin{equation}\label{T-t}
T_{t,0} = T_{t,(n-1)\tau}  \prod_{k= n-1}^{k=1} T_{k}
\end{equation}
of the \textit{one-step} evolution maps defined as
\begin{equation}\label{One-step-T}
T_{k} := T_{k\tau,(k-1)\tau} = e^{\tau L_{k}} \  \   \  {\rm{and}} \, \  \  T_{t,(n-1)\tau}= e^{\nu(t) L_{n}} \ .
\end{equation}
Consequently, the evolution of the initial state of the system (\ref{normal-states}) can be expressed as
\begin{equation}\label{S-state-evol}
\omega_{\mathcal{S}}^{t}(\cdot) = {\rm{Tr}} (\, \cdot \ T_{t,0}(\rho_{\mathcal{C}} \otimes \rho_{\mathcal{A}})) \ .
\end{equation}

The mathematical study of this kind of dynamics, for different types of repeated interaction $K_n (t)$ (\ref{Ham-Model}),
was initiated in papers by Attal, Bruneau, Joye, Merkli, Pautrat, Pillet:
\cite{BJM06,APa06,AJ107,AJ207,BJM08,BPi09,BJM10}. These works provide a rigorous framework for such physical
phenomenon as the "one-atom maser", see e.g. \cite{MWM,FJMb}.
The important mathematical aspect of repeated interactions is a piecewise constant (random) Hamiltonian dynamics like
(\ref{Sol-Liouv-Eq}), which in certain limit of $\tau \rightarrow 0$ may produce an effective quantum Markovian
dynamics, which drives the total system to an asymptotic state. In our case, we will be able
to obtain the exact asymptotics of the dynamics of our simple model directly, i.e.,
without taking a limit.

In the next part of the present paper (Section \ref{HD}) we exploit specific structure of the
Hamiltonian dynamics (\ref{Sol-Liouv-Eq}) and a special form of interaction (\ref{W-int}) to work out an effective
Hamiltonian
evolution of the {perfectly isolated} cavity $\mathcal{C}$. In this case, the pumping of the cavity by the atomic
beam leads to an unlimited growth of the number of photons. This means that the limiting state
of the cavity will not be described by a density matrix. For this case, our results concern evolution
of the photon-number expectation $N(t)$ in the reduced by the partial trace (\ref{C-A-states}), (\ref{S-state-evol})
time-dependent cavity state
\begin{equation}\label{C-state-t}
\omega_{\mathcal{C}}^{t}(\cdot):= \omega_{\mathcal{S}}^{t}(\cdot \otimes \idty)=
\Tr_{\mathcal{C}}( \cdot \ \rho_{\mathcal{C}}(t)) \ \ {\rm{for}} \ \
\rho_{\mathcal{C}}(t):= \Tr_{{\mathcal{A}}} \, (T_{t,0}(\rho_{\mathcal{C}} \otimes \rho_{\mathcal{A}})) \ .
\end{equation}

Note that in our model, see (\ref{Ham-Model}) and (\ref{Ham-n}), the atomic states $\rho_k$ do not evolve:
\begin{equation}\label{commut-atoms}
[\idty \otimes \rho_{k}, H(t)] = 0 \ , \  k \geq 1 \ .
\end{equation}

The form of the initial state (\ref{In-State}) together with (\ref{Sol-Liouv-Eq}) and (\ref{C-state-t}) determine
a \textit{discrete} time
evolution for the cavity state: $\rho_{\mathcal{C}}^{(n)}:= \rho_{\mathcal{C}}(t=n\tau)$, with the recursive formula
\begin{align}\label{C-state-n}
\rho_{\mathcal{C}}^{(n)}:=& \Tr_{{\mathcal{A}}}\rho_{\mathcal{S}}(t)=\Tr_{{\mathcal{A}}}[e^{\tau L_n}...
e^{\tau L_2}e^{\tau L_1}(\rho_{\mathcal{C}}\otimes\bigotimes_{k=1}^n\rho_{k})] \\
=& \Tr_{{\mathcal{A}_n}}[e^{\tau L_n}\{\Tr_{{\mathcal{A}_{n-1}}}...\Tr_{{\mathcal{A}_1}}e^{\tau L_{n-1}}...
e^{\tau L_2}e^{\tau L_1}(\rho_{\mathcal{C}}\otimes\bigotimes_{k=1}^{n-1}\rho_{k})\}\otimes\rho_{n}] \nonumber \\
=& \Tr_{{\mathcal{A}_n}}[e^{\tau L_n}(\rho_{\mathcal{C}}^{(n-1)}\otimes\rho_{n})] \ . \nonumber
\end{align}
Here we denoted the partial trace over the $k$-th atom space $\cH_{\cA_k}$ by $\Tr_{\cA_k}(\cdot)$.

For any density matrix $\rho \in \mathfrak{C}_{1}(\cH_{\mathcal{C}})$ corresponding to a normal state on
the operator algebra $\mathfrak{A}(\cH_{\mathcal{C}})$ we define the mapping $\cL: \rho \mapsto \cL(\rho)$, by
\begin{equation}\label{cL}
\cL(\rho):=\Tr_{\mathcal{A}_k}(e^{\tau L_k}(\rho\otimes\rho_{k})) = \Tr_{\mathcal{A}_k}[e^{-i\tau H_k}
(\rho\otimes\rho_{k}) e^{i\tau H_k}] \ .
\end{equation}
Here the last equality is due to (\ref{L-Gen-n}). Note that the mapping (\ref{cL}) does not depend on $k \geq 1$,
since the atomic
states $\{\rho_{k}\}_{k\geq 1}$ are homogeneous (\ref{atom-state}). Then the cavity state at $t=k\tau$ is defined
by the $k$-th
power of $\cL$:
\begin{equation}\label{cavity state at t=n}
\rho_{\mathcal{C}}^{(k)}=\cL(\rho_{\mathcal{C}}^{(k-1)})=\cL^k(\rho_{\mathcal{C}}) \ .
\end{equation}
Therefore, by (\ref{Sol-Liouv-Eq}), (\ref{C-state-n}) and (\ref{cavity state at t=n}) one obtains that for any time
$t=(n-1)\tau+\nu(t)$, where $\nu(t) \in[0, \tau)$, the cavity state is
\begin{equation}\label{cavity state at t}
\rho_{\mathcal{C}}(t)=\Tr_{\mathcal{A}_{n}}[e^{\nu(t) L_{n}}(\cL^{n-1}(\rho_{\mathcal{C}})\otimes\rho_{n})] \ .
\end{equation}

Our first result concerns the evolution of the expectation value $N(t)$ of the photon-number operator
$\hat{N}:=b^*b$ in
the cavity at the time $t$:
\begin{equation}\label{N(t)}
N(t): = \omega_{\mathcal{C}}^{t}(b^*b)= \Tr_{\mathcal{C}}(b^*b \ \rho_{\mathcal{C}}(t)) \ .
\end{equation}
For  $t=n\tau$ the state of the cavity can be expressed using (\ref{cavity state at t=n}). Then (\ref{N(t)}) yields
\begin{equation}\label{mean-photon-number-n}
N(n\tau)=\Tr_{\mathcal{C}}(b^*b \ \cL^n(\rho_{\mathcal{C}})) \ .
\end{equation}
In the theorem below we suppose that the initial cavity state $\omega_{\mathcal{C}}^{t}|_{t=0}(\cdot)=
\omega_{\mathcal{C}}(\cdot)$ is gauge invariant, i.e.
$e^{i \alpha \hat{N}}\rho_{\mathcal{C}} e^{- i \alpha \hat{N}}= \rho_{\mathcal{C}}$.

\begin{theorem}\label{N of photons}
Let $\rho_{\mathcal{C}}$ be a gauge-invariant state. Then for a homogeneous beam the expectation value
{\rm{(}}\ref{mean-photon-number-n}{\rm{)}} of the photon-number operator in the cavity at the moment $t=n\tau$ is equal to
\begin{equation}\label{number of photons_Hamiltonian}
N(t)=N(0)+ n \, p(1-p) \, \frac{2\lambda^2}{\epsilon^2} \, (1-\cos\epsilon\tau) +
p^2 \, \frac{2\lambda^2}{\epsilon^2}(1-\cos n\epsilon\tau)\ .
\end{equation}
\end{theorem}

If for the initial state $\rho_{\mathcal{C}}$ one takes in the theorem the Gibbs state for photons at the inverse
temperature $\beta$:
\begin{equation}\label{Gibbs-photons}
\rho_{\mathcal{C}}^{\beta}= {e^{-\beta\epsilon b^*b}}/{\Tr_{{\mathcal{C}}} e^{-\beta\epsilon b^*b}} \ ,
\end{equation}
{then the the first term in (\ref{number of photons_Hamiltonian}) is simply $N(0)= (e^{\beta\epsilon}-1)^{-1}$.}

If the initial cavity state $\rho_{\mathcal{C},r,\phi}$ is \textit{not} gauge-invariant with the breaking
gauge-invariance parameter:
$r\, e^{i \phi}:= \omega_{\mathcal{C},r,\phi}(b)= \Tr_{\mathcal{C}}(b \ \rho_{\mathcal{C},r,\phi})$, then
instead of (\ref{number of photons_Hamiltonian}) one gets by Lemma \ref{cL-T(B)}
\begin{equation}\label{number of photons not gauge}
\omega_{\mathcal{C},r,\phi}^{t}(\hat{N}) = N(t)+
p \, \frac{2 \lambda r}{\epsilon}\, [\cos\phi -\cos(n\epsilon\tau - \phi)]\ , \ t=n\tau \ ,
\end{equation}
where $\omega_{\mathcal{C},r,\phi}^{t}(\hat{N})\mid_{t=0} \ = N(0) \geq r^2$. The same lemma yields also
the positivity of (\ref{number of photons not gauge}) for any time $t \geq 0$.

\begin{remark}\label{p-dependence}
{If {all} atoms in the beam are in the  {ground-state}} $\otimes_{k\geq1}
\psi_{k}^{-}$, one has: $(E \, \eta_k) \psi_{k}^{-} = 0$,  and  $\rho_{k}\mid_{p=0} \ = (\psi_{k}^{-},
\cdot)_{\cH_{\mathcal{A}_k}} \psi_k^{-}$ is the projection operator. Then by (\ref{number of photons_Hamiltonian})
and (\ref{number of photons not gauge})
$\omega_{\mathcal{C},r,\phi}^{t}(\hat{N}) = \omega_{\mathcal{C},r,\phi}^{t}(\hat{N})\mid_{t=0}$.
Since the mean cavity energy is
$\varepsilon(t) =\Tr_{\mathcal{C}}( H_{\mathcal{C}} \ \rho_{\mathcal{C}}(t)) = \epsilon \, N(t)$.
Therefore, there is no energy transfer from the beam to the cavity when $p=0$. Whereas
in the other extreme case, when  {all} atoms are  {excited} $\rho_{k}(\eta_k) =1$, we obtain
\begin{align}\label{number of photons_Hamiltonian-p=1}
\omega_{\mathcal{C},r,\phi}^{n\tau}(\hat{N})=&\omega_{\mathcal{C},r,\phi}^{n\tau}|_{n=0}(\hat{N})+
\frac{2\lambda^2}{\epsilon^2}(1-\cos n\epsilon\tau) + \\
+ &\frac{2 \lambda r}{\epsilon}\, [\cos\phi -\cos(n\epsilon\tau - \phi)] \ . \nonumber
\end{align}

In case of the gauge-invariant initial state ($r=0$) the mean value of the photon number
(\ref{number of photons_Hamiltonian-p=1})
satisfies the estimate $N(t) \geq N(0)$ and it is oscillating between initial value $N(0)$ and
$N(0)+ {4\lambda^2}/{\epsilon^2}$
with the cavity resonant frequency $\epsilon$.
\end{remark}
\begin{remark}\label{Rabi-dual-evol}
These \textit{Rabi oscillations} {\rm{(}}see e.g. \cite{FJMb}{\rm{)}} are a simple consequence of
non-trivial Heisenberg time evolution of the cavity boson operators generated by Hamiltonian (\ref{Ham-n}).
For example, if $t\in[0, \tau)$, i.e. $n=1$, then we obtain (cf. Lemma \ref{n-dual-on-b}):
\begin{eqnarray}
&& t: b \otimes \idty \mapsto  T_{t,0}^{\ast}(b \otimes \idty):=
e^{i t H_{1}} (b \otimes \idty)e^{-i t H_{1}}  \nonumber \\
&& = e^{t L_{n=1}^\ast}(b \otimes \idty)= e^{-i \epsilon t} b \otimes \idty -
\idty \otimes \frac{\lambda}{\epsilon} \eta_{1} (1 - e^{-i \epsilon t})\, \label{b-evol} .
\end{eqnarray}
Here evolution maps on the algebra $\mathfrak{A}(\cH_{\mathcal{C}})\otimes\mathfrak{A}(\cH_{\mathcal{A}})$:
\begin{equation}\label{T-adj}
T_{k}^{\ast} = e^{\tau L_{k}^\ast} \ , \
T_{t,(n-1)\tau}^{\ast} = e^{\nu(t) L_{n}^\ast}  \ , \ t = (n-1)\tau + \nu(t) \ ,
\end{equation}
are adjoint to (\ref{T-t}), (\ref{One-step-T}) by \textit{duality} with respect to the state (\ref{normal-states}),
see Appendix A.2. Here $L_{k}^\ast$ denotes operator, which is adjoint to the Liouvillian generator (\ref{L-Gen-n}).
Then (\ref{S-state-evol}) and (\ref{T-t}),(\ref{T-adj}) yield
\begin{equation}\label{S-state-evol-adj}
\omega_{\mathcal{S}}^{t}(\cdot)= {\rm{Tr}}\, (T_{t,0}^{\ast}(\cdot) \ \rho_{\mathcal{C}} \otimes \rho_{\mathcal{A}})
\ \ ,\ \
T_{t,0}^{\ast} = \prod_{k=1}^{n-1}e^{\tau L_{k}^\ast} \  e^{\nu(t) L_{n}^\ast} \ .
\end{equation}
\end{remark}
\begin{remark}\label{G-Inv-breaking}
The Hamiltonian evolution (\ref{b-evol}) \textit{breaks down} the gauge invariance of the initial state
$\omega_{\mathcal{C}}(\cdot)= {\Tr}_{\mathcal{C}} (\, \cdot \ \rho_{\mathcal{C}})$.
Indeed, by (\ref{C-A-states}) and (\ref{normal-states}) one gets for the initial gauge-invariant cavity state:
$\omega_{\mathcal{C}}(b) = \omega_{\mathcal{S}}(b \otimes \idty) = 0$.
Since by (\ref{C-state-t}) and (\ref{S-state-evol-adj}) we have
\begin{equation}\label{G-InvBreak1}
\omega_{\mathcal{C}}^{t}(b) = \omega_{\mathcal{S}}^{t}(b \otimes \idty) =
{\rm{Tr}}\, (T_{t,0}^{\ast}(b \otimes \idty) \ \rho_{\mathcal{C}} \otimes \rho_{\mathcal{A}}) \ ,
\end{equation}
(\ref{atom-state}) and (\ref{b-evol}) imply that for $t\in[0, \tau)$
\begin{equation}\label{G-InvBreak2}
\omega_{\mathcal{C}}^{t}(b) = p \frac{\lambda}{\epsilon} (e^{-i \epsilon t} -1) \ .
\end{equation}
This property of dynamics is due to interaction (\ref{Ham-n}), and it has also a non-trivial impact for the
open cavity evolution.
\end{remark}
\subsection{Quantum Dynamics of Open Cavity}

To make a contact of our model with a more physically realistic situation, we consider an open cavity.
This allows the photons to leak \textit{out} of the cavity, but also to diffuse \textit{in} from the environment.
In Theorems \ref{cavity-lim-state-sigma} and \ref{number of photons-sigma} we show that if the leaking rate is
greater than the rate the environmental pumping, the photon number in the cavity stabilizes at a finite mean value.

We consider this case in the framework of Kossakowski-Lindblad extension of the Hamiltonian Dynamics to
\textit{irreversible}
Quantum Dynamics (see Appendix A.1 and \cite{AlFa,AJPII}, ) with time-dependent generator:
\begin{align}\label{Generator}
L_{\sigma}(t)(\rho_{\mathcal{S}}(t)):= &-i[H(t),\rho_{\mathcal{S}}(t)]+
\sigma_- \ b\otimes\idty \ \rho_{\mathcal{S}}(t)\ b^* \otimes\idty -\frac{\sigma_-}{2}
\{bb^*\otimes\idty,\rho_{\mathcal{S}}(t)\}  \nonumber\\
&+ \sigma_+ \ b^*\otimes\idty \ \rho_{\mathcal{S}}(t) \ b \otimes\idty -
\frac{\sigma_+}{2}\{b^*b\otimes\idty, \rho_{\mathcal{S}}(t)\} \ ,
\end{align}
for the \textit{complete positive} evolution of total system, when the initial
state is $\rho_{\mathcal{S}}(t=0) = \rho_{\mathcal{C}}\otimes\rho_{\mathcal{A}}$ (\ref{normal-states}).
For $\sigma_{\mp}> 0$ the $\sigma$-part of this generator (cf (\ref{L-Gen-t})) corresponds to the
non-Hamiltonian part of dynamics. Here $\sigma_-$ describes the rate of the photons \textit{leaking} out of the open
cavity into the environment, whereas $\sigma_+$ corresponds to the cavity \textit{pumping} rate due to the
photon infiltration from the environment. It has to be distinguished from the pumping
mechanism due to the interaction with the atomic beam.

Note that similar to (\ref{Liouville's-DE}) the evolution of the state is defined by the solution of the non-autonomous
Cauchy problem corresponding to the  {time-dependent} generator (\ref{Generator}). The proof of existence of this
solution is a non-trivial problem (Appendix A.1), see for example \cite{NVZ} for the case of lattice systems and
 {bounded} generators. For  {unbounded} generators and for a general setting the proof that
the Kossakowski-Lindblad generator in the form (\ref{Generator}) corresponds to a properly defined continuous evolution
is more involved \cite{NZ,VWZ}. A separate problem is to prove that this map is \textit{trace-preserving} and
verifies the property of \textit{complete positivity}, which are indispensable for correct description of the
open system evolution, see Appendix A.1 and the references there.

To avoid these complications we consider the case of the {\textit{tuned}} repeated interaction, when the
Hamiltonian dynamics is piecewise  {autonomous} for each interval $[(k-1)\tau, k\tau)$. Then for
$t\in[(k-1)\tau,k\tau)$ the generator (\ref{Generator}) has the form:
\begin{align}
L_{\sigma,k}(\cdot):=& -i[H_k,\cdot]+ \sigma_{-} \ b\otimes\idty \ (\cdot)\ b^* \otimes\idty -
\frac{\sigma_-}{2}\{bb^*\otimes\idty, \cdot\} \nonumber \\
& +\sigma_{+} \ b^* \otimes\idty \ (\cdot)\ b \otimes\idty-\frac{\sigma_+}{2}\{b^*b \otimes\idty, \cdot\} \ , \ \ \
\ k \geq 1  \ , \label{Generator-KL}
\end{align}
and $L_{\sigma, k}\big|_{\sigma_{\mp} = 0} = L_{k}$, see (\ref{L-Gen-n}).
Hence, the solution of the non-autonomous Cauchy problem
\begin{equation}\label{Liouville's-DE-sigma}
\frac{d}{dt}\rho_{\mathcal{S,\sigma}}(t) = L_{\sigma}(t)(\rho_{\mathcal{S,\sigma}}(t)) \ , \
\rho_{\mathcal{S,\sigma}}(t)\big|_{t=0}= \rho_{\mathcal{C}} \otimes \rho_{\mathcal{A}} \ ,
\end{equation}
for the piecewise \textit{constant} generators (\ref{Generator-KL}), has the same form as for the ideal cavity,
$\sigma_{\mp} = 0$ (\ref{Sol-Liouv-Eq}), (\ref{T-t}):
\begin{align}\label{Sol-Liouv-Eq-sigma}
\rho_{\mathcal{S,\sigma}}(t)&=T_{t,0}^\sigma(\rho_{\mathcal{C}} \otimes \rho_{\mathcal{A}}):= \\
&= e^{\nu(t) L_{\sigma, n}} e^{\tau L_{\sigma, n-1}} \, ... \
e^{\tau L_{\sigma, 2}}e^{\tau L_{\sigma, 1}}(\rho_{\mathcal{C}} \otimes \rho_{\mathcal{A}})
\, . \nonumber
\end{align}
Here $t = (n-1)\tau + \nu(t)$ and the mapping $T_{t,0}^\sigma$ is the composition of the one-step
non-Hamiltonian evolution maps defined by (\ref{Generator-KL}):
\begin{equation}\label{One-step-T-sigma}
T_{k}^\sigma := T_{k\tau,(k-1)\tau}^\sigma = e^{\tau L_{\sigma, k}} \  \   \  {\rm{and}} \, \  \
T_{t,(n-1)\tau}^\sigma= e^{\nu(t) L_{\sigma,n}} \ .
\end{equation}

By duality with respect to the initial state
$\omega_{\mathcal{S}}^{0}(\cdot)=\Tr ( \cdot \ \rho_{\mathcal{C}} \otimes \rho_{\mathcal{A}})$ one can now
define the adjoint evolution mapping $\{(T_{t,0}^\sigma)^*\}_{t\geq 0} $ by the relation
\begin{equation}\label{dual-sigma}
\omega_{\mathcal{S},\sigma}^{t}(A)= \Tr ( A \ T_{t,0}^{\sigma}(\rho_{\mathcal{C}} \otimes \rho_{\mathcal{A}}))=
\omega_{\mathcal{S}}^{0}((T_{t,0}^\sigma)^*(A)) \ ,
\end{equation}
for any $A \in \mathfrak{A}(\cH_{\mathcal{C}} \otimes \cH_{\mathcal{A}})$. Then similar to the nonleaky case
(\ref{S-state-evol-adj}) we obtain for any $t = (n-1)\tau + \nu(t)$ that
\begin{equation}\label{S-state-evol-adj-sigma}
\omega_{\mathcal{S},\sigma}^{t}(\cdot)=
{\rm{Tr}}\, ((T_{t,0}^\sigma)^{\ast}(\cdot) \ \rho_{\mathcal{C}} \otimes \rho_{\mathcal{A}})
\ \ ,\ \
(T_{t,0}^\sigma)^{\ast} = \prod_{k=1}^{n-1}e^{\tau L_{\sigma,k}^\ast} \  e^{\nu(t) L_{\sigma,n}^\ast} \ .
\end{equation}
Here $\{L_{\sigma,k}^\ast\}_{k \geq 1}$ are generators, which are adjoint to (\ref{Generator-KL}).

Since the restriction of (\ref{Sol-Liouv-Eq-sigma}) to the dynamics of a cavity state
$\rho \in \mathfrak{C}_{1}(\cH_{\mathcal{C}})$
is the partial trace over beam states, the corresponding discrete evolution mappings (\ref{cL}) have to
be modified for the open
cavity as follows:
\begin{equation}\label{cL-sigma}
\cL_{\sigma}(\rho):=\Tr_{\mathcal{A}_k}(e^{\tau L_{\sigma,k}}(\rho\otimes\rho_{k})) \ .
\end{equation}
The mapping (\ref{cL-sigma}) does not depend on $k \geq 1$, since the atomic states $\{\rho_{k}\}_{k\geq 1}$ are
homogeneous
(\ref{atom-state}).
If $\rho_{\mathcal{C}}:= \rho_{\mathcal{C,\sigma}}(t)\mid_{t=0}$ is initial state of the open cavity, then similar to
(\ref{cavity state at t}) we obtain by (\ref{cL-sigma}) for $\rho_{\mathcal{C,\sigma}}(t)$ at the moment
$t=(n-1)\tau + \nu(t)$:
\begin{equation}\label{cavity state at t-sigma}
\rho_{\mathcal{C,\sigma}}(t)=\Tr_{\mathcal{A}_{n}}[e^{\nu(t) L_{\sigma,n}}(\cL_{\sigma}^{n-1}(\rho_{\mathcal{C}})
\otimes\rho_{n})]  \ .
\end{equation}

By (\ref{cavity state at t-sigma}) and (\ref{cL-sigma}) we obtain for the time-dependent open cavity state
\begin{equation}\label{C-state-t-sigma}
\omega_{\mathcal{C,\sigma}}^{t}(\cdot):= \omega_{\mathcal{S,\sigma}}^{t}(\, \cdot \otimes \idty)=
\Tr_{\mathcal{C}}(\, \cdot \ \rho_{\mathcal{C,\sigma}}(t)) \ ,
\end{equation}
where
$\omega_{\mathcal{S,\sigma}}^{t}(\cdot):= \Tr ( \, \cdot \ T_{t,0}^{\sigma}(\rho_{\mathcal{C}} \otimes
\rho_{\mathcal{A}}))$
by (\ref{Sol-Liouv-Eq-sigma}). We also define
\begin{equation}\label{t=0-and-inf}
\omega_{\mathcal{C,\sigma}}^{0}(\cdot):= \Tr_{\mathcal{C}}(\, \cdot \ \rho_{\mathcal{C}}) \ \ \ {\rm{and}}
\ \ \ \omega_{\mathcal{C,\sigma}}(\cdot):=\lim_{t\rightarrow\infty} \omega_{\mathcal{C,\sigma}}^{t}(\cdot) \ .
\end{equation}

To study the infinite-time limit $\omega_{\mathcal{C,\sigma}}(\cdot)$, we consider the {functional}
\begin{equation}\label{Weyl-func-sigma}
\omega_{\mathcal{C,\sigma}}(W(\zeta))=\lim_{t\rightarrow\infty} \omega_{\mathcal{C,\sigma}}^{t}(W(\zeta)) \ ,
\end{equation}
generated by the Weyl operators on $\cH_{\mathcal{C}}$:
\begin{equation}\label{Weyl}
W(\zeta)=e^{\frac{i}{\sqrt{2}}(\overline{\zeta} b+\zeta b^*)} \ , \ \zeta\in\mathbb{C} \ .
\end{equation}
Notice that convergence (\ref{Weyl-func-sigma}) on the family of the Weyl operators guarantees the weak limit
\cite{BrRo1} of the states $\omega_{\mathcal{C}}^{t}(\cdot)$, when $t\rightarrow\infty$, see Appendices A.2 and A.3.
The following theorem is our first result about the open cavity.
\begin{theorem}\label{cavity-lim-state-sigma}
Let $\sigma_{+} \geq 0$  and $\sigma_{-} - \sigma_{+}  > 0$.  Then for any gauge-invariant
initial cavity state $\rho_{\mathcal{C}}$ and for a homogenous atomic beam with parameter
$p={\rm{Tr}}_{\cH_{\mathcal{A}_n}}(\eta_n \ \rho_{n})$, the limiting cavity state
\begin{equation}\label{lim-cavity-sigma}
\omega_{\mathcal{C},\sigma}(\cdot):= \lim_{t\rightarrow\infty} \omega_{\mathcal{C},\sigma}^{t}(\cdot)
\end{equation}
exists and it does not dependent on $\rho_{\mathcal{C}}$. Here the limit means trace-norm convergence
of the sequence {\rm{(}}\ref{cavity state at t-sigma}{\rm{)}} to a density matrix $\rho_{C,\sigma}:$
\begin{equation}\label{lim-dens-matr}
\lim_{t\rightarrow\infty} \|\rho_{\mathcal{C},\sigma} - \rho_{\mathcal{C},\sigma}(t)\|_1=0,
\end{equation}
where the norm $\|\cdot\|_1$ denotes the trace-norm on the space of trace-class operators
${\mathfrak{C}_{1}(\mathcal{H}_\mathcal{C})}$, see {\rm{(}}\ref{Tr-norm-conv}{\rm{)}}, Appendix A.2. The explicit
form of the
limiting  {functional} {\rm{(}}\ref{Weyl-func-sigma}{\rm{)}} is
\begin{align}\label{lim-W-funct-sigma}
&\omega_{\cC,\sigma}(W(\zeta))=e^{-\frac{|\zeta|^2}{4}\frac{\sigma_-+\sigma_+}{\sigma_--\sigma_+}} \times \\
&\times \prod_{k=0}^{\infty}
\left\{p \exp{\frac{1}{\sqrt{2}}\left(\frac{\lambda}{\mu}(1-e^{-\mu\tau})e^{-k\mu\tau}\overline{\zeta}-
\frac{\lambda}{\overline{\mu}}(1-e^{-\overline{\mu}\tau})e^{-k\overline{\mu}\tau}\zeta\right)}+1-p\right\} \,,
\nonumber
\end{align}
where $\mu: = i \epsilon + (\sigma_- - \sigma_+)/2$.
\end{theorem}

This result is obviously different from the case: $\sigma_- = \sigma_+ =0$, when there is no
regular limiting state, because the number of photons in the ideal cavity cavity unboundedly increases with time,
see Theorem \ref{N of photons}.
\begin{corollary}\label{cavity-lim-state-sigma-corollary}
Theorem \ref{cavity-lim-state-sigma} implies that $\omega_{C,\sigma}(\cdot)$  is a regular, normal
(see Appendix A.5) and, in general, non-gauge-invariant state. One also sees that it is not quasi-free for $0 < p < 1$,
but it obviously does for $p =0$ or $p=1$.
\end{corollary}
To verify these properties notice that by the Araki-Segal theorem (see Theorem \ref{Araki-Segal}), any \textit{regular} state over
CCR($\mathcal{H}_{\mathcal{C}}$) is uniquely defined by the characteristic  {functional}
$\left\{\zeta \mapsto \omega_{C,\sigma}(W(\zeta))\right\}_{\zeta\in\mathbb{C}}$. To check the conditions of this
theorem we express the infinite product in (\ref{lim-W-funct-sigma}) as a uniformly converging (by estimate (\ref{h-k}))
infinite sum of terms that have a generic form:
\begin{align}\label{terms}
e^{-\frac{|\zeta|^2}{4}\frac{\sigma_- + \sigma_+}{\sigma_- - \sigma_+}} \  p^n (1-p)^m \ e^{i r_{k}(\zeta)} \ ,
\end{align}
where
\begin{equation}\label{r-form}
r_{k}(\zeta) :=
2 \ {\rm{Im}}\left([{\lambda}(1-e^{-\mu\tau})e^{-k\mu\tau}\overline{\zeta}]/{{\mu}\sqrt{2}}\right) \ .
\end{equation}
By virtue of (\ref{QFree1}), apart from the normalization, expression (\ref{terms}) is nothing
but the characteristic function of a \textit{non-gauge-invariant} (see (\ref{r-form})) and \textit{quasi-free} state on
CCR($\mathcal{H}_{\mathcal{C}}$), which trivially verifies the Araki-Segal theorem. As a consequence,
(\ref{lim-W-funct-sigma}) defines a \textit{normal} state since it is a convergent infinite convex combination of normal
quasi-free states but, when $0 < p < 1$, the infinite combination of quasi-free states
in (\ref{lim-W-funct-sigma}) is \textit{not} quasi-free.

Let $N_\sigma(t)$ be the expectation value of the photon-number operator $b^*b$ in the open cavity at the time $t$ as:
\begin{equation}\label{N(t)_s}
N_\sigma(t): = \omega_{\mathcal{C},\sigma}^{t}(b^*b)= \Tr_{\mathcal{C}}(b^*b \ \rho_{\mathcal{C},\sigma}(t)) \ .
\end{equation}

For the open cavity the result corresponding to the asymptotic behaviour of the photon number in the cavity takes
the form.
\begin{theorem}\label{number of photons-sigma} Let $\sigma_-  - \sigma_+ > 0$. Then for an arbitrary initial
gauge-invariant cavity state $\rho_{\mathcal{C}}$ such that the initial mean-value of the photon number in the cavity
is bounded {\rm{:}}
\begin{equation}\label{int-cond-photons-sigma}
N_\sigma (0) = \omega_{\mathcal{C,\sigma}}^{t}(b^*b)\mid_{t=0} = {\rm{Tr}}_{{\mathcal{C}}}(b^*b \
\rho_{\mathcal{C}}) < \infty \ ,
\end{equation}
we obtain that the expected number of photos at $t=(n-1)\tau+\nu(t)$ has the form{\rm{:}}
\begin{align}\label{mean-valueBIS}
&N_\sigma(t) = e^{-(\sigma_--\sigma_+) t}N_\sigma(0)\\
&+p\frac{\lambda^2}{|\mu|^2}e^{-(\sigma_--\sigma_+)\tau}(1-e^{\mu\tau})(1-e^{\bar{\mu}\tau})
\frac{1-e^{-(\sigma_--\sigma_+) t}}
{1-e^{-(\sigma_--\sigma_+)\tau}}\nonumber\\
&-p^2\frac{2\lambda^2}{|\mu|^2}\frac{1-e^{-(\sigma_--\sigma_+)t}}
{1-e^{-(\sigma_--\sigma_+)\tau}}(1-e^{-{(\sigma_--\sigma_+)\tau}/{2}} \cos\epsilon\tau) \nonumber\\
&+p^2\frac{2\lambda^2}{|\mu|^2}(1-e^{-(\sigma_{-} - \sigma_{+})t/2}\cos \epsilon t)+\frac{\sigma_+}
{\sigma_--\sigma_+}(1-e^{-(\sigma_--\sigma_+) t})\nonumber.
\end{align}
Here $\mu := (\sigma_--\sigma_+)/2 + i \, \epsilon$.
The limit of the expected number of photons in the cavity is
\begin{align}\label{lim-photons-number-sigma}
&\omega_{\mathcal{C,\sigma}}(b^*b) :=
\lim_{t\rightarrow\infty} \omega_{\mathcal{C},\sigma}^{t}(b^*b) \\
&= p(1-p) \frac{2\lambda^2}{|\mu|^2}
\frac{1 - e^{-(\sigma_--\sigma_+)\tau/2} \cos\epsilon\tau}{1-e^{-(\sigma_--\sigma_+)\tau}}
+ p(2p-1)\frac{\lambda^2}{|\mu|^2} + \frac{\sigma_+}{\sigma_--\sigma_+} \ . \nonumber
\end{align}
\end{theorem}
\begin{remark}\label{instructive cases} Note that the limit (\ref{lim-photons-number-sigma}) satisfies the estimate from
above:
\begin{equation}\label{lim-photons-number-sigma-above}
\omega_{\mathcal{C,\sigma}}(b^*b) \leq \frac{2\lambda^2}{|\mu|^2}\frac{p(1-p)}{1-e^{-(\sigma_--\sigma_+)\tau}}
+ p(2p-1)\frac{\lambda^2}{|\mu|^2} + \frac{\sigma_+}{\sigma_--\sigma_+} \ ,
\end{equation}
as well as the estimate from below:
\begin{equation}\label{lim-photons-number-sigma-below}
\omega_{\mathcal{C,\sigma}}(b^*b) \geq \frac{2\lambda^2}{|\mu|^2}\frac{p^2}{1+e^{(\sigma_--\sigma_+)\tau}}
+ \frac{\sigma_+}{\sigma_--\sigma_+} \ ,
\end{equation}
There are several instructive cases that one can observe as corollaries of the above theorem:
\begin{itemize}
\item Let $p=0$ (or $\lambda=0$), i.e. no pumping by the atomic beam. Then the limiting value
(\ref{lim-photons-number-sigma}),(\ref{lim-photons-number-sigma-below})
is ${\sigma_+}/{(\sigma_--\sigma_+)}$, which coincides with the formula (\ref{state-I}) in Appendix A.3.
If in this case one put $\sigma_- > 0$ and consider $\sigma_+ \rightarrow 0$, then the limit
(\ref{lim-photons-number-sigma})
is equal to {\textit{zero}}. This simply means that in the absence of atomic pumping the leaky cavity relaxes to
the vacuum state empty of photons, or formally to the Gibbs state (\ref{lim-state}) with zero temperature
$\beta_{{\rm{cav}}}= +\infty$.
\item When $\lambda \neq 0$ and $p=1$, still there is no (unlimited) pumping by the atomic beam but only the photon
Rabi oscillations (\ref{number of photons_Hamiltonian}). Then for the open cavity $\sigma_-  >  \sigma_+ \geq 0$ 
the limit (\ref{lim-photons-number-sigma}) and (\ref{lim-photons-number-sigma-above}) give
\begin{equation}\label{lim-photons-number-sigma-p=1}
\omega_{\mathcal{C},\sigma}(b^*b)=  \frac{\lambda^2}{|\mu|^2} + \frac{\sigma_+}{\sigma_--\sigma_+} \ .
\end{equation}
for any leaking $\sigma_- > 0$, including the \textit{limit} of the ideal cavity when $\sigma_- \rightarrow 0$.
Note that in the ideal cavity for the case $p=1$ the mean-value of photons (\ref{number of photons_Hamiltonian})
is bounded,
but oscillating. In other words, the limits: $t\rightarrow \infty$ and $\sigma_- \rightarrow 0$ do not commute.
\item If $0 < p < 1$ and $\sigma_+ =0$, then for the limit of ideal cavity (\ref{lim-photons-number-sigma}) yields
\begin{equation}\label{lim-photons-number-sigma-0<p<1}
\lim_{\sigma_- \rightarrow 0}\omega_{\mathcal{C},\sigma}(b^*b) =
 p(1-p) \, \frac{2\lambda^2}{\epsilon^2} \, \lim_{\sigma_- \rightarrow 0}
 \frac{(1-\cos\epsilon\tau)}{1-e^{- \tau\sigma_-}} \ .
\end{equation}
Hence, for the non-resonant case $\epsilon\tau\neq 2\pi s$, where $s\in\bZ$, this limit is infinite,
i.e. corresponding to the conclusion of the Theorem \ref{N of photons} about unlimited pumping of the ideal cavity.
Indeed, (\ref{mean-valueBIS}) for $\sigma_+ =0$ implies
\begin{equation}\label{lim-photons-number-sigma-0<p<1-bis}
\lim_{\sigma_- \rightarrow 0}\omega_{\mathcal{C},\sigma}^{t=n\tau}(b^*b)|_{\sigma_+ =0} = N(n\tau) \ .
\end{equation}
Here $N(n\tau)$ coincides with (\ref{number of photons_Hamiltonian}), which diverges for the non-resonant case as
$ n \, p(1-p) \, {2\lambda^2}(1-\cos\epsilon\tau)/{\epsilon^2}|_{n \rightarrow \infty}$, cf.
(\ref{lim-photons-number-sigma-0<p<1}).
\item When $\sigma_+ > 0$ and $\sigma_-\rightarrow\sigma_+$, the limit (\ref{lim-photons-number-sigma}) yields
\begin{equation}\label{sigma to zero2}
\lim_{\sigma_- \rightarrow \sigma_+} \omega_{\mathcal{C},\sigma}(b^*b) = + \infty  \ .
\end{equation}
It means that if the leaking and the environmental pumping have the same rate, the limiting state corresponds to the
infinite temperature state, see (\ref{lim-W-funct-sigma}) and (\ref{state-II}), (\ref{lim-state}).
On the Weyl operators this state is
given by the Kronecker delta-functional
$$ \omega_{\cC,\sigma}(W(\zeta)) = \left\{
     \begin{array}{lr}
       1 & \text{if } \zeta=0,\\
       0 & \text{if } \zeta\neq 0.
     \end{array}
\right.$$
Since the Kronecker characteristic functional is not continuous, the corresponding state is not regular.
Consequently, the Araki-Segal Theorem \ref{Araki-Segal} is not applicable.
\item  Note that for $\sigma_+ > 0$ and $\sigma_-\rightarrow\sigma_+$ by (\ref{mean-valueBIS}) one gets for the expected
number of photos at $t=n\tau$
\begin{equation}\label{lim-photons-number-sigma=sigma}
\lim_{\sigma_- \rightarrow \sigma_+}\omega_{\mathcal{C},\sigma}^{n\tau}(b^*b) = N(n\tau) + n \, \tau\sigma_+ \ .
\end{equation}
Hence, in this case the pumping by the random non-resonant atomic beam (\ref{number of photons_Hamiltonian})
and by the environmental
pumping due to $\sigma_+ > 0$, give the same \textit{linear} rate for increasing of the mean number of photons in
the cavity.
Consequently, for $t=n\tau \rightarrow\infty$ the infinite photon number cavity state coincides with the
infinite-temperature state
that we discussed above, cf. (\ref{lim-state}).
\end{itemize}
\end{remark}
This concludes the description of our main results. The rest of the paper is organized as follows.
In Section \ref{HD} we use the specific properties of our model to {diagonalise} it, which is the key for the
further analysis. Then we give the proof of Theorem \ref{N of photons} for the Hamiltonian dynamics of the ideal cavity.

In Section \ref{IQD} we present our results for the case of the leaking cavity described by the Kossakowski-Lindblad
irreversible quantum dynamics, i.e., Theorems \ref{cavity-lim-state-sigma} and \ref{number of photons-sigma}.

The results concerning Energy-Entropy relations are presented in Section \ref{EEP}. There we calculate the energy
variation (Theorem \ref{ThDEn-1,n1} and Theorem \ref{Energy-sigma-variation}) and
obtain a formula for the entropy production (formula (\ref{2nd-Law})) for the ideal cavity.

The Section \ref{CR} is reserved for comments, remarks and open problems.

For the reader convenience we collect in an Appendix certain results and definitions necessary for the main text.

\section{Hamiltonian Dynamics: The Ideal Cavity}\label{HD}

Here we consider the case of the ideal cavity, i.e. $\sigma_-=\sigma_+=0$. In this case, the discrete evolution
map $\cL_{\sigma=0} = \cL$ is given by (\ref{cL}). Recall that $p=\Tr(\eta_k\rho_{k})$ for the
homogeneous atomic states $\{\rho_{k}\}_{k\geq 1}$ in the beam.

Our first result concerns the {expectation} of the photon-number operator $\hat{N}=b^*b$ in the cavity
(\ref{N(t)}). For $t=n\tau$ this expectation involves the calculation of $\cL^{n}(\rho)$ (\ref{mean-photon-number-n}).
Instead, we use the $n$-th power of the \textit{adjoint} operator $\cL^{*}$ defined by duality with respect to the
cavity state $\omega_{\mathcal{C}}(\cdot) = \Tr_{\mathcal{C}}(\, \cdot \ \rho_{\mathcal{C}})$, see
(\ref{mean-photon-number-n})
and Remark \ref{dual-operator}.

Let $\widehat{S}_{k}$ be $*$-isomorphism ({unitary shift}) on the algebra
$\mathfrak{A}(\cH_{\mathcal{C}})\otimes \mathfrak{A}(\cH_{\mathcal{A}_{k}})$ defined by
\begin{equation}\label{Vn}
\widehat{S}_{k}(\cdot):= e^{iV_k}(\cdot) \, e^{-iV_k}  \ , \ V_k:={\lambda}(b^*-b)\otimes \eta_k /{i\epsilon} \ .
\end{equation}
Since $\eta_{k}^{2} = \eta_k$ and $\widehat{S}_{k}(\idty\otimes \eta_k)=\idty\otimes\eta_k$, whereas
\begin{equation}\label{Vn-shift}
\widehat{S}_{k}(b\otimes\idty)= b\otimes\idty-\idty\otimes\frac{\lambda}{\epsilon}\eta_{k} \ , \
\widehat{S}_{k}(b^*\otimes\idty) = b^*\otimes\idty-\idty\otimes\frac{\lb}{\e}\eta_k \ ,
\end{equation}
the transformation (\ref{Vn}) of the Hamiltonian (\ref{Ham-n}) gives
\begin{equation}\label{Diagonalized Hamiltonian}
\widehat{H}_{k}: = \widehat{S}_{k}(H_k)=\epsilon \ b^*b \otimes \idty + \idty \otimes (E-\frac{\lambda^2}{\epsilon}) \,
\eta_k + \sum_{s\geq1: s \neq k}\idty\otimes E \, \eta_s\ .
\end{equation}
Notice that the dynamics generated by $\widehat{H}_{k}$ (or by ${H}_k$) leaves the atomic operator $\eta_k =
(\sigma^{z} + I)/2$
invariant
\begin{equation*}
e^{i\tau \widehat{H}_{k}}(\idty\otimes\eta_k)e^{-i\tau \widehat{H}_{k}}=\idty\otimes\eta_k \ .
\end{equation*}
Similarly to the unitary shift (\ref{Vn}), we define on $\mathfrak{A}(\cH_{\mathcal{C}})$ the $*$-isomorphism:
\begin{equation}\label{V}
S(\cdot):= e^{iV}(\cdot) \, e^{-iV}  \ \ , \ \  V:= {\lambda}(b^*-b)/{i\epsilon} \ .
\end{equation}
\begin{lemma} \label{cL(rho)} For any state $\rho$ on $\mathfrak{A}(\cH_{\mathcal{C}})$ the one-step evolution mapping
$\cL$ {\rm{(}}\ref{cL}{\rm{)}} has the form:
\begin{equation}\label{cL(rho)_p}
\cL(\rho)= p \ S^{-1}(e^{-i\tau\epsilon b^*b} S(\rho) e^{i\tau\epsilon b^*b})
+(1-p)\ e^{-i\tau\epsilon b^*b}\rho \, e^{i\tau\epsilon b^*b} \ .
\end{equation}
Here $p={\rm{Tr}}_{\cH_{\mathcal{A}_k}} (\eta_k \ \rho_k)$ is the probability to find the $k$-th atom in the excited
 state with energy $E$.
\end{lemma}
\begin{proof}
Using the shift transformation (\ref{Vn}) and the cyclicity of trace, one can re-write (\ref{cL}) as
\begin{equation}\label{cL(rho)1}
\cL(\rho)=\Tr_{{\mathcal{A}_k}}[\widehat{S}_{k}^{-1}(e^{-i\tau \widehat{H}_{k}} \widehat{S}_{k}(\rho\otimes\rho_{k})\,
e^{i\tau \widehat{H}_{k}})] \ .
\end{equation}
Since the operator $\eta_k = (\sigma^{z} + I)/2$ is idempotent ($\eta_{k}^{2} = \eta_k$) and since it commutes
with $\rho_{k}$ (\ref{atom-state}), the combination of expansions of (\ref{Vn}), (\ref{V}) with definitions of $V_k$ and
$V$ gives
\begin{align}\label{relationSSn}
&\widehat{S}_{k}(\rho\otimes\rho_{k}) = S(\rho) \otimes {\eta_{k}}\rho_{k}+\rho\otimes(I - \eta_k)\rho_{k}  \ , \\
&\widehat{S}_{k}^{-1}(\rho\otimes\rho_{k}) = S^{-1}(\rho) \otimes {\eta_{k}}\rho_{k}+\rho\otimes(I - \eta_k)\rho_{k}  \ .
\label{relationSSnBIS}
\end{align}
Therefore, plugging (\ref{relationSSn}) into (\ref{cL(rho)1}) we obtain
\begin{align*}
\cL(\rho)=&\Tr_{{\mathcal{A}_k}}[\widehat{S}_{k}^{-1}(e^{-i\tau \widehat{H}_{k}}(S(\rho) \otimes {\eta_{k}}\rho_{k})
e^{i\tau \widehat{H}_{k}})]\\
&+\Tr_{{\mathcal{A}_k}}[\widehat{S}_{k}^{-1}(e^{-i\tau \widehat{H}_{k}}(\rho\otimes(I- \eta_k)\rho_{k})
e^{i\tau \widehat{H}_{k}})]
\ .
\end{align*}
Now, the diagonal form (\ref{Diagonalized Hamiltonian}) of the Hamiltonian $\widehat{H}_{k}$ and (\ref{relationSSnBIS})
for $\widehat{S}_{k}^{-1}$, imply:
\begin{align*}
\cL(\rho)&=\Tr_{{\mathcal{A}_k}}[S^{-1}(e^{-i\tau\epsilon b^*b} S(\rho) e^{i\tau\epsilon b^*b})\otimes
{\eta_{k}}\rho_{k}]\\
&+\Tr_{{\mathcal{A}_k}}[e^{-i\tau\epsilon b^*b} S(\rho) e^{i\tau\epsilon b^*b}\otimes (I - \eta_k){\eta_{k}}\rho_{k}]\\
&+\Tr_{{\mathcal{A}_k}}[S^{-1}(e^{-i\tau\epsilon b^*b} \rho \, e^{i\tau\epsilon b^*b})\otimes \eta_k(I -
\eta_k)\rho_{k}]\\
&+\Tr_{{\mathcal{A}_k}}[e^{-i\tau\epsilon b^*b}\rho e^{i\tau\epsilon b^*b}\otimes(I - \eta_k)\rho_{k}] \ .
\end{align*}
Since $(I - \eta_k)\eta_k=0$ and $p = \Tr_{{\mathcal{A}_k}}[\idty \otimes {\eta_{k}}\rho_{k}]$, one gets (\ref{cL(rho)_p}).
\end{proof}
\begin{remark} \label{dual-operator}
{\rm{(a)}} To calculate the expectation of the photon-number operator
$\hat{N} = b^*b$ at $t=n\tau$ using (\ref{mean-photon-number-n}), we would need to find the action of the $n$-th power
$\cL^n(\rho)$ of the operator (\ref{cL(rho)1}). \\
We show that in fact it is  easier to calculate this mean value using the $n$-th power of the adjoint (or \textit{dual})
mapping $\cL^{*}$. For any bounded operator $A \in \cB(\cH_{\mathcal{C}})$ and $\rho \in \mathfrak{C}_{1}
(\cH_{\mathcal{C}})$, it is defined by relation
\begin{equation}\label{cL-dual-0}
\Tr_{{\mathcal{C}}}(\cL^{*}(A) \rho): = \Tr_{{\mathcal{C}}}( A \ \cL(\rho)) \ ,
\end{equation}
see Appendix A.2 for discussion and details.\\
{\rm{(b)}} Using invariance of the atomics states (\ref{commut-atoms}) one can obtain explicit expression for the
one-step dual mapping $\cL^{*}$. Indeed, by (\ref{commut-atoms}) and by (\ref{cL})
\begin{align}\label{cL-dual-1}
\Tr_{{\mathcal{C}}}(A \ \cL(\rho))=& \Tr_{\cH_{\mathcal{C}}\otimes\cH_{\mathcal{A}_k}}\{ (A\otimes \idty) \
e^{- i\tau H_k}(\rho\otimes \idty) (\idty \otimes \rho_k)e^{i\tau H_k}\} \\
=& \Tr_{\cH_{\mathcal{C}}\otimes\cH_{\mathcal{A}_k}}\{(\idty \otimes \rho_k) e^{-i\tau H_k} (A\otimes \idty) e^{i\tau H_k}
(\rho\otimes \idty)\} \nonumber \\
=& \Tr_{\cH_{\mathcal{C}}\otimes\cH_{\mathcal{A}_k}}\{e^{i\tau H_k}(A \otimes \rho_k) \ e^{-i\tau H_k}(\rho\otimes \idty)\}
\nonumber \ ,
\end{align}
where we used cyclicity of the full trace $\Tr_{\cH_{\mathcal{C}}\otimes\cH_{\mathcal{A}_k}}$.
Hence, (\ref{cL-dual-0}) together with (\ref{cL-dual-1}) yield expression for the one-step mapping
\begin{equation}\label{cL-dual-2}
\cL^{*}(A)=\Tr_{{\mathcal{A}_k}}\{e^{i\tau H_k}(A\otimes \rho_k)e^{-i\tau H_k}\} \ ,
\end{equation}
which according to (\ref{cL}) is independent of $k\geq 1$.\\
{\rm{(c)}} Let $\rho \in \mathfrak{C}_{1}(\cH_{\mathcal{C}})$ be density matrix such that $\Tr_{{\mathcal{C}}}
( \mathcal{P}(b,b^*) \, \cL(\rho)) < \infty$
for any polynomial $\mathcal{P}(b,b^*)$. Then one can extend definition (\ref{cL-dual-0}) of $\cL^{*}$ to the class of
unbounded observables, which contains $\mathcal{P}(b,b^*)$ and in particular the photon-number operator $b^*b$.
The advantage to use the adjoint mapping $\cL^{*}$ is that its consecutive applications do not
increase the degree of polynomials in variables $b$ and $b^*$.
\end{remark}
Now, applying to (\ref{cL-dual-2}) the same line of reasoning as in the proof of Lemma \ref{cL(rho)}, we find explicit
expression for the adjoint one-step mapping
\begin{equation}\label{dual cL}
\cL^{*}(A)=  \ p \ S^{-1}(e^{i\tau\epsilon b^*b} S(A) e^{-i\tau\epsilon b^*b})
+  \ (1-p)e^{i\tau\epsilon b^*b}A e^{-i\tau\epsilon b^*b} \ .
\end{equation}
Note that alternatively one can obtain (\ref{dual cL}) using (\ref{cL(rho)_p}) and definition (\ref{cL-dual-0}).
\begin{lemma}\label{cL-T(B)}
For $A= {b^*b}$ and for the adjoint operator $\cL^{*}$ defined by {\rm{(}}\ref{dual cL}{\rm{)}} we obtain:
\begin{align}\label{L^n(b^*b)}
(\cL^{*})^n(b^*b)=& \ b^*b+p \ \frac{\lambda}{\epsilon}[(1-e^{ni\epsilon\tau})\ b^*+(1-e^{-ni\epsilon\tau})\ b] \\
+& \ n \, p (1-p)\ \frac{2\lambda^2}{\epsilon^2} (1-\cos\epsilon\tau)+
p^2 \ \frac{2\lambda^2}{\epsilon^2}(1-\cos n\epsilon\tau)  \nonumber \ .
\end{align}
\end{lemma}
\begin{proof}
Since $*$-isomorphism (\ref{V}) is a shift transformation, one gets
\begin{equation*}
S(b^* b)=(b^*-{\lambda}/{\epsilon})(b-{\lambda}/{\epsilon}) \ .
\end{equation*}
The CCR relations for $b^*$ and $b$ yield:
\begin{align*}
e^{i\tau\epsilon b^*b} b^* e^{-i\tau\epsilon b^*b}=e^{i\tau\epsilon}b^* \ , \
e^{i\tau\epsilon b^*b} b e^{-i\tau\epsilon b^*b}=e^{-i\tau\epsilon}b \ .
\end{align*}
Consequently,
\begin{align}\label{evolution-rho_shift}
&S^{-1}(e^{i\tau\epsilon b^*b} S(b^* b) e^{-i\tau\epsilon b^*b})
=S^{-1}((e^{i\tau\epsilon}b^*-{\lambda}/{\epsilon})(e^{-i\tau\epsilon}b-{\lambda}/{\epsilon}))\\
&=(e^{i\tau\epsilon}b^*-(1-e^{i\tau\epsilon}){\lambda}/{\epsilon})
(e^{-i\tau\epsilon}b-(1-e^{-i\tau\epsilon}){\lambda}/{\epsilon}).\nonumber
\end{align}
Hence, (\ref{evolution-rho_shift}) implies for (\ref{dual cL})
\begin{align}\label{cL(b^*b)}
\cL^{*}(b^*b)
=& \, p \, (e^{i\epsilon\tau}b^*-(1-e^{i\epsilon\tau}){\lambda}/{\epsilon})(e^{-i\epsilon\tau}b-(1-e^{-i\epsilon\tau})
{\lambda}/{\epsilon})  \\
+& \, (1-p) \, b^*b \nonumber \\
=& \, b^*b + p\frac{\lambda}{\epsilon}(1-e^{i\epsilon\tau})b^* + p \frac{\lambda}{\epsilon}(1-e^{-i\epsilon\tau})b +
p \frac{2\lambda^2}{\epsilon^2}(1-\cos\epsilon\tau) \ , \nonumber
\end{align}
which coincides with expression (\ref{L^n(b^*b)}) for $n=1$. If we insert in (\ref{dual cL}) $A=b^*$, and then $A=b$,
we obtain correspondingly:
\begin{align}\label{cL(b^*)}
&\cL^{*}(b^*)=e^{i\epsilon\tau}b^*-p(1-e^{i\epsilon\tau}){\lambda}/{\epsilon} \ , \\
&\cL^{*}(b)=e^{-i\epsilon\tau}b-p(1-e^{-i\epsilon\tau}){\lambda}/{\epsilon} \ . \label{cL(b)}
\end{align}
Now one can use induction. Since $(\cL^{*})^n(b^*b)= \cL^{*}((\cL^{*})^{n-1}(b^*b))$, we can apply
(\ref{cL(b^*b)})-(\ref{cL(b)}) to (\ref{L^n(b^*b)}) for $n-1$ to check the formula (\ref{L^n(b^*b)}) for the
$n$-th power.
\end{proof}
\begin{remark} As we indicated in Remark \ref{G-Inv-breaking} the evolution $\cL$ \textit{breaks} the
gauge invariance of the initial cavity state $\rho_{\mathcal{C}}$. For $n=1$ the gauge breaking parameters $r_{n=1}$ and
$\phi_{n=1}$ are defined by (\ref{G-InvBreak2}), or by (\ref{cL-dual-0}), (\ref{cL(b)}):
\begin{equation}\label{G-InvBreak3}
\omega_{\mathcal{C}}^{\tau}(b) = \Tr_{{\mathcal{C}}}(\cL^{*}(b) \rho_{\mathcal{C}})=
p \frac{\lambda}{\epsilon} (e^{-i \epsilon \tau} -1) = r_{1}\, e^{i \, \phi_{1}}\ .
\end{equation}
Formulae (\ref{cL(b^*)}),(\ref{cL(b)}) allow to calculate $r_{n}$ and $\phi_{n}$ for any $n\geq 1$.
Iterating them one obtains
\begin{align}\label{cL(b^*-n)}
&(\cL^{*})^n(b^*)=e^{i n\epsilon\tau}b^*-p(1-e^{i n \epsilon\tau}){\lambda}/{\epsilon} \ , \\
&(\cL^{*})^n(b)=e^{-i n\epsilon\tau}b-p(1-e^{-i n \epsilon\tau}){\lambda}/{\epsilon} \ . \label{cL(b-n)}
\end{align}
Therefore, the gauge breaking parameters $r_{n}$ and $\phi_{n}$ are defined by equation:
\begin{equation}\label{G-InvBreak-n}
\omega_{\mathcal{C}}^{n\tau}(b) = \Tr_{{\mathcal{C}}}((\cL^{*})^n(b) \rho_{\mathcal{C}})=
p \frac{\lambda}{\epsilon} (e^{-i n \epsilon \tau} -1) =: r_{n}\, e^{i \, \phi_{n}}\ .
\end{equation}
\end{remark}
\begin{proof}(of Theorem \ref{N of photons}) To find the mean-value of the number of photons in the cavity at the time
$t=n\tau$, we calculate the expectation of (\ref{L^n(b^*b)}) in the initial cavity state
$\omega_{\mathcal{C}}(\cdot) = \Tr_{{\mathcal{C}}}(\, \cdot \, \rho_{\mathcal{C}})$, (\ref{cL-dual-0}).
Since $\rho_{\mathcal{C}}$ is supposed to be \textit{gauge-invariant}, this yields
\begin{align}\label{number of photons_HamiltonianBIS}
N(t)=&\Tr_{{\mathcal{C}}}(b^*b \ \cL^n(\rho_{\mathcal{C}}))=\Tr_{{\mathcal{C}}}((\cL^{*})^n(b^*b)
\, \rho_{\mathcal{C}})\\
=&N(0)+ n \, p (1-p)\ \frac{2\lambda^2}{\epsilon^2} \ (1-\cos\epsilon\tau) +
p^2\frac{2\lambda^2}{\epsilon^2}(1-\cos n\epsilon\tau) \ , \nonumber
\end{align}
which coincides with (\ref{number of photons_Hamiltonian}) for $t=n\tau$.
\end{proof}
\begin{corollary}\label{N-non-g-inv} If the initial cavity state $\rho_{\mathcal{C},r,\phi}$ is not gauge-invariant,
then (\ref{number of photons not gauge}) and (\ref{number of photons_HamiltonianBIS}) give for $t=n\tau$
\begin{align}
N(t)=&N(0)+ n \, p (1-p)\ \frac{2\lambda^2}{\epsilon^2} \ (1-\cos\epsilon\tau)  +
p^2\frac{2\lambda^2}{\epsilon^2}(1-\cos n\epsilon\tau)  \nonumber \\
 + &p \, \frac{2 \lambda r}{\epsilon}\, [\cos\phi -\cos(n\epsilon\tau - \phi)] \ .
\label{number of photons_Hamiltonian-non-g-inv}
\end{align}
\end{corollary}
\begin{remark} \label{0<p<1}
Theorem \ref{N of photons} implies that the pumping effect is non-trivial (i.e. $N(t)$ is
increasing to infinity) only for the beam of randomly exited atoms, i.e. $0<p<1$. There is no \textit{unlimited}
pumping of cavity by the non-exited, $p=0$, or totaly exited, $p=1$, atomic beams.
Note that by (\ref{Ham-Model}) the mean value of the photon cavity energy is defined by $\mathcal{E}_{\mathcal{C}}(t):=
\epsilon N(t)$. Hence, its unlimited increasing is due to the kicking by randomly excited chain of rigid atoms,
which are \textit{pushed} through the cavity.
\end{remark}
\section{Irreversible Quantum Dynamics: The Open Cavity }\label{IQD}
\noindent
In a real physical cavity it is always possible for photons to leak \textit{out} of the cavity as well as to
diffuse \textit{in} from the environment. This realistic situation is modeled by what is known as an
\textit{open} cavity. In the present section we include these effects, via the Kossakowski-Lindblad extension
(\ref{Generator}) of  the Hamiltonian dynamics (\ref{Ham-Model}) with generator (\ref{L-Gen-t}). Below we consider 
the case of the open cavity with \textit{dominating} photon leaking: $\sigma_{-} > \sigma_{+} \geq 0$, although our 
results allow also to study certain limiting cases, see Remark \ref{instructive cases}.

Similar to the ideal cavity, see (\ref{cL(rho)_p}) in Lemma \ref{cL(rho)}, we first find a new expression for the
one-step evolution mapping in the case of the open cavity (\ref{cL-sigma}).
\begin{lemma}\label{L-sigma}
For any state $\rho \in \mathfrak{C}_{1}(\cH_{\mathcal{C}})$ the one-step evolution mapping {\rm{(}}\ref{cL-sigma}{\rm{)}}
for the open cavity has the form:
\begin{equation}\label{cL(rho)_p-sigma}
\cL_{\sigma}(\rho)=p \, S^{-1}(e^{\tau L_{\lambda,\sigma}}(S (\rho))) +(1-p)\, e^{\tau L_{0,\sigma}}(\rho).
\end{equation}
Here $S$ is defined by {\rm{(}}\ref{V}{\rm{)}} and $L_{\lambda,\sigma}$ acts on $\mathfrak{A}(\cH_{\mathcal{C}})$
as follows
\begin{align}\label{L^C_lambda}
L_{\lambda,\sigma}(\rho):=& -i[\epsilon b^*b, \rho] \\
&+ \sigma_- (b-{\lambda}/{\epsilon})\rho
(b^*-{\lambda}/{\epsilon})-\frac{\sigma_-
}{2}\{(b-{\lambda}/{\epsilon})(b^*-{\lambda}/{\epsilon}), \rho\}\nonumber \\
&+\sigma_+ (b^*-{\lambda}/{\epsilon})\rho (b-{\lambda}/{\epsilon})-\frac{\sigma_+
}{2}\{(b^*-{\lambda}/{\epsilon})(b-{\lambda}/{\epsilon}), \rho\}\nonumber,
\end{align}
with $L_{0,\sigma}:=L_{\lambda=0,\sigma}$.
\end{lemma}
\begin{proof} Using $*$-isomorphisms (\ref{Vn}),(\ref{V}) and equations (\ref{relationSSn}),(\ref{relationSSnBIS})
we define instead of generator (\ref{Generator-KL}) the following operator:
\begin{align}\label{tilde_L}
&\widehat{L}_{\sigma,k}(\rho\otimes\rho_{k}):= \widehat{S}_{k}(L_{\sigma,k}(\widehat{S}^{-1}_{k}(\rho\otimes\rho_{k}))
\\
&=-i[\epsilon b^*b\otimes\idty +\idty\otimes(E-{\lambda^2}/{\epsilon})\eta_k, \rho\otimes(\eta_k \rho_{k} +
(I - \eta_k) \rho_{k})]\nonumber\\
&+ \left(\sigma_- (b-{\lambda}/{\epsilon})\rho
(b^*-{\lambda}/{\epsilon}) -\frac{\sigma_- }{2}\{(b-{\lambda}/{\epsilon})(b^*-{\lambda}/{\epsilon}),
\rho\}\right)\otimes \eta_k\rho_{k}\nonumber\\
&+ \left(\sigma_+ (b^*-{\lambda}/{\epsilon})\rho(b-{\lambda}/{\epsilon})-\frac{\sigma_+}{2}\{(b^*-{\lambda}/{\epsilon})
(b-{\lambda}/{\epsilon}), \rho\}\right)\otimes \eta_k\rho_{k}\nonumber\\
&+\left(\sigma_- b\rho b^*-\frac{\sigma_-}{2}\{bb^*, \rho\}\right)\otimes(I-\eta_k)\rho_{k}
+\left(\sigma_+ b^*\rho b-\frac{\sigma_-}{2}
\{b^*b, \rho\}\right)\otimes(I-\eta_k)\rho_{k} \nonumber \\
&= L_{\lambda,\sigma}(\rho)\otimes \eta_k\rho_{k}+L_{0,\sigma}(\rho)\otimes(I -\eta_k)\rho_{k} \nonumber \ .
\end{align}
With help of (\ref{tilde_L}) and (\ref{relationSSn}),(\ref{relationSSnBIS}) we obtain the representation
\begin{align}
&e^{\tau L_{\sigma,k}}(\rho\otimes\rho_{k})=\widehat{S}^{-1}_{k}(e^{\tau\widehat{L}_{\sigma,k}}
(\widehat{S}_{k}(\rho\otimes\rho_{k})))
\label{e_L_tilde}\\
&=\widehat{S}^{-1}_{k}(e^{\tau\widehat{L}_{\sigma,k}}(S(\rho)
\otimes {\eta_{k}}\rho_{k}+\rho\otimes(I-\eta_k)\rho_{k}))\nonumber\\
&=\widehat{S}^{-1}_{k}(e^{\tau L_{\lambda,\sigma}}(S(\rho)
\otimes {\eta_{k}}\rho_{k}))
+\widehat{S}^{-1}_{k}(e^{\tau L_{0,\sigma}}(\rho)\otimes(I-\eta_k) \rho_{k})\nonumber\\
&=S^{-1}(e^{\tau L_{\lambda,\sigma}} (S(\rho)))\otimes {\eta_{k}}\rho_{k}
+e^{\tau L_{0,\sigma}}(\rho)\otimes(I-\eta_k)\rho_{k}\nonumber  \ .
\end{align}
Therefore, the one-step mapping (\ref{cL-sigma}) takes the form
\begin{align*}
&\cL_{\sigma}(\rho)=\Tr_{{\mathcal{A}_k}} (e^{\tau L_{\sigma,k}}(\rho\otimes\rho_{k}))\\
&=p \, S^{-1}(e^{\tau L_{\lambda,\sigma}}(S(\rho))) +(1-p)\, e^{\tau L_{0,\sigma}}(\rho) \ ,
\end{align*}
which is claimed by Lemma in (\ref{cL(rho)_p-sigma}).
\end{proof}
\begin{corollary}\label{dual-operator-sigma1}
To define the adjoint mapping $\cL_{\sigma}^{*}$ we use duality relation (\ref{cL-dual-0}) for (\ref{cL(rho)_p-sigma})
and any bounded operator $A\in\cB(\cH_{\mathcal{C}})$:
\begin{equation}\label{cL-dual-sigma}
\Tr_{{\mathcal{C}}}(\cL^{*}_{\sigma}(A) \rho): = \Tr_{{\mathcal{C}}}( A \ \cL_{\sigma}(\rho)) \ .
\end{equation}
Then by definitions (\ref{V}), (\ref{cL-dual-sigma}) and by explicit form (\ref{cL(rho)_p-sigma}) of the operator
$\cL_{\sigma}(\rho)$ one gets that
\begin{equation}\label{cL-dual-sigma1}
\cL^{*}_{\sigma}(A) = p \, S^{-1}((e^{\tau L_{\lambda,\sigma}})^{*}(S(A))) +(1-p)\, (e^{\tau L_{0,\sigma}})^{*}(A) \ .
\end{equation}
If for adjoint operators $(e^{\tau L_{\lambda,\sigma}})^{*}$ and $(e^{\tau L_{0,\sigma}})^{*}$ we define the
corresponding adjoint generators by
\begin{equation}\label{cL-dual-sigma2}
e^{\tau L_{\lambda,\sigma}^{*}}:=(e^{\tau L_{\lambda,\sigma}})^{*} \ \ {\rm{and}} \ \
e^{\tau L_{0,\sigma}^{*}}:=(e^{\tau L_{0,\sigma}})^{*} \ ,
\end{equation}
then (\ref{cL(rho)_p-sigma}),(\ref{L^C_lambda}) and (\ref{cL-dual-sigma})-(\ref{cL-dual-sigma2}) allow to find them
explicitly:
\begin{align}
L^*_{\lambda,\sigma}(A)&=i[\epsilon b^*b, A]+\frac{\sigma_-}{2}(b^*-{\lambda}/{\epsilon})[A,b]+
\frac{\sigma_-}{2}[b^*, A](b-{\lambda}/{\epsilon}) \label{L_lambda(B)}\\
&+\frac{\sigma_+}{2}(b-{\lambda}/{\epsilon})[A,b^*]+\frac{\sigma_+}{2}[b, A](b^*-{\lambda}/{\epsilon}) \ ,\nonumber\\
L^*_{0,\sigma}(A)&=i[\epsilon b^*b, A]+\frac{\sigma_-}{2}b^*[A,b]+\frac{\sigma_-}{2}[b^*, A]b \label{L_0_lambda(B)}\\
&+\frac{\sigma_+}{2}b[A,b^*]+\frac{\sigma_+}{2}[b, A]b^* \ , \nonumber
\end{align}
by straightforward calculations.
\end{corollary}
\begin{remark}\label{dual-operator-sigma}
As we indicated in Remark \ref{dual-operator} {\rm{(c)}}, one can extend the adjoint one-step mapping $\cL_{\sigma}^{*}$
to the algebra generated by polynomials in the annihilation and creation operators.
\end{remark}
\begin{remark}\label{dual-operator-sigma-bis}
Duality allows us to define by the relation
\begin{equation}\label{dual-S}
\Tr (\widehat{S}_{k}^\ast (\widehat{A}) \ \rho\otimes\rho_{k}):=
\Tr (\widehat{A} \ \widehat{S}_{k}(\rho\otimes\rho_{k})) \ ,
\end{equation}
the operator $\widehat{S}_{k}^\ast$ with domain consisting of operators
$\widehat{A}= A \otimes a \in \mathfrak{A}(\cH_{\mathcal{C}})\otimes \mathfrak{A}(\cH_{\mathcal{A}_{k}})$.
By definition (\ref{Vn}) and by cyclicity of the trace, one obtains from (\ref{dual-S}) that
$\widehat{S}_{k}^\ast = \widehat{S}_{k}^{-1}$, cf (\ref{relationSSn}), (\ref{relationSSnBIS}).
Then (\ref{tilde_L}) and (\ref{dual-S}) yield for the adjoint of the operator defined in (\ref{tilde_L})
the following expression
\begin{equation}\label{hat-L-dual}
\widehat{L}^*_{\sigma,k}(\cdot) = \widehat{S}_{k}({L}^*_{\sigma,k}(\widehat{S}_{k}^{-1}(\cdot))) \ .
\end{equation}
Here ${L}^*_{\sigma,k}$ is defined in (\ref{S-state-evol-adj-sigma}). It can be explicitly calculated on
operators $A \otimes a$ with help of (\ref{Generator-KL}). Since (\ref{tilde_L}) and (\ref{hat-L-dual}) imply
\begin{equation}\label{hat-L-dual-BIS}
\widehat{L}^*_{\sigma,k}(A \otimes a)=
L_{\lambda,\sigma}^*(A)\otimes \eta_k a + L_{0,\sigma}^*(A)\otimes(I -\eta_k) a  \ ,
\end{equation}
this equation gives an alternative way to establish (\ref{L_lambda(B)}) and (\ref{L_0_lambda(B)}).
Similar to (\ref{e_L_tilde}), the representation (\ref{hat-L-dual-BIS}) indicates that for each $k\geq 1$
the mapping $\{e^{\tau \widehat{L}^*_{\sigma,k}}\}$ is reducible by two sub-algebras of
$\mathfrak{A}(\cH_{\mathcal{C}})\otimes \mathfrak{A}(\cH_{\mathcal{A}_{k}})$:
\begin{align}\label{reduc1}
& e^{\tau \widehat{L}^*_{\sigma,k}}: A \otimes \eta_k a \mapsto e^{\tau L_{\lambda,\sigma}^*}(A)\otimes \eta_k a \ ,\\
& e^{\tau \widehat{L}^*_{\sigma,k}}: A \otimes (I -\eta_k) a \mapsto
e^{\tau L_{0,\sigma}^*}(A)\otimes (I -\eta_k) a \ . \label{reduc2}
\end{align}
\end{remark}
\begin{proof}(of Theorem \ref{cavity-lim-state-sigma})
Using the Baker-Campbell-Hausdorff formula one can rewrite the Weyl operator in the form
\begin{equation}\label{B-C-H-formula}
W(\zeta)=e^{{i}(\overline{\zeta} b+\zeta b^*)/{\sqrt{2}}}=e^{{i}\zeta b^*/{\sqrt{2}}}
e^{{i}\overline{\zeta} b/{\sqrt{2}}}e^{-{|\zeta|^2}/{4}} \ .
\end{equation}
{From} the CCR we find that for any $\beta, \gamma\in\bC$
\begin{equation}\label{commutation_Weyl}
e^{\beta b^*}b=(b-\beta)e^{\beta b^*} \text{ and \ }e^{\gamma b}b^*=(b^*+\gamma)e^{\gamma b}.
\end{equation}
By (\ref{B-C-H-formula}) and (\ref{commutation_Weyl}) the action of $L_{\lambda,\sigma}^{*}$ is
\begin{align}
&L_{\lambda,\sigma}^{*}(W(\zeta))=-\frac{i}{\sqrt{2}}\mu \, \overline{\zeta} \, W(\zeta)\,
b-\frac{i}{\sqrt{2}}\overline{\mu} \, \zeta \, b^* W(\zeta)-\frac{\sigma_+}{2}|\zeta|^2W(\zeta)\label{LW}\\
&+\frac{i}{\sqrt{2}}
\frac{\lambda(\sigma_--\sigma_+)}{2\epsilon}(\zeta+\overline{\zeta})
W(\zeta)\nonumber \ .
\end{align}
Recall that $\mu = i\epsilon + (\sigma_--\sigma_+ )/{2}$ . Therefore, the mapping generated by (\ref{L_lambda(B)})
(and correspondingly by (\ref{L_0_lambda(B)})):
\begin{equation}\label{Dyn-C}
\gamma_{\lambda,\tau}:=e^{\tau L_{\lambda,\sigma}^{*}}
\end{equation}
is \textit{quasi-free} \cite{DVV1,DVV2}. In particular, there exist functions $\zeta(\tau)$ and
$\Omega_{\lambda,\tau}(\zeta)$ such that
\begin{equation}\label{gammaW}
\gamma_{\lambda,\tau}(W(\zeta))=e^{-\Omega_{\lambda,\tau}(\zeta)}W(\zeta(\tau)).
\end{equation}
See Appendix A.4 for definitions and details.

To check this claim we calculate explicitly  $\zeta(\tau)$ and $\Omega_{\lambda, \tau}(\zeta)$ using differential equation
for (\ref{gammaW}):
\begin{equation}\label{Diff_W}
\frac{d\gamma_{\lambda,\tau}(W(\zeta))}{d\tau}=L_{\lambda,\sigma}^{*}(\gamma_{\lambda,\tau}(W(\zeta))).
\end{equation}
By (\ref{LW}) the right-hand side of (\ref{Diff_W}) is given by
\begin{align}
&L_{\lambda,\sigma}^{*}(\gamma_{\lambda,\tau}(W(\zeta)))\nonumber\\
&=e^{-\Omega_{\lambda,\tau}(\zeta)}(-\frac{i}{\sqrt{2}}\mu\overline{\zeta(\tau)}W(\zeta(\tau))
b-\frac{i}{\sqrt{2}}\overline{\mu}\zeta(\tau) b^*W(\zeta(\tau))\nonumber\\
&-\frac{\sigma_+}{2}|\zeta(\tau)|^2W(\zeta(\tau))+\frac{i}{\sqrt{2}}\frac{\lambda(\sigma_--\sigma_+)}
{2\epsilon}(\zeta(\tau)+\overline{\zeta(\tau)})W(\zeta(\tau)))\label{LW_tau}.
\end{align}
The derivative of the $\tau$-dependent Weyl operator $W(\zeta(\tau))$ is calculated using Baker-Campbell-Hausdorff
formula (\ref{B-C-H-formula})
\begin{align*}
&\frac{dW(\zeta(\tau))}{d\tau}\\
&=\frac{i}{\sqrt{2}}\frac{d\zeta(\tau)}{d\tau}b^*W(\zeta(\tau))+\frac{i}{\sqrt{2}}
\frac{d\overline{\zeta(\tau)}}{d\tau}W(\zeta(\tau))b-\frac{1}{4}W(\zeta(\tau))\frac{d(\zeta(\tau)
\overline{\zeta(\tau)})}{d\tau}.
\end{align*}
Therefore, $\gamma_{\lambda,\tau}(W(\zeta))$ satisfies the following differential equation
\begin{align}
&\frac{d\gamma_{\lambda,\tau}(W(\zeta))}{d\tau}\nonumber\\
&=e^{-\Omega_{\lambda,\tau}(\zeta)}\left(\frac{i}{\sqrt{2}}\frac{d\zeta(\tau)}{d\tau}b^*W(\zeta(\tau))+
\frac{i}{\sqrt{2}}\frac{d\overline{\zeta(\tau)}}{d\tau}W(\zeta(\tau))b\right)\label{dW}\\
&+e^{-\Omega_{\lambda,\tau}(\zeta)}\left(-\frac{d\Omega_{\lambda, \tau}
(\zeta)}{d\tau}W(\zeta(\tau))-\frac{1}{4}W(\zeta(\tau))\frac{d(\zeta(\tau)\overline{\zeta(\tau)})}{d\tau}
\right)\nonumber.
\end{align}

Due to (\ref{Diff_W}), we can match the right-hand side of  (\ref{dW})  with (\ref{LW_tau}) and
obtain the following system of differential equations for the functions
$\zeta(\tau)$ and $\Omega_{\lambda, \tau}(\zeta)$:
\begin{equation*}
\frac{d\zeta(\tau)}{d\tau}=-\overline{\mu} \zeta(\tau)
\end{equation*}
and
\begin{equation}
\frac{d\Omega_{\lambda, \tau}(\zeta)}{d\tau}=-\frac{1}{4}\frac{d(\zeta(\tau)\overline{\zeta(\tau)})}{d\tau}+
\frac{\sigma_+}{2}|\zeta(\tau)|^2-\frac{i}{\sqrt{2}} \frac{\lambda(\sigma_--\sigma_+)}{2\epsilon}(\zeta(\tau)+
\overline{\zeta(\tau)}) .\label{diff_eq_Omega}
\end{equation}
The solution to the first differential equation is
\begin{equation*}
\zeta(\tau)= e^{-\overline{\mu}\tau}\zeta,
\end{equation*}
and using it in the second equation we find
\begin{equation*}
\frac{d\Omega_{\lambda,\tau}(\zeta)}{d\tau}=\frac{\sigma_-+\sigma_+}{4}e^{-(\sigma_--\sigma_+)\tau}|\zeta|^2-
\frac{i}{\sqrt{2}}\frac{\lambda(\sigma_--\sigma_+)}{2\epsilon}(e^{-\overline{\mu}\tau}\zeta+e^{-\mu\tau}\overline{\zeta}).
\end{equation*}
Therefore, the solution of (\ref{diff_eq_Omega}) is
\begin{align*}
\Omega_{\lambda,\tau}(\zeta)&=\frac{|\zeta|^2}{4}\frac{\sigma_-+\sigma_+}{\sigma_--\sigma_+}
(1-e^{-(\sigma_--\sigma_+)\tau})\\
&-\frac{i}{\sqrt{2}} \frac{\lambda(\sigma_--\sigma_+)}{2\epsilon\mu}(1-e^{-\mu\tau})\overline{\zeta}-
\frac{i}{\sqrt{2}}\frac{\lambda(\sigma_--\sigma_+)}
{2\epsilon\overline{\mu}} (1-e^{-\overline{\mu}\tau}){\zeta}.
\end{align*}
Combining the solutions $\zeta(\tau)$ and $\Omega_{\lambda,\tau}(\zeta)$ with the expression for
$\gamma_{\lambda,\tau}(W(\zeta))$ (\ref{gammaW}) we get
\begin{align*}
&\gamma_{\lambda,\tau}(W(\zeta))=W(e^{-\overline{\mu}\tau}
\zeta)\exp\left(-\frac{|\zeta|^2}{4}\frac{\sigma_-+\sigma_+}{\sigma_--\sigma_+}(1-e^{-(\sigma_--\sigma_+)\tau})\right)\\
&\times\exp\left(\frac{i}{\sqrt{2}}
\frac{\lambda(\sigma_--\sigma_+)}{2\epsilon\mu}(1-e^{-\mu\tau})\overline{\zeta}+\frac{i}{\sqrt{2}}
\frac{\lambda(\sigma_--\sigma_+)}
{2\epsilon\overline{\mu}} (1-e^{-\overline{\mu}\tau})\zeta\right).
\end{align*}

To calculate $\cL_{\sigma}^{*}(W(\zeta))$ we use (\ref{cL-dual-sigma1}). Then for the first term one gets
\begin{align*}
&p \ S^{-1}(\gamma_{\lambda,\tau}(S(W(\zeta)))\\
&=p \ e^{-{i}{\lambda}(\zeta+\overline{\zeta})/{\sqrt{2}}{\epsilon}} \
e^{-{\lambda}(b^*-b)/{\epsilon}}\gamma_{\lambda,\tau}(W(\zeta))e^{{\lambda}(b^*-b)/{\epsilon}}\\
&=p \ \exp\left(-\frac{|\zeta|^2}{4}\frac{\sigma_-+\sigma_+}{\sigma_--\sigma_+}(1-e^{-(\sigma_--\sigma_+)\tau})\right)\\
&\times \exp\left(\frac{1}{\sqrt{2}}\left(\frac{\lambda}{{\mu}}(1-e^{-{\mu}\tau})\overline{\zeta}-\frac{\lambda}
{\overline{\mu}}(1-e^{-\overline{\mu}\tau}){\zeta}\right)\right) W(e^{-\overline{\mu}\tau}\zeta) \ .
\end{align*}
If we put $\lambda=0$ in this expression and substitute $p$ for $1-p$, then we obtain the second term of
(\ref{cL-dual-sigma1}). Together this yields the one-step evolution for the Weyl operator:
\begin{align}\label{W-sigma-n=1}
&\cL_{\sigma}^{*}(W(\zeta))=W(e^{-\overline{\mu}\tau}\zeta) \, \exp\left(-\frac{|\zeta|^2}{4}\frac{\sigma_- + \sigma_+}
{\sigma_--\sigma_+}(1-e^{-(\sigma_--\sigma_+)\tau})\right)\\
&\times \Big\{p\exp\left(\frac{1}{\sqrt{2}}\left(\frac{\lambda}{{\mu}}(1-e^{-{\mu}\tau})\overline{\zeta}-\frac{\lambda}
{\overline{\mu}}(1-e^{-\overline{\mu}\tau}){\zeta}\right)\right)+1-p\Big\} \ .\nonumber
\end{align}
Note that operator $\cL_{\sigma}^{*}$ is a convex combination (\ref{cL-dual-sigma1}) of quasi-free (completely positive)
maps and hence, in general, it is \textit{not} quasi-free itself.
Evolution of the Weyl operator for the time $t=n\tau$ ($n$-step evolution) follows from (\ref{W-sigma-n=1}) by induction:
\begin{align}\label{nth power on W}
&(\cL_{\sigma}^{*})^n(W(\zeta))=W(e^{-n\overline{\mu}\tau}\zeta)\exp\left(-\frac{|\zeta|^2}{4}\frac{\sigma_-+\sigma_+}
{\sigma_--\sigma_+}(1-e^{-n(\sigma_--\sigma_+)\tau})\right)\\
&\times\prod_{k=0}^{n-1}\Big\{p\exp\left(\frac{1}{\sqrt{2}}\left(\frac{\lambda}{{\mu}}(1-e^{-{\mu}\tau})
e^{-k{\mu}\tau}\overline{\zeta}-\frac{\lambda}{\overline{\mu}}(1-e^{-\overline{\mu}\tau})e^{-k\overline{\mu}\tau}
{\zeta}\right)\right)+1-p\Big\}.\nonumber
\end{align}

If $\sigma_{-} > \sigma_{+}$, then ${\rm{Re}}(\mu) <0$, which implies
\begin{equation}\label{w-lim of W}
w^*-\lim_{n \rightarrow \infty} W(e^{-n\overline{\mu}\tau}\zeta) = \idty \ ,
\end{equation}
see Appendix A.3. To prove the convergence of the product (\ref{nth power on W}) for the limit $n \rightarrow \infty$,
we denote
\begin{equation*}
h_k(\zeta):=p(\exp{\frac{1}{\sqrt{2}}\left(\frac{\lambda}{\mu}(1-e^{-\mu\tau})e^{-k\mu\tau}\overline{\zeta}-
\frac{\lambda}{\overline{\mu}}(1-e^{-\overline{\mu}\tau})
e^{-k\overline{\mu}\tau}\zeta\right)}-1) \ .
\end{equation*}
Since the product $\prod_{k=0}^\infty (1+h_k(\zeta))$ converges if and only if we establish convergence of the
series $\sum_{k=0}^\infty |h_k(\zeta)|$, we have to estimate the terms $\{|h_k(\zeta)|\}_{k\geq 1}$. Note that
\begin{equation}\label{h-k}
|h_k(\zeta)|\leq 2\sqrt{2}p\frac{\lambda|\zeta|}{|\mu|^2}(\frac{\sigma_--\sigma_+}{2}+\epsilon)
(1+e^{-(\sigma_{-} - \sigma_{+})\tau/2})e^{-k(\sigma_{-} - \sigma_{+})\tau/2} \ .
\end{equation}
Hence, (\ref{h-k}) ensures the convergence of the series and the infinite product for $\sigma_{-} - \sigma_{+} >0$.

Summarizing, we obtain that for the open cavity with parameters $\sigma_{-} > \sigma_{+} \geq 0$ the characteristic
functional (see Appendix A.5) for the limiting state $\omega_{C,\sigma}(\cdot)$ (\ref{lim-cavity-sigma}) exists and
is given by
\begin{align}\label{lim-cavity}
&\omega_{C,\sigma}(W(\zeta))=\exp\left(-\frac{|\zeta|^2}{4}\frac{\sigma_-+\sigma_+}
{\sigma_--\sigma_+}\right)\\
&\times\prod_{k=0}^{\infty}\Big\{p\exp\left(\frac{1}{\sqrt{2}}\left(\frac{\lambda}{{\mu}}(1-e^{-{\mu}\tau})
e^{-k{\mu}\tau}\overline{\zeta}-\frac{\lambda}{\overline{\mu}}(1-e^{-\overline{\mu}\tau})e^{-k\overline{\mu}\tau}
{\zeta}\right)\right)+1-p\Big\}, \nonumber
\end{align}
which is independent of the initial cavity state $\rho_{\mathcal{C}}$.
\end{proof}
\begin{remark}
Notice that the evolution of the Weyl operator for the time $t=n\tau$ can be written as a convex linear
combination of quasi-free completely positive maps (see Appendix A.4):
\begin{align*}
&(\cL_{\sigma}^{*})^n(W(\zeta))=W(e^{-n\overline{\mu}\tau}\zeta)\exp\left(-\frac{|\zeta|^2}{4}\frac{\sigma_-+\sigma_+}
{\sigma_--\sigma_+}
(1-e^{-n(\sigma_--\sigma_+)\tau})\right)\\
&\times\sum_{m=0}^{n-1}p^m(1-p)^{n-1-m}\sum_{1\leq {k_1}<...<{k_m}\leq n-1}\\
&\times\exp\left(\frac{1}{\sqrt{2}}
\frac{\lambda}{|\mu|^2}\left( \overline{\mu}(1-e^{-{\mu}\tau})(e^{-k_1{\mu}\tau}+...+
e^{-k_m{\mu}\tau}) \overline{\zeta} - {\rm c.c.}\right)\right).
\end{align*}
Here c.c. stands for the complex conjugate of the first term. Then (\ref{lim-cavity}) and the last formula suggest that
the limiting state $\omega_{C,\sigma}(\cdot)$ is \textit{not} quasi-free.
\end{remark}

The rest of the chapter is devoted to the proof of Theorem \ref{number of photons-sigma}, which is similar to the
case $\sigma_\mp=0$. The number of photons (\ref{N(t)_s}) for the time $t=n\tau$ can be calculated using the
adjoint operator
\begin{equation}\label{n-number-sigma}
N_{\sigma}(n\tau)= \Tr_{{\mathcal{C}}}(b^*b  \ \cL_{\sigma}^n(\rho_{\mathcal{C}})) = \Tr_{{\mathcal{C}}}
((\cL_{\sigma}^{*})^n(b^*b)\ \rho_{\mathcal{C}}) \ ,
\end{equation}
where $\rho_{\mathcal{C}}$ is the initial gauge-invariant state of the cavity.
Note that by (\ref{cL-dual-sigma1}) and (\ref{cL-dual-sigma2}) we obtain for $A=b^*b$
\begin{align}\label{cL^t(b^*b)+s}
\cL_{\sigma}^{*}(b^*b) = p \ S^{-1}(e^{\tau L_{\lambda,\sigma}^{*}} S(b^*b))
+(1-p)e^{\tau L_{0,\sigma}^{*}}(b^*b) \ .
\end{align}
\begin{lemma}\label{number-photon-sigma}
The action of the adjoint operator $\cL_{\sigma}^{*}$ on the photon number operator
is explicitly given by
\begin{align}
&\cL_{\sigma}^{*}(b^*b)\label{action on b^*b}\\
&=e^{-(\sigma_--\sigma_+)\tau}b^*b+p\frac{i\lambda}{\mu}e^{-(\sigma_--\sigma_+)\tau}
(1-e^{\mu\tau})b^*-p\frac{i\lambda}{\overline{\mu}}e^{-(\sigma_--\sigma_+)\tau}
(1-e^{\overline{\mu}\tau})b\nonumber\\
&+p\frac{\lambda^2}{|\mu|^2}e^{-(\sigma_--\sigma_+)\tau}(1-e^{\mu\tau})(1-e^{\overline{\mu}\tau})+
\frac{\sigma_+}{\sigma_--\sigma_+}(1-e^{-(\sigma_--\sigma_+)\tau}),\nonumber
\end{align}
where $\mu = i\epsilon + ( \sigma_--\sigma_+ )/{2}$ .
\end{lemma}
\begin{proof} We start with the first term in the right-hand side of (\ref{cL^t(b^*b)+s}). Since
$\gamma_{\lambda,\tau}(A)=
e^{\tau L_{\lambda,\sigma}^{*}}(A)$ (\ref{Dyn-C}), one can calculate $\gamma_{\lambda,\tau}(S(b^*b))$ by taking into
account (\ref{L_lambda(B)}). Then
\begin{align*}
&L_{\lambda,\sigma}^{*}((b^*-{\lambda}/{\epsilon})(b-{\lambda}/{\epsilon}))=i\lambda b-i\lambda b^*-(\sigma_--\sigma_+)
(b^*-{\lambda}/{\epsilon})(b-{\lambda}/{\epsilon})+\sigma_+
\end{align*}
and
\begin{align*}
L_{\lambda,\sigma}^{*}(b^*)=-\overline{\mu}b^*+\frac{\lambda(\sigma_--\sigma_+)}{2\epsilon} \ , \
L_{\lambda,\sigma}^{*}(b)=-\mu b+\frac{ \lambda(\sigma_--\sigma_+)}{2\epsilon} \ .
\end{align*}
Therefore, we obtain for mapping $\gamma_{\lambda,\tau}(\cdot)$ the following system of differential equations:
\begin{align*}
&\frac{d\gamma_{\lambda,\tau}((b^*-{\lambda}/{\epsilon})(b-{\lambda}/{\epsilon})}{d\tau}\\
&=
-(\sigma_--\sigma_+)\gamma_{\lambda,\tau}
((b^*-{\lambda}/{\epsilon})(b-{\lambda}/{\epsilon}))+i\lambda\gamma_{\lambda,\tau}(b)-i\lambda\gamma_{\lambda,\tau}
(b^*)+\sigma_+\\
&\frac{d\gamma_{\lambda,\tau}(b)}{d\tau}=-\mu\gamma_{\lambda,\tau}(b)+\frac{\lambda(\sigma_--\sigma_+)}{2\epsilon}\\
&\frac{d\gamma_{\lambda,\tau}(b^*)}{d\tau}=-\overline{\mu}\gamma_{\lambda,\tau}(b^*)+\frac{\lambda(\sigma_--\sigma_+)}
{2\epsilon}.
\end{align*}
The solution of this system is
\begin{align}
&\gamma_{\lambda,\tau}(b^*)=e^{-\overline{\mu}\tau}b^*+\frac{\lambda}{\epsilon}(1-e^{-\overline{\mu}\tau})+
i \ \frac{\lambda}{\overline{\mu}}(1-e^{-\overline{\mu}\tau}) \ ,
\label{gamma_b^*}\\
&\gamma_{\lambda,\tau}(b)= e^{-\mu\tau}b+\frac{\lambda}{\epsilon}(1-e^{-{\mu}\tau})
-i \ \frac{\lambda}{{\mu}}(1-e^{-{\mu}\tau}) \ , \label{gamma_b}\\
&\gamma_{\lambda,\tau}((b^*-{\lambda}/{\epsilon})(b-{\lambda}/{\epsilon}))=e^{-(\sigma_--\sigma_+)\tau}b^*b+
(\frac{i\lambda}{\mu}(1-e^{\mu\tau})-\frac{\lambda}{\epsilon})e^{-(\sigma_--\sigma_+)\tau}b^*\label{gamma_N_lambda}
\nonumber\\
&+(-\frac{i\lambda}{\overline{\mu}}(1-e^{\overline{\mu}\tau})-\frac{\lambda}{\epsilon})e^{-(\sigma_--\sigma_+)\tau}b+
\frac{\lambda^2}{|\mu|^2}(1-e^{-(\sigma_--\sigma_+)\tau})\nonumber\\
&+\frac{\sigma_+}{\sigma_--\sigma_+}(1-e^{-(\sigma_--\sigma_+)
\tau})+\frac{\lambda^2}{\epsilon^2}e^{-(\sigma_--\sigma_+)\tau}\nonumber\\
&-\frac{\lambda^2(\sigma_--\sigma_+)\sin\epsilon\tau}{\epsilon|\mu|^2}e^{-(\sigma_{-} - \sigma_{+})\tau/2} \ .
\end{align}
Making the shift transformation (\ref{Vn}) of $\gamma_{\lambda,\tau}((b^*-{\lambda}/{\epsilon})(b-{\lambda}/{\epsilon}))$
and calculating the second term in (\ref{cL^t(b^*b)+s}) by setting $\lambda=0$, we obtain (\ref{action on b^*b}).
\end{proof}
If one plugs in (\ref{cL-dual-sigma1}) and in (\ref{L_lambda(B)}),(\ref{L_0_lambda(B)}) operators $A=b^*$ or $=b$,
then (\ref{gamma_b^*}), (\ref{gamma_b}) yield for $k\geq1$
\begin{align}\label{action on b^*}
(\cL_{\sigma}^{*})^{k}(b^*)=e^{-k\bar{\mu}\tau}b^*+p\frac{i\lambda}{\bar{\mu}}(1-e^{-k\bar{\mu}\tau}) \ ,
\end{align}
\begin{align}\label{action on b}
(\cL_{\sigma}^{*})^{k}(b)=e^{-k\mu\tau}b-p\frac{i\lambda}{\mu}(1-e^{-k\mu\tau}) \ .
\end{align}
Since $\cL_{\sigma}^{*}(\cdot)\mid_{\sigma_{-} = \sigma_{+} =0} =\cL^{*}(\cdot)$, formulae (\ref{action on b^*}),
(\ref{action on b}) coincide for $\sigma_{-} = \sigma_{+} =0$ with (\ref{cL(b^*-n)}), (\ref{cL(b-n)}) for the ideal
cavity.

Notice also that (\ref{cL-dual-sigma1}),(\ref{cL-dual-sigma2}) and (\ref{L_lambda(B)}),(\ref{L_0_lambda(B)}) imply
\begin{equation}\label{action on I}
(\cL_{\sigma}^{*})(\idty) = 0  \ .
\end{equation}

\begin{proof}(of Theorem \ref{number of photons-sigma})
To this end, we construct first a $4\times4$ matrix $\widehat{\cL}_{\sigma}^{*}$ acting on the complex linear space
spanned by the operator-valued vectors $(b^*b,0,0,0)$, $(0,b^*,0,0)$, $(0,0,b,0)$, $(0,0,0,\idty))$, according to the
formulae (\ref{action on b^*b}), (\ref{action on b^*}), (\ref{action on b}) and (\ref{action on I}), for $k=1$.
Then diagonalisation of $\widehat{\cL}_{\sigma}^{*}$ allows to calculate powers $(\widehat{\cL}_{\sigma}^{*})^n$ and to
find explicit expressions for the $n$-step mapping
\begin{align}
&(\cL_{\sigma}^{*})^n(b^*b)=e^{-n(\sigma_--\sigma_+)\tau}b^*b+p\frac{i\lambda}{\mu}(e^{-n(\sigma_--\sigma_+)\tau}-
e^{-n\bar{\mu}\tau})b^*\label{n-th
power on b^*b}\\
&-p\frac{i\lambda} {\bar{\mu}}(e^{-n(\sigma_--\sigma_+)\tau}-e^{-n\mu\tau})b\nonumber\\
&+p\frac{\lambda^2}{|\mu|^2}e^{-(\sigma_--\sigma_+)\tau}(1-e^{\mu\tau})(1-e^{\bar{\mu}\tau})
\frac{1-e^{-n(\sigma_--\sigma_+)\tau}}
{1-e^{-(\sigma_--\sigma_+)\tau}}\nonumber\\
&-p^2\frac{2\lambda^2}{|\mu|^2}\frac{1-e^{-n(\sigma_--\sigma_+)\tau}}
{1-e^{-(\sigma_--\sigma_+)\tau}}(1-e^{-\frac{\sigma_--\sigma_+}{2}
\tau}\cos\epsilon\tau)\nonumber\\
&+p^2\frac{2\lambda^2}{|\mu|^2}(1-e^{-n(\sigma_{-} - \sigma_{+})\tau/2}\cos n\epsilon\tau)+\frac{\sigma_+}
{\sigma_--\sigma_+}(1-e^{-n(\sigma_--\sigma_+)\tau})\nonumber.
\end{align}
Note that (\ref{n-th power on b^*b}) reduces to (\ref{action on b^*b}) for $n=1$. If $\tau=0$,
then (\ref{n-th power on b^*b}) yields $b^*b$ for any $n$. By (\ref{n-th power on b^*b}) we obtain
for a \textit{gauge-invariant} initial state the mean-value of the photon number (\ref{N(t)_s}) in the open cavity
at $t=n\tau$:
\begin{align}\label{mean-value}
&N_\sigma(n\tau) := \Tr_\cC(\rho_C(\cL_{\sigma}^{*})^n(b^*b))=e^{-n(\sigma_--\sigma_+)\tau}N_\sigma(0)\\
&+p(1-p)\frac{2\lambda^2}{|\mu|^2}(1-e^{-(\sigma_{-} - \sigma_{+})\tau/2}\cos \epsilon\tau)
\frac{1-e^{-n(\sigma_--\sigma_+)\tau}}{1-e^{-(\sigma_--\sigma_+)\tau}}\nonumber\\
&+p^2\frac{2\lambda^2}{|\mu|^2}(1-e^{-n(\sigma_{-} - \sigma_{+})\tau/2}\cos n\epsilon\tau)
+\frac{\sigma_+}{\sigma_--\sigma_+}(1-e^{-n(\sigma_--\sigma_+)\tau})\nonumber \ .
\end{align}
Note that (\ref{mean-value}) coincides with (\ref{mean-valueBIS}) for $t=n\tau$.
Finally, by virtue of (\ref{n-th power on b^*b}) and (\ref{mean-value}) one gets for the $w^*$-limit (see Appendix A.3)
\begin{align}\label{w-lim-sigma}
&w^*-\lim_{n\rightarrow\infty}(\cL_{\sigma}^{*})^n(b^*b) \\
&= p(1-p) \frac{2\lambda^2}{|\mu|^2}
\frac{1 - e^{-(\sigma_--\sigma_+)\tau/2} \cos\epsilon\tau}{1-e^{-(\sigma_--\sigma_+)\tau}}
+ p (2p-1)\frac{\lambda^2}{|\mu|^2} + \frac{\sigma_+}{\sigma_--\sigma_+} \ . \nonumber
\end{align}
Therefore, $\lim_{n\rightarrow\infty}\omega_{\mathcal{C},\sigma}^{n\tau}(b^*b)=
\lim_{n\rightarrow\infty}\Tr_\cC(\rho_C(\cL_{\sigma}^{*})^n(b^*b))$ and (\ref{w-lim-sigma}) yield
the proof of (\ref{lim-photons-number-sigma}).
\end{proof}

Notice that in the limit of the ideal cavity: $\sigma_{-}\rightarrow +0, \, \sigma_{-} > \sigma_{+} \geq 0$,
(\ref{mean-value}) gives $\lim_{\sigma_{-}\rightarrow +0} N_\sigma(t) = N(t)$, which is
(\ref{number of photons_Hamiltonian}). For other limiting cases see Remark \ref{instructive cases}.

\section{Energy and Entropy Variations}\label{EEP}
\subsection{Energy variation in the ideal cavity}\label{EEP_0}
Since the time-dependent interaction in (\ref{W-int}) is piecewise constant, the system is {autonomous}
on each interval $[(n-1)\tau, n\tau)$. Therefore, there is no
variation of energy on this interval, but it may \textit{jump}, when a new atom enters into the cavity.
Note that although the \textit{total} energy corresponding to the infinite system (\ref{Ham-Model}) is undefined
its variation is well-defined \cite{BJM06}-\cite{BJM10}.

The times when the $n$-th atom is actually traveling in the cavity are of the
form  $t = n(t)\tau + \nu(t)$, with $n(t)=n-1$ and $\nu(t)\in[0, \tau)$, see (\ref{t}).
To calculate the energy variation we compare two total energy expectations: for the moment
$t_{n}=(n-1)\tau + \nu(t_{n})$, when the $n$-th atom is present in the cavity, and for the moment
$t_{n-1}=(n-2)\tau + \nu(t_{n-1})$, when the $(n-1)$-th atom was in the cavity.
Then by (\ref{Ham-Model}), (\ref{Ham-n}) and by (\ref{Sol-Liouv-Eq}) we obtain
for the total energy variation in the system (\ref{Ham-Model}) between two moments $t_{n-1}$ and $t_{n}$
the following expression:
\begin{align}\label{DEn-1,n1}
&\Delta\mathcal{E}(t_{n},t_{n-1}):= \omega^{t_{n}}_{\mathcal{S}}(H(t_{n})) -
\omega^{t_{n-1}}_{\mathcal{S}}(H(t_{n-1})) \\
&=\Tr(e^{-i \nu(t_{n})H_{n}}\rho_{\mathcal{S}}((n-1)\tau)e^{-i \nu(t_{n})H_{n}} H_{n})
\nonumber \\
&- \, \Tr
(e^{-i \nu(t_{n-1})H_{n-1}}\rho_{\mathcal{S}}((n-2)\tau)e^{i \nu(t_{n-1})H_{n-1}}
H_{n-1}) \nonumber \\
&=\Tr(\rho_{\mathcal{S}}((n-1)\tau) H_{n}) \nonumber \\
&- \, \Tr(e^{-i \tau H_{n-1}}\rho_{\mathcal{S}}((n-2)\tau)
e^{i \tau H_{n-1}} H_{n-1})
\nonumber \\
&=\Tr(T_{(n-1)\tau,0}(\rho_{\mathcal{C}} \otimes \rho_{\mathcal{A}}) [H_{n} - H_{n-1}]) \nonumber \,.
\end{align}
Here we used that the system (\ref{Ham-Model}), (\ref{Ham-n}) is piecewise autonomous with
$H(t+0) = H(t)$, and the state $\rho_{\mathcal{S}}(t)$ is $w^\ast$-time-continuous (Appendix A.2).

Recall that by duality (\ref{S-state-evol-adj}) we have
\begin{align}\label{PI}
&{\rm{Tr}}\, (T_{(n-1)\tau,0}(\rho_{\mathcal{C}} \otimes \rho_{\mathcal{A}})[H_{n} - H_{n-1}])\\
&={\rm{Tr}}\, (\rho_{\mathcal{C}} \otimes \rho_{\mathcal{A}} \ T_{(n-1)\tau,0}^{\ast}(H_{n} - H_{n-1})) \nonumber
 \ .
\end{align}
Since (\ref{Ham-n}) implies
\begin{equation}\label{H-n-H-n-1}
H_{n} - H_{n-1} = \lambda \, (b^*+b)\otimes (\eta_n -\eta_{n-1})\ ,
\end{equation}
by (\ref{DEn-1,n1}), (\ref{PI}) and by $[H_{k'}, \eta_k] = 0$ we obtain
\begin{align}\label{DEn-1,n2}
&\Delta\mathcal{E}(t_{n},t_{n-1}) \\
&= \Tr\{\rho_{\mathcal{C}}\otimes\rho_{\mathcal{A}} \ T_{(n-1)\tau,0}^{\ast}
(\lambda \,(b^*+b)\otimes \idty) [\idty\otimes (\eta_n -\eta_{n-1})]\} \nonumber
\end{align}
\begin{lemma}\label{n-dual-on-b} For any $n\geq 1$ one gets:
\begin{align}\label{k-iter}
&T_{n\tau,0}^{\ast} (b^* \otimes \idty) = e^{n i \tau \epsilon} b^* \otimes \idty -
\frac{\lambda}{\epsilon} (1- e^{i \tau \epsilon})\sum_{k=1}^{n} e^{(n-k) i \tau \epsilon} \idty \otimes \eta_k  \ , \\
&T_{n\tau,0}^{\ast} (b \otimes \idty) = e^{- n i \tau \epsilon} b\otimes \idty -
\frac{\lambda}{\epsilon} (1- e^{- i \tau \epsilon})\sum_{k=1}^{n} e^{-(n-k) i \tau \epsilon} \idty \otimes \eta_k
\label{k-iter-b}\ .
\end{align}
\end{lemma}
\begin{proof} Let us define
\begin{equation}\label{B-k}
B_{k}^*(\tau) : = T_{k}^{\ast}(b^* \otimes \idty)= e^{i \tau H_{k}} (b^* \otimes \idty)e^{- i \tau H_{k}} \ ,
\  k \geq 1 \ .
\end{equation}
Then by (\ref{Ham-n}) and (\ref{b-evol}), (\ref{T-adj}) the operator (\ref{B-k}) is solution of equation
\begin{equation*}
\partial_{s}B_{k}^*(s)= i [H_{k}, B_{k}^*(s)] = i \epsilon B_{k}^*(s) + \lambda \idty \otimes \eta_k \ ,
\ B_{k}^*(0) = b^* \otimes \idty \ ,
\end{equation*}
which has the following explicit form:
\begin{equation}\label{B*-k-sol}
B_{k}^*(\tau) = e^{i \tau \epsilon} (b^* \otimes \idty) - \frac{\lambda}{\epsilon}
(1- e^{i \tau \epsilon}) \idty \otimes \eta_k \ .
\end{equation}
Similarly one obtains
\begin{equation}\label{B-k-sol}
B_{k}(\tau) = e^{- i \tau \epsilon} (b \otimes \idty) - \frac{\lambda}{\epsilon}
(1- e^{- i \tau \epsilon}) \idty \otimes \eta_k \ .
\end{equation}
Note that by iteration of (\ref{B*-k-sol}) for $k=1,2$ one gets
\begin{align}\label{2-iter}
T_{2\tau,0}^{\ast}(b^* \otimes \idty) &= e^{2 i \tau \epsilon} (b^* \otimes \idty) -
\frac{\lambda}{\epsilon} (1- e^{i \tau \epsilon})e^{i \tau \epsilon} \idty \otimes \eta_1   \\
&- \frac{\lambda}{\epsilon} (1- e^{i \tau \epsilon}) \idty \otimes \eta_2  \nonumber \ .
\end{align}
Proceeding with iteration of (\ref{2-iter}) we obtain (\ref{k-iter}). Using (\ref{B-k-sol}) one
proves in a similar way (\ref{k-iter-b}).
\end{proof}
Recall that in the present paper we suppose that atomic beam is homogeneous (\ref{atom-state}), i.e.
$p = \Tr\{\rho_{\mathcal{C}}\otimes\rho_{\mathcal{A}}
(\idty \otimes \eta_n)\}$ is independent of $n$ probability that the $n$-th atom is in the excited state, and
that the atomic beam is Bernoulli (Remark \ref{Rem-Q-El}):
\begin{equation}\label{p-2p}
\Tr_{{\mathcal{A}}}\{\rho_{\mathcal{A}} (\eta_{n_1} \eta_{n_2})\} = \delta_{{n_1},{n_2}}\, p +
(1-\delta_{{n_1},{n_2}})\, p^2 \ .
\end{equation}
Let the initial cavity state $\rho_{\mathcal{C}}$ be gauge-invariant. Then by Lemma \ref{n-dual-on-b}
and (\ref{p-2p}) one obtains:
\begin{align}\label{1st-term}
&\Tr\{\rho_{\mathcal{C}}\otimes\rho_{\mathcal{A}} \ T_{(n-1)\tau,0}^{\ast}
(\lambda \,(b^*+b)\otimes \idty)[\idty\otimes (\eta_n -\eta_{n-1})]\} \\
&=- \frac{\lambda^2}{\epsilon}  (1- e^{i \tau \epsilon}) \sum_{k=1}^{n-1} e^{(n-k-1) i \tau \epsilon}
\Tr_{{\mathcal{A}}}  \rho_\cA \eta_k (\eta_n -\eta_{n-1})                                      \nonumber \\
& - \frac{\lambda^2}{\epsilon}  (1- e^{-i \tau \epsilon}) \sum_{k=1}^{n-1} e^{-(n-k-1) i \tau \epsilon}
\Tr_{{\mathcal{A}}}  \rho_\cA\eta_k (\eta_n -\eta_{n-1}) \nonumber \\
&= p(1-p) \ \frac{2\lambda^2}{\epsilon} \  (1-\cos\tau\epsilon) \nonumber  \ .
\end{align}
Hence, formulae (\ref{DEn-1,n2}), (\ref{1st-term}) prove for the energy variation on the interval $(t_{n-1},t_{n})$
the following statement.
\begin{theorem}\label{ThDEn-1,n1}
For the case of the ideal cavity the total energy variation {\rm{(}}\ref{DEn-1,n1}{\rm{)}} between two moments
$t_{n-1}$ and $t_{n}$, where $n\geq 1$ is
\begin{equation}\label{DEn-1,n3}
\Delta\mathcal{E}(t_{n},t_{n-1}) = p(1-p) \ \frac{2\lambda^2}{\epsilon}(1-\cos\tau\epsilon)\ ,
\end{equation}
i.e. for the variation between $t_{0}= \nu(t_{0})$ and $t_{n} \geq t_{0}$ we obtain:
\begin{equation}\label{DEn-1,n1-fin}
\Delta\mathcal{E}(t_{n},t_{0}) = \sum_{k=1}^{n} \Delta\mathcal{E}(t_{k},t_{k-1}) =
(n-1)\ p(1-p) \ \frac{2\lambda^2}{\epsilon}(1-\cos\tau\epsilon) \ .
\end{equation}
\end{theorem}
\begin{remark}\label{RemEn-1,n1}
The total energy variation (\ref{DEn-1,n1}), when the $n$-th atom is traveling
through the cavity between the moments $t^{\prime} = (n-1) \tau$ and $t^{\prime \prime}=n \tau - 0$, can be written as
\begin{equation}\label{EnVar1}
\Delta\mathcal{E}(t^{\prime \prime},t^{\prime}) =
\omega^{n \tau}_{\mathcal{S}}(H_n) - \omega^{(n-1) \tau}_{\mathcal{S}}(H_n) \ .
\end{equation}
Here again we used that for $t\in[(n-1)\tau, n\tau)$ the Hamiltonian (\ref{Ham-Model}) is piecewise
\textit{constant} and that it has the form (\ref{Ham-n}), as well as that the state $\omega^{t}_{\mathcal{S}}(\cdot)$ is
$w^\ast$-continuous in time (Appendix A.2). Since the system (\ref{Ham-Model}) is autonomous on the interval
$[(n-1)\tau, n\tau)$, one obtains $\Delta\mathcal{E}(n \tau - 0,(n-1) \tau) = 0$. Then (\ref{EnVar1}) implies that
on this interval (in contrast to (\ref{DEn-1,n2})) the variation of the interaction-energy \textit{completely}
compensates the energy variation due to the photon number pumping:
\begin{align}\label{Eq-int-energ-var}
&\omega^{n \tau}_{\mathcal{S}}(\lambda \,(b^*+b)\otimes \eta_n) -
\omega^{(n-1) \tau}_{\mathcal{S}}(\lambda \,(b^*+b)\otimes \eta_{n-1}) \\
&= - [\omega^{n \tau}_{\mathcal{S}}(\epsilon b^*b\otimes \idty) -
\omega^{(n-1) \tau}_{\mathcal{S}}(\epsilon b^*b\otimes \idty)] \ . \nonumber
\end{align}
Note that similar to (\ref{1st-term}) one can check this identity explicitly using Lemma \ref{n-dual-on-b} and
(\ref{p-2p}) applied to the left-hand side of (\ref{Eq-int-energ-var}).
\end{remark}
\subsection{Energy variation in the open cavity}\label{EEP_n0}
Although for the open cavity the time-dependent generator (\ref{Generator}) is still piecewise \textit{constant}
(\ref{Generator-KL}), the cavity energy  is continuously varying between the moments $\{t = k \tau\}_{k\geq0}$
(when the interaction may to jump (\ref{Ham-Model})) because of the leaking/injection of photons.

Therefore, as above we first concentrate on the elementary variation of the total energy, when the $n$-th atom is
traveling through the cavity between the moments $t^{\prime} = (n-1) \tau$ and $t^{\prime \prime} =n \tau - 0$:
\begin{align}
\Delta \mathcal{E}_{\sigma}(t^{\prime \prime}, t^{\prime}):=
\omega^{n \tau}_{\mathcal{S},\sigma}(H_n) - \omega^{(n-1) \tau}_{\mathcal{S},\sigma}(H_n) \ . \label{EnVar-sigma0}
\end{align}
Here again we used two facts: (1) for $t\in[(n-1)\tau, n\tau)$ the Hamiltonian (\ref{Ham-Model}) of the form (\ref{Ham-n})
is piecewise constant; (2) the state $\omega^{t}_{\mathcal{S},\sigma}(\cdot)$
(\ref{dual-sigma}) is time-continuous (Appendix A.2).

By virtue of (\ref{dual-sigma}) and of (\ref{S-state-evol-adj-sigma}) for operator $A= H_n$ (\ref{Ham-n}), we see
that the problem (\ref{EnVar-sigma0}) reduces to calculation of the following expectations:
\begin{equation}\label{EnVar-sigma1}
\epsilon \ \omega^{k\tau}_{\mathcal{S},\sigma}(b^*b \otimes \idty) \ \ {\rm{and}} \ \
\lambda \ \omega^{s\tau}_{\mathcal{S},\sigma} ((b^*+b)\otimes \eta_k) \ , \ k,s \geq 1 \ .
\end{equation}
The first expectation in (\ref{EnVar-sigma1}) is known due to (\ref{C-state-t-sigma}) and Theorem
\ref{number of photons-sigma}, see (\ref{mean-value}):
\begin{equation}\label{EnVar-sigma2}
\epsilon \ \omega^{k\tau}_{\mathcal{S},\sigma}(b^*b \otimes \idty) = \epsilon \, N_\sigma(k\tau) \ .
\end{equation}
To calculate the second expectation in (\ref{EnVar-sigma1}) we use (\ref{dual-sigma}) for the operator
$A =((b^*+b)\otimes \idty)(\idty \otimes \eta_k)$
and the representation (\ref{S-state-evol-adj-sigma}) for initial gauge-invariant state $\rho_{\mathcal{C}}$ and
for homogeneous atoms state $\rho_{\mathcal{A}}$.
\begin{lemma} \label{n-dual-on-b-sigma}
Let $\sigma_- > \sigma_+ \geq 0$. Then for the mappings $\{(T^\sigma_{n\tau,0})^\ast\}_{n\geq0}$, see
{\rm{(}}\ref{S-state-evol-adj-sigma}{\rm{)}}, one obtains:
\begin{align}\label{T-sigma-n-star}
(T^\sigma_{n\tau,0})^*(b\otimes\idty)&=e^{-n\mu\tau}b\otimes\idty-\frac{\lambda i}{\mu}(1-e^{-\mu\tau})
\sum_{k=1}^ne^{-(n-k)\mu\tau}\idty\otimes \eta_k \ ,
\end{align}
and $(T^\sigma_{n\tau,0})^*(b^*\otimes\idty) = ((T^\sigma_{n\tau,0})^*(b\otimes\idty))^*$.
\end{lemma}
\begin{proof}
We prove this lemma by induction. Suppose that formula (\ref{T-sigma-n-star}) is true for $(T_{n\tau,0}^\sigma)^*$.
Then we show that it is also valid for $(T^\sigma_{(n+1)\tau,0})^*$.
By virtue of (\ref{S-state-evol-adj-sigma}) and by (\ref{e_L_tilde}),(\ref{hat-L-dual}) for the action of
operator $e^{\tau L_{\sigma, n}^*}$, one gets
\begin{align*}
&(T^\sigma_{n+1})^*(b\otimes\idty)=(T^\sigma_{n})^*(e^{\tau L_{\sigma, n+1}^*}(b\otimes\idty))\\
&=(T^\sigma_{n})^*(\widehat{S}_{n+1}^{-1}(e^{\tau \widehat{L}^*_{\sigma, n+1}}(\widehat{S}_{n+1}
(b\otimes\idty))))\\
&=(T^\sigma_{n})^*(\widehat{S}_{n+1}^{-1}(e^{\tau \widehat{L}^*_{\sigma, n+1}} \Bigl((b-\frac{\lambda}{\epsilon})\otimes
\eta_{n+1}+b\otimes (I-\eta_{n+1})\Bigr))) \ ,
\end{align*}
where we used (\ref{Vn-shift}) in the last line. By (\ref{reduc1}),(\ref{reduc2}) and (\ref{Dyn-C}),(\ref{gamma_b})
combined with the shift $\widehat{S}_{n+1}^{-1}$ (\ref{relationSSnBIS}), we obtain
\begin{align*}
&\widehat{S}_{n+1}^{-1}(e^{\tau\widehat{L}^*_{\sigma, n+1}}((b-\frac{\lambda}{\epsilon})\otimes \eta_{n+1}))=
\widehat{S}_{n+1}^{-1}((\gamma_{\lambda,\tau}(b)-\frac{\lambda}{\epsilon})\otimes \eta_{n+1})\\
&=(e^{-\mu\tau}b-\frac{\lambda i}{\mu}(1-e^{-\mu\tau}))\otimes \eta_{n+1} \ ,\\
&\widehat{S}_{n+1}^{-1}(e^{\tau\widehat{L}^*_{\sigma, n+1}}(b\otimes (I-\eta_{n+1}))) =
\widehat{S}_{n+1}^{-1}(\gamma_{0,\tau}(b)\otimes (I-\eta_{n+1}))\\
&=e^{-\mu\tau}b\otimes(I - \eta_{n+1}) \ .
\end{align*}
Consequently,
\begin{align*}
&(T^\sigma_{n+1})^*(b\otimes\idty))\\
&=(T^\sigma_{n})^*((e^{-\mu\tau}b-\frac{\lambda i}{\mu}(1-e^{-\mu\tau}))\otimes
\eta_{n+1}+e^{-\mu\tau}b\otimes(\idty- \eta_{n+1}))\\
&=(T^\sigma_{n})^*(e^{-\mu\tau}b\otimes\idty-\frac{\lambda i}{\mu}(1-e^{-\mu\tau})\idty\otimes \eta_{n+1})\\
&=e^{-(n+1)\mu\tau}b\otimes\idty-\frac{\lambda i}{\mu}(1-e^{-\mu\tau})\sum_{k=1}^{n+1}
e^{-(n-k+1)\mu\tau}\idty\otimes \eta_k,
\end{align*}
which proves the lemma.
\end{proof}
Recall that $\mu= i\epsilon + (\sigma_{-} - \sigma_{+})/2$. Hence, in the limit $\sigma_{-} \rightarrow +0$
one recovers from this Lemma formulae (\ref{k-iter}) and (\ref{k-iter-b}) for the ideal cavity.
\begin{corollary}\label{int-energy-n-n-1}
Let initial cavity state $\rho_{\mathcal{C}}$ be gauge-invariant state for homogeneous
state $\rho_{\mathcal{A}}$ of the atomic beam.
Then with help of (\ref{p-2p}) and (\ref{T-sigma-n-star}) one obtains the interaction energy expectations
(\ref{EnVar-sigma1}) corresponding to the difference (\ref{EnVar-sigma0}):
\begin{align}\label{int-energy-n}
&\lambda \ \omega^{n\tau}_{\mathcal{S},\sigma} ((b^*+b)\otimes \eta_n) =  \\
&-\frac{2\lambda^2\epsilon }{|\mu|^2}\left[p(1-p)(1-e^{-(\sigma_{-} - \sigma_{+})\tau/2}\cos\epsilon\tau)
+ p^{2}(1-e^{-n(\sigma_{-} - \sigma_{+})\tau/2}\cos n\epsilon\tau)\right] \nonumber \\
&+\frac{\lambda^2(\sigma_--\sigma_+)}{|\mu|^2}\left[p(1-p)e^{-(\sigma_{-} - \sigma_{+})\tau/2}\sin\epsilon\tau
+p^{2}e^{-n(\sigma_{-} - \sigma_{+})\tau/2}\sin n\epsilon\tau\right] \ , \nonumber \\
&\lambda \ \omega^{(n-1)\tau}_{\mathcal{S},\sigma} ((b^*+b)\otimes \eta_n) =
- \frac{2\lambda^2\epsilon}{|\mu|^2}p^{2} (1-e^{-(n-1)(\sigma_{-} - \sigma_{+})\tau/2}\cos (n-1)\epsilon\tau) \nonumber \\
&+ \frac{\lambda^2(\sigma_--\sigma_+)}{|\mu|^2} p^{2} e^{-(n-1)(\sigma_{-} - \sigma_{+})\tau/2}\sin (n-1)\epsilon\tau \ .
\label{int-energy-n-1}
\end{align}
\end{corollary}
\begin{corollary}\label{energy-n-n-1}
Taking into account Theorem \ref{number of photons-sigma} and (\ref{int-energy-n}),
(\ref{int-energy-n-1}) we get for the elementary variation of the total energy (\ref{EnVar-sigma0})
\begin{align}\label{EnVar-sigma3}
&\Delta \mathcal{E}_{\sigma}(n \tau - 0, (n-1) \tau)=\epsilon \, (N_\sigma(n\tau) - N_\sigma((n-1)\tau)) \\
&+\lambda \ (\omega^{n\tau}_{\mathcal{S},\sigma} ((b^*+b)\otimes \eta_n) -
\omega^{(n-1)\tau}_{\mathcal{S},\sigma} ((b^*+b)\otimes \eta_n)) \nonumber \\
&= \epsilon \, (\frac{\sigma_{+}}{\sigma_{-} - \sigma_{+}}-N_\sigma(0))(1 - e^{-(\sigma_{-} - \sigma_{+})\tau})
e^{-(n-1)(\sigma_{-} - \sigma_{+})\tau} \nonumber \\
&- p(1-p)\frac{2\lambda^2\epsilon }{|\mu|^2}(1-e^{-(\sigma_{-} - \sigma_{+})\tau/2}\cos\epsilon\tau)
(1-e^{-(n-1)(\sigma_{-} - \sigma_{+})\tau}) \nonumber \\
&+ p(1-p) \frac{\lambda^2(\sigma_--\sigma_+)}{|\mu|^2} e^{-(\sigma_{-} - \sigma_{+})\tau/2}\sin \epsilon\tau
\nonumber \\
&+ p^{2}\frac{\lambda^2(\sigma_--\sigma_+)}{|\mu|^2}\left[e^{-n(\sigma_{-} - \sigma_{+})\tau/2}\sin n\epsilon\tau -
e^{-(n-1)(\sigma_{-} - \sigma_{+})\tau/2}\sin (n-1)\epsilon\tau\right]. \nonumber
\end{align}
\end{corollary}
Note that in the limit of the ideal cavity: $\sigma_+ \rightarrow 0$ and $\sigma_- \rightarrow 0$, one finds
for total energy variation (\ref{EnVar-sigma3}): $\Delta \mathcal{E}_{\sigma}(n \tau - 0, (n-1) \tau) =0$,
which corresponds to the autonomous case, see Remark \ref{RemEn-1,n1}. Whereas for $\sigma_- > \sigma_+ \geq 0 $
the external pumping due to ${\sigma_{+}}/(\sigma_{-} - \sigma_{+}) \geq 0$ is in competition with the energy leaking,
see the second term in the right-hand side of (\ref{EnVar-sigma3}). The limit of the energy increment when
$n\rightarrow\infty$ is
\begin{align}\label{lim-increm}
&\lim_{n\rightarrow\infty} \Delta \mathcal{E}_{\sigma}(n \tau - 0, (n-1) \tau) =\\
&p(1-p)\frac{\lambda^2}{|\mu|^2}\left[- 2\epsilon (1-e^{-(\sigma_{-} - \sigma_{+})\tau/2}\cos\epsilon\tau)
+ (\sigma_--\sigma_+) e^{-(\sigma_{-} - \sigma_{+})\tau/2}\sin \epsilon\tau \right]. \nonumber
\end{align}
\begin{remark}\label{EnVar-sigma-jump}
To consider the impact when the $n$-th atom enters the cavity we study the total energy variation
on the extended interval $((n-1)\tau -0, n\tau -0)$. Then
\begin{align}\label{EnVar-sigma-jump1}
&\Delta \mathcal{E}_{\sigma}(n\tau - 0, (n-1)\tau -0) =
(\omega^{n \tau-0}_{\mathcal{S},\sigma}(H_n) - \omega^{(n-1) \tau}_{\mathcal{S},\sigma}(H_n))\\
&+ (\omega^{(n-1) \tau}_{\mathcal{S},\sigma}(H_n) - \omega^{(n-1)\tau -0}_{\mathcal{S},\sigma}(H_{n-1})) \ ,
\nonumber
\end{align}
where the second difference $\Delta \mathcal{E}_{\sigma}((n-1)\tau,(n-1)\tau -0):=
\omega^{(n-1) \tau}_{\mathcal{S},\sigma}(H_n) - \omega^{(n-1)\tau -0}_{\mathcal{S},\sigma}(H_{n-1})$ corresponds to the
energy variation (\textit{jump}), when the $n$-th atom enters the cavity and the $(n-1)$-th atom leaves it.
\end{remark}
To calculate $\Delta \mathcal{E}_{\sigma}((n-1)\tau, (n-1)\tau -0)$ note that by the time continuity of the state
\begin{align*}
&\Delta \mathcal{E}_{\sigma}((n-1)\tau,(n-1)\tau -0)=\Tr(\rho_\cS((n-1)\tau)H_{n})-\Tr(\rho_\cS((n-1)\tau)H_{n-1})\\
&=\Tr\left(e^{\tau L_{\sigma,n-1}}\, ... \ e^{\tau L_{\sigma, 1}}(\rho_\cC\otimes\rho_\cA))(H_n-H_{n-1})\right)\\
&=\Tr\left(T^\sigma_{(n-1)\tau,0}(\rho_\cC\otimes\rho_\cA)(H_n-H_{n-1})\right)\\
&=\Tr\left(\rho_\cC\otimes\rho_\cA(T^\sigma_{(n-1)\tau,0})^*(\lambda(b^*+b)\otimes (\eta_n-\eta_{n-1})\right)\\
&=\Tr\left(\rho_\cC\otimes\rho_\cA(T^\sigma_{(n-1)\tau,0})^*(\lambda(b^*+b)\otimes (\eta_n-\eta_{n-1})\right)\\
&=\Tr\{\rho_\cC\otimes\rho_\cA (T^\sigma_{(n-1)\tau,0})^*(\lambda (b^*+b)\otimes \idty)[\idty\otimes (\eta_n -\eta_{n-1})]\},
\end{align*}
where $T^\sigma_{t=n\tau,0}=e^{\tau L_{\sigma,n}}\, ... \ e^{\tau L_{\sigma, 1}}$ is defined by (\ref{Sol-Liouv-Eq-sigma}).

If the initial cavity state is gauge-invariant, then (\ref{T-sigma-n-star}) yields for the energy jump at the moment
$t=(n-1)\tau$:
\begin{align}\label{energy-jump}
&\Delta \mathcal{E}_{\sigma}((n-1)\tau,(n-1)\tau -0)=\\
&\Tr\{\rho_\cC\otimes\rho_\cA (T^\sigma_{(n-1)\tau,0})^* (\lambda (b^*+b)\otimes \idty)
[\idty\otimes (\eta_n -\eta_{n-1})]\}\nonumber  \\
&= \frac{\lambda^2 i}{\bar{\mu}}(1-e^{-\bar{\mu}\tau})\sum_{k=1}^{n-1}e^{-(n-k-1)\bar{\mu}\tau}
\Tr_\cA (\rho_\cA\eta_k(\eta_n -\eta_{n-1})) \nonumber\\
&-\frac{\lambda^2 i}{\mu}(1-e^{-\mu\tau})\sum_{k=1}^{n-1}e^{-(n-k-1)\mu\tau}\Tr_\cA (\rho_\cA \eta_k(\eta_n -
\eta_{n-1})) .\nonumber
\end{align}
Taking into account the Bernoulli property (\ref{p-2p}) we obtain from (\ref{energy-jump})
\begin{align}
&\Delta \mathcal{E}_{\sigma}((n-1)\tau,(n-1)\tau -0)=
\frac{\lambda^2 i}{\mu}(1-e^{-\mu\tau})p(1-p)-\frac{\lambda^2 i}{\bar{\mu}}(1-e^{-\bar{\mu}\tau})p(1-p)\nonumber\\
&=p(1-p)\frac{2\lambda^2\epsilon}{|\mu|^2}(1-e^{-(\sigma_{-} - \sigma_{+})\tau/2}\cos\epsilon\tau)\nonumber \\
&-p(1-p)\frac{\lambda^2(\sigma_--\sigma_+)}{|\mu|^2}e^{-(\sigma_{-} - \sigma_{+})\tau/2}\sin\epsilon\tau.
\label{second_term}
\end{align}

Notice again that for $\sigma_{-} \rightarrow +0$ one obtains from (\ref{second_term}) the one-step energy variation
for the ideal cavity (\ref{1st-term}).

Summarising (\ref{EnVar-sigma3}) and (\ref{second_term}), we obtain the energy increment (\ref{EnVar-sigma-jump1})
which is due to impact of the open cavity effects (\ref{EnVar-sigma3}) and to the atomic beam pumping (\ref{second_term}):
\begin{align}\label{EnVar-sigma-jump2}
&\Delta \mathcal{E}_{\sigma}(n\tau - 0, (n-1)\tau -0) = \\
&= \epsilon \, (\frac{\sigma_{+}}{\sigma_{-} - \sigma_{+}}-N_\sigma(0))(1 - e^{-(\sigma_{-} - \sigma_{+})\tau})
e^{-(n-1)(\sigma_{-} - \sigma_{+})\tau} \nonumber \\
&+ p(1-p)\frac{2\lambda^2\epsilon }{|\mu|^2}(1-e^{-(\sigma_{-} - \sigma_{+})\tau/2}\cos\epsilon\tau)
e^{-(n-1)(\sigma_{-} - \sigma_{+})\tau} \nonumber \\
&+p^{2}\frac{\lambda^2(\sigma_--\sigma_+)}{|\mu|^2}\left[e^{-n(\sigma_{-} - \sigma_{+})\tau/2}\sin n\epsilon\tau -
e^{-(n-1)(\sigma_{-} - \sigma_{+})\tau/2}\sin (n-1)\epsilon\tau\right]. \nonumber
\end{align}
\begin{theorem}\label{Energy-sigma-variation}
By virtue of (\ref{EnVar-sigma-jump2}) the total energy variation between initial state at the moment $t_0 := -0$,
when the cavity is empty, and the moment $t_n := n\tau -0$, just before the $n$-th atom is ready to leave the
cavity, is
\begin{align}\label{energ-var-op-syst-0-n}
&\Delta \mathcal{E}_{\sigma}(t_n, t_0) = \sum_{k=1}^n \Delta \mathcal{E}_{\sigma}(k\tau - 0, (k-1)\tau -0)\\
&= \epsilon \, (\frac{\sigma_{+}}{\sigma_{-} - \sigma_{+}}-N_\sigma(0))(1 - e^{-n(\sigma_{-} - \sigma_{+})\tau})
\nonumber \\
&+p(1-p)\frac{2\lambda^2\epsilon }{|\mu|^2}(1-e^{-(\sigma_{-} - \sigma_{+})\tau/2}\cos\epsilon\tau)
\frac{1 - e^{-n(\sigma_{-} - \sigma_{+})\tau}}{1 - e^{-(\sigma_{-} - \sigma_{+})\tau}}\nonumber \\
&+p^{2}\frac{\lambda^2(\sigma_--\sigma_+)}{|\mu|^2} e^{-n(\sigma_{-} - \sigma_{+})\tau/2} \sin n\epsilon\tau
\nonumber \ .
\end{align}
Here $N_\sigma(0) = \omega^{t_0}_{\mathcal{S},\sigma}(b^*b\otimes\idty)$ is the initial number of photons in
the cavity.
\end{theorem}
\begin{remark}\label{Energy-sigma-variationBIS}
Note that the total energy variation (\ref{energ-var-op-syst-0-n}) is due to evolution of the photon number in the
open cavity (\ref{mean-value}) and the variation of the interaction energy (\ref{int-energy-n}), that give
\begin{align}\label{energ-var-op-syst-0-nBIS}
\Delta \mathcal{E}_{\sigma}(t_n, t_0) &=
\epsilon \omega^{n\tau}_{\mathcal{S},\sigma}(b^*b\otimes\idty) +
\lambda \ \omega^{n\tau}_{\mathcal{S},\sigma} ((b^*+b)\otimes \eta_n) \\
&-\epsilon \omega^{t_0}_{\mathcal{S},\sigma}(b^*b\otimes\idty) \ . \nonumber
\end{align}
For $\sigma_{-} - \sigma_{+}>0$ it is uniformly bounded from above
\begin{align}\label{energ-var-op-syst-0-n+}
\Delta \mathcal{E}_{\sigma}(t_n, t_0) &\leq \epsilon \, \frac{\sigma_{+}}{\sigma_{-} - \sigma_{+}}
+\frac{2\lambda^2\epsilon }{|\mu|^2} \frac{p(1-p)}{1 - e^{-(\sigma_{-} - \sigma_{+})\tau/2}} \\
&+p^{2}\frac{\lambda^2(\sigma_--\sigma_+)}{|\mu|^2} . \nonumber
\end{align}
The lower bound of (\ref{energ-var-op-syst-0-n}) is also evident. It strongly depends on the initial condition
$N_\sigma(0)$ and can be negative.
\end{remark}
The long-time asymptotic of (\ref{energ-var-op-syst-0-n}), or (\ref{energ-var-op-syst-0-nBIS}), is
\begin{align}\label{energ-var-op-syst-0-inf}
&\Delta \mathcal{E}_{\sigma}:=\lim_{n\rightarrow\infty}\Delta \mathcal{E}_{\sigma}(t_n, t_0)
=\epsilon \, (\frac{\sigma_{+}}{\sigma_{-} - \sigma_{+}}-N_\sigma(0))\\
&+p(1-p)\frac{2\lambda^2\epsilon }{|\mu|^2}
\frac{1-e^{-(\sigma_{-} - \sigma_{+})\tau/2}\cos\epsilon\tau}{1 - e^{-(\sigma_{-} - \sigma_{+})\tau}}\nonumber \ .
\end{align}
From (\ref{energ-var-op-syst-0-inf}) one gets that in the open cavity with $\sigma_{-} - \sigma_{+}>0$
the asymptotic of the total-energy variation is bounded from above and from below
\begin{align}\label{bound+}
&\Delta \mathcal{E}_{\sigma}\leq\epsilon \, (\frac{\sigma_{+}}{\sigma_{-} - \sigma_{+}}-N_\sigma(0))
+ \frac{2\lambda^2\epsilon }{|\mu|^2} \frac{p(1-p)}{1 - e^{-(\sigma_{-} - \sigma_{+})\tau/2}} \ , \\
&\Delta \mathcal{E}_{\sigma}\geq\epsilon \, (\frac{\sigma_{+}}{\sigma_{-} - \sigma_{+}}-N_\sigma(0))
+ \frac{2\lambda^2\epsilon }{|\mu|^2} \frac{p(1-p)}{1 + e^{-(\sigma_{-} - \sigma_{+})\tau/2}} \ .
\label{bound-}
\end{align}

For the short-time regime $n\tau\ll 1$ one gets for (\ref{energ-var-op-syst-0-n})
\begin{align}\label{energ-var-op-syst-0-short}
&\Delta \mathcal{E}_{\sigma}(t_n, t_0) = n \tau \epsilon \, \sigma_{+} - n\tau (\sigma_{-} - \sigma_{+})N_\sigma(0)\\
&+n\tau p(1-p)\frac{2\lambda^2\epsilon }{|\mu|^2}(1-\cos\epsilon\tau)
\frac{\sigma_{-} - \sigma_{+}}{1 - e^{-(\sigma_{-} - \sigma_{+})\tau}} \nonumber \\
&+n\tau p^{2}\frac{\lambda^2(\sigma_--\sigma_+)\epsilon}{|\mu|^2} + \mathcal{O}((n\tau)^2)
\nonumber \ ,
\end{align}
i.e. a linear asymptotic behaviour.

Another asymptotics one finds for the small difference between leaking and external pumping:
$\sigma_{-} - \sigma_{+} \rightarrow 0$. Then (\ref{energ-var-op-syst-0-n}) yields linear behaviour
\begin{align}\label{energ-var-op-syst-0-DIFF}
\Delta \mathcal{E}_{\sigma}(t_n, t_0) = n \tau \epsilon \, \sigma_{+}
+n p(1-p)\frac{2\lambda^2}{\epsilon}(1-\cos\epsilon\tau) \ ,
\end{align}
which is a \textit{growing} of the total energy due to the both external and atomic beam pumping. Note that
in this limit the energy variation $\Delta \mathcal{E}_{\sigma}(t_n, t_0)$ is not bounded from above
(\ref{energ-var-op-syst-0-n+}), (\ref{bound+}).

This case coincides with result for the ideal cavity (Theorem \ref{ThDEn-1,n1}) when the rate of environmental
pumping $\sigma_{+} =0$.

\subsection{Entropy production in the ideal cavity.}

One of the central quantities to study in the  non-equilibrium statistical mechanics is the
entropy production (or the entropy production rate). We refer to the series of
papers  \cite{BJM06}-\cite{BJM10} by Bruneau, Joye, and Merkli, for  a detailed
discussion of this quantity in the context of open quantum systems with repeated
interactions, and we shall adopt definitions of these authors.

Let $\rho$ and $\rho_0$ be two normal states on the algebra $\mathfrak{A}(\mathcal{H})$.
We define the \textit{relative entropy}  ${\rm{Ent}}(\rho|\rho_0)$ of the state $\rho$ with respect to
a "reference" state $\rho_0$ by
\begin{equation}\label{Rel-Ent}
{\rm{Ent}}(\rho|\rho_0):= \Tr(\rho \ln \rho - \rho\ln \rho_0) \geq 0 \,.
\end{equation}
The non-negativity follows from the Jensen inequality: $\Tr(\rho \ln{B})
\leq \ln \Tr(\rho{B})$, applied to observable ${B}:= \rho_0/\rho$.

Here we calculate the entropy production:
\begin{equation}\label{1DS-sigma=0}
\Delta S(t) := {\rm{Ent}}(\rho_\mathcal{S}(t)|\rho_{\mathcal{S}}^{ref})
- {\rm{Ent}}(\rho_\mathcal{S}(0)|\rho_{\mathcal{S}}^{ref}) \ ,
\end{equation}
for the ideal cavity: $\sigma_-=\sigma_+=0$, with dynamics $\rho_\mathcal{S}: t \mapsto \rho_\mathcal{S}(t)$
(\ref{Sol-Liouv-Eq}), and for a reference state $\rho_{\cC}^{ref}\otimes\rho_{\cA}$.
To make a contact with thermodynamics, we suppose that all atoms of the beam are in the \textit{Gibbs state} with
the temperature $1/\beta$. Formally this can be written as
\begin{equation}\label{Atom-Gibbs}
\rho_\cA (\beta) : = \bigotimes_{n\geq1} \rho_{\cA_n} (\beta)  \  , \  \ \rho_{\cA_n} (\beta):=
\frac{e^{-\beta H_{\cA_n}}}{Z(\beta)} \ ,
\end{equation}
see (\ref{Ham-Model}). Since
\begin{equation}\label{ref-state0}
\rho_{\mathcal{S}}^{ref} = \rho_{\cC}^{ref}\otimes\rho_\cA = (\rho_{\cC}^{ref}\otimes \idty) (\idty \otimes\rho_\cA) \ ,
\end{equation}
and since for the unitary dynamics (\ref{L-Gen-t})
\begin{equation}\label{ref-state1}
\Tr\{\rho_\mathcal{S}(t) \ln \rho_\mathcal{S}(t)\} = \Tr\{\rho_\mathcal{S}(0)\ln \rho_\mathcal{S}(0)\},
\end{equation}
the relative entropy (\ref{1DS-sigma=0}) is
\begin{align}\label{2DS-sigma=0}
\Delta S(t) :&= \Tr_{\cC} \{[\rho_\cC^{(0)}-\rho_\cC^{(n)}]\ln \rho_{\cC}^{ref}\} \\
&- \beta \ \sum_{k=1}^{n} \Tr\{[\rho_\mathcal{S}(0)-\rho_\mathcal{S}(n\tau)] (\idty \otimes H_{\cA_k})\}  \nonumber \ ,
\end{align}
for $t = n \tau + \nu$, see (\ref{t}). Here $\rho_\cC^{(n)}$ is defined by (\ref{C-state-n}).
\begin{remark}\label{Entr-commut}
For any Hamiltonian $H_n$ that acts non-trivially on $\cH_\cC\otimes\cH_{\cA_n}$ and for any Hamiltonian $H_{\cA_k}$
acting on $\cH_{\cA_k}$ one gets  $[H_n, H_{\cA_k}] = 0$ if $n \neq k$. Note that in our
model (\ref{commut-atoms}) implies $[H_n, H_{\cA_k}] = 0$ for any $n,k$.
\end{remark}
Therefore, by virtue of (\ref{C-state-n}) and (\ref{cL}) we obtain
\begin{align}\label{rho-S-nk}
&\Tr\{\rho_\mathcal{S}(n\tau) \ (\idty \otimes H_{\cA_k})\} \\
& = \Tr\{e^{- i \tau H_n}...e^{-i \tau H_{k+1}} \rho_\mathcal{S}(k\tau)e^{i \tau H_{k+1}} \ldots {e^{i \tau H_n}} \
(\idty \otimes H_{\cA_k})\} \nonumber \\
& = \Tr\{\rho_\mathcal{S}((k\tau) \ (\idty \otimes H_{\cA_k})\}  \nonumber \ .
\end{align}
Then by the same arguments one gets also that
\begin{align}\label{rho-S-0k-1}
&\Tr\{\rho_\mathcal{S}(0) \ (\idty \otimes H_{\cA_k})\} \\
& = \Tr\{e^{- i \tau H_{k-1}}...e^{-i \tau H_1} \rho_\mathcal{S}(0)e^{i \tau H_1} \ldots e^{i \tau H_{k-1}} \
(\idty \otimes H_{\cA_k})\} \nonumber \\
& = \Tr\{\rho_\mathcal{S}((k-1)\tau) \ (\idty \otimes H_{\cA_k})\}  \nonumber \ .
\end{align}

In the case of our model (see Remark \ref{Entr-commut}) $H_{\cA_k}=e^{i\tau H_k}(H_{\cA_k})e^{-i\tau H_k}$.
Therefore the last formula gets the form
\begin{align}\label{rho-S-0k-1-model}
&\Tr\{\rho_\mathcal{S}(0) \ (\idty \otimes H_{\cA_k})\}=\Tr\{\rho_\mathcal{S}(k\tau) \ (\idty \otimes H_{\cA_k}).
\end{align}
Equations (\ref{rho-S-nk}) and (\ref{rho-S-0k-1-model}) shows that the second term in the entropy production
(\ref{2DS-sigma=0}) \textit{vanishes}.

If we suppose that the reference state is the Gibbs state (\ref{Gibbs-photons}) at temperature $1/\beta_{\mathcal{C}}$ ,
then by (\ref{N(t)}) one gets
\begin{align}\label{EntrProd-OurMod}
&\Delta S(t)= \Tr_{\cC} \{[\rho_\cC^{(0)}-\rho_\cC^{(n)}]\ln \rho_{\cC}^{ref}\} \\
& =\Tr_{\cC} \{[\rho_\cC^{(0)}-\rho_\cC^{(n)}] (- \beta_{\mathcal{C}} \, \epsilon \, b^{*}b )\} \nonumber \\
&= \beta_{\mathcal{C}}\, \epsilon(N(t)-N(0))\ , \nonumber
\end{align}
where $N(t)$ is the mean photon number defined by (\ref{N(t)}).

Let us define by $\Delta \mathcal{E}^{\cC}(t) = \epsilon (N(t) - N(0))$ the energy variation of the cavity
due to the photon number evolution. Then (\ref{EntrProd-OurMod}) expresses the 2nd Law of Thermodynamics
\begin{equation}\label{2nd-Law}
\Delta S(t) = \beta_{\mathcal{C}} \, \Delta \mathcal{E}^{\cC}(t) \ ,
\end{equation}
for the pumping by atomic beam. Note that relation (\ref{2nd-Law}) does not depend on the initial cavity
state $\rho_\cC^{(0)}$.
\begin{remark}
In general, when $[H_n, H_{\cA_n}]\neq 0$, the combination of (\ref{2DS-sigma=0}) with (\ref{rho-S-nk}) and
(\ref{rho-S-0k-1}) yield for the entropy production at the moment
$t = n \tau + \nu$ the expression:
\begin{align}\label{DS-sigma=0-bis}
\Delta S(t) :&= \Tr_{\cC} \{[\rho_\cC^{(0)}-\rho_\cC^{(n)}]\ln \rho_{\cC}^{ref}\} \\
&+ \beta \ \sum_{k=1}^{n} \Tr\{[\rho_\mathcal{S}(k \tau) - \rho_\mathcal{S}((k-1)\tau)] \ (\idty \otimes H_{\cA_k})\}
\nonumber \ .
\end{align}
The last term in (\ref{DS-sigma=0-bis}) can be rewritten into the standard form \cite{BJM06}-\cite{BJM10},
if one uses the identities:
\begin{align*}\label{last-term1}
&\Tr\{\rho_\mathcal{S}(k\tau) (\idty \otimes H_{\cA_k})\} =
\Tr\{e^{\tau L_k}\! \ldots e^{\tau L_1}(\rho_{\cC}\otimes \rho_{\cA})\ (\idty \otimes H_{\cA_k})\} \\
&=\Tr\{e^{- i \tau H_{k}}(e^{\tau L_{k-1}}\! \ldots e^{\tau L_1}[\rho_{\cC}\otimes\!
\bigotimes_{n=1}^{k-1}\rho_{n}])
\otimes \rho_{k} e^{i \tau H_{k}} \ (\idty \otimes H_{\cA_k})\} [\idty \otimes\! \bigotimes_{m>k}\! \rho_{m}]\nonumber \\
&=\Tr \{(\rho_{\cC}^{(k-1)} \otimes \rho_{k}) e^{i \tau H_{k}} \ (\idty \otimes H_{\cA_k})
e^{- i \tau H_{k}}\}
\nonumber  \ ,
\end{align*}
and
\begin{equation*}\label{last-term2}
\Tr\{\rho_\mathcal{S}((k-1)\tau) \ (\idty \otimes H_{\cA_k})\} =
\Tr \{\rho_{\cC}^{(k-1)} \otimes \rho_{k} H_{\cA_k}\} \ .
\end{equation*}
For $t = n \tau + \nu$ this gives the formula for the entropy production in the ideal cavity:
\begin{align}\label{DS-sigma=0-fin}
&\Delta S(t)= \Tr_{\cC} \{[\rho_\cC^{(0)}-\rho_\cC^{(n)}]\ln \rho_{\cC}^{ref}\} \\
&+\beta \ \sum_{k=1}^{n} \Tr_{\cH_\cC \otimes \cH_{\cA_{k}}} \{(\rho_{\cC}^{(k-1)} \otimes \rho_{k})
[e^{i \tau H_{k}} (\idty \otimes H_{\cA_k})e^{- i \tau H_{k}} - \idty \otimes H_{\cA_k}]\}  \nonumber  \ .
\end{align}

One expects that for an open cavity the gain and loss of photons would result in a corresponding additional
flux of entropy due to the quantum Markov evolution. It is not yet completely clear (see e.g.\cite{FaRe}) how to define this entropy production correctly and whether contributions of these two processes are independent.
Therefore, we leave the analysis of the entropy flux for the open cavity to be considered in the future.
\end{remark}

\section{Concluding Remarks}\label{CR}

{From} Theorem \ref{N of photons} we learn that in the ideal cavity the photon number
increases linearly in time up to bounded oscillations, (see Figure \ref{assymptotics}, green curve)
The rate of the growth is non-zero for any value of the probability of excited atoms in the beam except
$p=0$ and $p=1$. For large time it is independent of the initial state $\rho_{\mathcal{C}}$.
The rate of the linear growth of the mean-value of photons $N(t)$ with respect to the time $t=n\tau$ can
be seen from (\ref{number of photons_Hamiltonian}) since
\begin{equation}\label{rate_growth}
\frac{N(t)}{n\tau}=\frac{N(0)}{n\tau}+  \, p(1-p) \, \frac{2\lambda^2}{\tau\epsilon^2} \, (1-\cos\epsilon\tau) +
\frac{p^2}{n\tau} \, \frac{2\lambda^2}{\epsilon^2}(1-\cos n\epsilon\tau)\ .
\end{equation}
Hence, for $N(0)=0$ one gets linear growth modulo bounded oscillations:
\begin{equation}\label{rate_growthBIS}
\frac{N(t)}{n\tau} = p(1-p) \, \frac{2\lambda^2}{\tau\epsilon^2} \, (1-\cos\epsilon\tau)
+ \frac{p^2}{n\tau} \, \frac{2\lambda^2}{\epsilon^2}(1-\cos n\epsilon\tau)\ .
\end{equation}
This linear growth of the photon number is observed in experiments with one-atom masers
for the high-quality resonators (nearly ideal cavities)  \cite{MWM}.

\begin{figure}[t]
\includegraphics[width=0.7\textwidth]{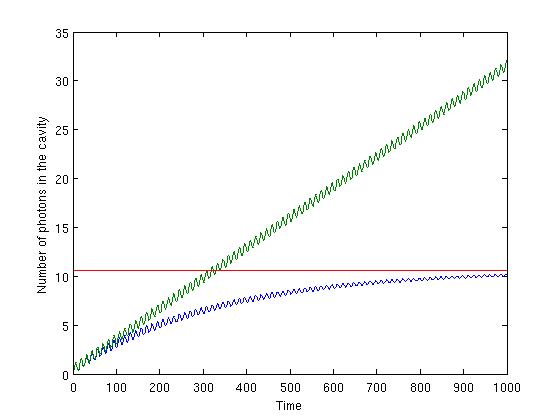}
\caption{In blue, we have plotted the photon number in an open cavity with parameters
 $\sigma_-=0.003$, $\sigma_+=0$, $\epsilon=0.5$, $\tau=0.5$, $p=\frac{1}{2}$,
 $\frac{\lambda^2}{|\mu|^2}=1$ and vanishing initial number of photons $N_{\sigma}(0)=0$.
The green curve is the mean-value of photons in the ideal cavity for the same parameters
and initial condition $N(0)=0$, except $\sigma_-=0$. The expressions for these quantities
are given by (\ref{number of photons_Hamiltonian}) for the ideal cavity and by
(\ref{mean-value}) for the open cavity.
The red line is the asymptotic value (\ref{lim-photons-number-sigma}) of the mean photon
number in the open cavity.}
\label{assymptotics}
\end{figure}
As can be clearly seen in Figure \ref{assymptotics} (blue curve),
the mean photon number $N_{\sigma}(t)$ in the open cavity (\ref{mean-value}) is also initially
increasing linearly up to bounded oscillations, but then it stabilizes for large times.
Similarly to (\ref{rate_growth}) the rate of the growth in the open cavity with respect to the time $t=n\tau$
can be deduced from
\begin{align}\label{mean-value-rate}
&\frac{N_\sigma(n\tau)}{n\tau} =e^{-n(\sigma_--\sigma_+)\tau}\frac{N_\sigma(0)}{n\tau}\\
&+p(1-p)\frac{2\lambda^2}{|\mu|^2}(1-e^{-(\sigma_{-} - \sigma_{+})\tau/2}\cos \epsilon\tau)
\frac{1-e^{-n(\sigma_--\sigma_+)\tau}}{n\tau(1-e^{-(\sigma_--\sigma_+)\tau})}\nonumber\\
&+\frac{p^2}{n\tau}\frac{2\lambda^2}{|\mu|^2}(1-e^{-n(\sigma_{-} - \sigma_{+})\tau/2}\cos n\epsilon\tau)
+\frac{\sigma_+}{\sigma_--\sigma_+}\frac{1-e^{-n(\sigma_--\sigma_+)\tau}}{n\tau}\nonumber \ .
\end{align}
which for $\sigma_+ =0$ and then $\sigma_- =0$ gives (\ref{rate_growth}).

Let $N_\sigma(0)=0$ and $\sigma_+ =0$. Since $\mu =(\sigma_--\sigma_+)/2 + i \epsilon $, the short-time
behaviour of (\ref{mean-value-rate}) for $n(\sigma_--\sigma_+)\tau\ll 1$ is linear (modulo bounded oscillations):
\begin{align}\label{mean-value-rateBIS}
\frac{N_\sigma(n\tau)}{n\tau} &= p(1-p)\frac{2\lambda^2}{\tau\epsilon^2}(1-\cos \epsilon\tau) +
\frac{p^2}{n\tau} \, \frac{2\lambda^2}{\epsilon^2}(1-\cos n\epsilon\tau) + \\
&+\mathcal{O}((n(\sigma_--\sigma_+)\tau)^{2}) \nonumber \ ,
\end{align}
and asymptotically it is close to (\ref{rate_growthBIS}). This is clearly visible in the
Figure \ref{assymptotics} for $n< 100$.

There are several generalisations of the beam-cavity problem considered in this paper
that could be handled by suitable modifications of our methods. First, there is the
detuned case, when the distance between atoms is greater than the length of the cavity, $d>l$.
This is a situation where there is still at most one atom interacting with the cavity at
any given time, but there are time intervals without an atom present. The modifications
needed to analyze this situation are straightforward.
A more interesting generalization would be to consider random interatomic distances
$d\geq l$, when again still there is at most one atom in the cavity, but they arrive randomly.

Due to properties of the apparatus that produces the atom beam, one may expect short range
correlations in the chain of atoms. Such a correlated atomic beam can be described by a
classical Markov chain or, more generally, by a so-called Finitely Correlated
State \cite{FNW92}. Calculating the asymptotic behavior of the cavity in this situation
will be a bit more complicated but should still be doable.

We would like also to mention the following two other problems. The first is to consider
a cavity interaction with `soft atoms'. One of the possible interactions is of the form
$b^*\otimes \sigma^{-} + b \otimes \sigma^{+}$, i.e., interaction of the Jaynes-Cummings type.
Here in contrast to (\ref{W-int}) the atomic operators $\sigma^{\pm}:= (\sigma^{x} \pm i \sigma^{y})/2$
are \textit{off-diagonal}, constructed with the Pauli matrices $\sigma^{x}$ and $\sigma^{y}$.
For $p<1/2$ and for an ideal cavity this problem has been studied in \cite{BPi09}. It was shown that the
Jaynes-Cummings interaction allows the cavity state to converge to some explicit thermal state. So, the
pumping is saturated and the photon number expectation is well-defined and bounded.

The second problem is to calculate the entropy production for the Kossakowski-Lindblad open cavity and to determine
its relation to the energy flux in the frame work of our model. For a discussion of the entropy production problem 
for the quantum Markov semigroups see, e.g., the recent paper \cite{FaRe}.

\subsection*{Acknowledgments}
Our interest in open systems with repeated interactions was sparked by a nice series of lectures given by
Laurent Bruneau and Alain Joye on the Summer school on ``Non-equilibrium statistical mechanics"  (July 1-29, 2011)
CRM - Universit\'{e} de Montr\'{e}al. We would like to thank Alain Joye, Marco Merkli and Claude-Alain Pillet for
useful discussions of different aspects of the mathematical theory of the open systems.
Valentin Zagrebnov is thankful to the University of California - Davis for the warm hospitality
during his multiple visits, which allowed the authors to complete this project.

We are grateful to the referee for pointing out a few erroneous reasonings in the first version of our manuscript
as well as for useful remarks and suggestions. This friendly criticism motivated us to significantly revise the original manuscript.


\newpage

\section*{Appendix: Open One-Mode Photon Cavity}\label{Appendix}

\noindent Here we collect some remarks and recall certain statements concerning the quantum theory of open systems, see
e.g. \cite{Dav1,AlFa,AJPII}. To this end we consider example of the open one-mode photon cavity
$\mathcal{C}$. This soluble example is useful for our discussion of the leaky cavity pumped by atomic beam starting on
Section \ref{IQD}.

\smallskip
\noindent \textbf{A.1 Markovian Master Equation.}\label{A1}
We treat the cavity: $H_\mathcal{C} = \epsilon \, b^* b$, interacting with external reservoir $\mathcal{R}$ in the
framework of the Markovian approach \cite{AJPII,AJPIII}. Then $\mathcal{R}$ is a source of a \textit{leaking},
which decreases the cavity energy with the rate $\sigma_{-}$, or/and  of a \textit{pumping} with the rate $\sigma_{+}$.
The corresponding Markovian master equation for evolution of the normal cavity states $\omega^{t}_\mathcal{C}(\cdot)$
with trace-class density matrices $\rho_{\mathcal{C}}(t)\in \mathfrak{C}_{1}(\mathcal{H}_\mathcal{C})$ is extension
of the Hamiltonian dynamics by Kossakowski-Lindblad damping (leaking) and pumping terms \cite{Al,AlFa}:
\begin{align}\label{Damp-Pump-Cav}
&\frac{d}{dt} \rho_{\mathcal{C}}(t) = L_{\mathcal{C},\sigma} (\rho_{\mathcal{C}}(t)) :=
- i [H_\mathcal{C} , \rho_{\mathcal{C}}(t)]  \\
&+\frac{1}{2} \sigma_{-} ([b \rho_{\mathcal{C}}(t), b^*] + [b, \rho_{\mathcal{C}}(t) b^*]) +
\frac{1}{2} \sigma_{+} ([b^* \rho_{\mathcal{C}}(t), b] + [b^*, \rho_{\mathcal{C}}(t) b]) \ . \nonumber
\end{align}
Here $\rho_{\mathcal{C}}(t)\in {\rm{dom}}(L_{\mathcal{C},\sigma})$ and parameters $\sigma_{-}, \sigma_{+} \geq 0$.

Note that evolution (\ref{Damp-Pump-Cav}) for $\rho_{\mathcal{C}}:=\rho_{\mathcal{C}}(0)\in \mathfrak{C}_{1}
(\mathcal{H}_\mathcal{C})$:
\begin{equation}\label{Schr-Evol}
\mathcal{L}^{t}_{\mathcal{C},\sigma}: \rho_{\mathcal{C}} \mapsto \rho_{\mathcal{C}}(t) \ ,
\end{equation}
is the case of \textit{trace-norm} continuous semigroup: $\mathcal{L}^{t}_{\mathcal{C},\sigma}:=
e^{t \, L_{\mathcal{C},\sigma}}$ on the Banach space $\mathfrak{C}_{1}(\mathcal{H}_\mathcal{C})$, with \textit{unbounded}
generator $L_{\mathcal{C},\sigma}$, see e.g. \cite{Za}, Ch.2.4.

It is known that this case of quantum Markovian dynamics (\ref{Schr-Evol}) needs a special care, see e.g. \cite{Dav1},
\cite{Si,ChFa}, but for the open one-mode photon cavity $\mathcal{C}$ all necessary properties can be checked
explicitly.

Denote by $\widehat{L}_{\mathcal{C},\sigma}$ the {non-Hamiltonian} part of unbounded generator
(\ref{Damp-Pump-Cav}). Then for any $\rho \in \mathfrak{C}_{1}(\mathcal{H}_\mathcal{C})$ one has:
\begin{equation}\label{non-Ham-gen1}
\widehat{L}_{\mathcal{C},\sigma} (\rho) = \sum_{\alpha = \downarrow, \uparrow} \sigma_{\alpha} \{V_{\alpha} \rho
V_{\alpha}^{*} - \frac{1}{2} (V_{\alpha}^{*}V_{\alpha} \rho + \rho V_{\alpha}^{*}V_{\alpha})\} \ , \ V_{\downarrow} =
b \ , \ V_{\uparrow} =b^* \ .
\end{equation}
By virtue of the trace cyclicity this canonical form of the Kossakowski-Lindblad generator:
\begin{equation}\label{K-L-gen}
L_{\mathcal{C},\sigma} := - i \, [H_\mathcal{C} , \cdot] + \widehat{L}_{\mathcal{C},\sigma} \ \ ,
\end{equation}
ensures the \textit{trace-preserving} property of dynamics (\ref{Damp-Pump-Cav}):
\begin{equation}\label{Tr-inv}
\frac{d}{dt} \Tr_{{\mathcal{C}}} \rho_{\mathcal{C}}(t) = 0.
\end{equation}
Note that by virtue of (\ref{Damp-Pump-Cav}) the Markovian dynamics (\ref{Schr-Evol}) is also  {unity-preserving}:
$\mathcal{L}^{t}_{\mathcal{C},\sigma}(\idty) = \idty$ for $t\geq 0$.

To check another important property: the (complete) \textit{positivity} of the trace-norm continuous semigroup
$\{\mathcal{L}^{t}_{\mathcal{C},\sigma}\}_{t\geq0}$ on the space $\mathfrak{C}_{1}(\mathcal{H}_\mathcal{C})$, let us
present its generator (\ref{K-L-gen}) as $L_{\mathcal{C},\sigma} := \Phi - \Gamma$, where
\begin{align}\label{gen-Damp-Pump1}
&\Phi (\rho) := \sum_{\alpha = \downarrow, \uparrow} \sigma_{\alpha} \, V_{\alpha} \, \rho \, V_{\alpha}^{*} \  \ , \ \
\ \ \sigma_{\alpha} \geq 0 \ \ , \\
&\Gamma (\rho) := \Psi \, \rho + \rho \, \Psi^* \ \ \ {\rm{with}} \ \  \ \Psi := i \, H_\mathcal{C} +
\frac{1}{2} \sum_{\alpha = \downarrow, \uparrow} \sigma_{\alpha} \, V_{\alpha} \, V_{\alpha}^{*} \ . \label{gen-Damp-Pump2}
\end{align}
First we reduce our analysis to  {positivity} and we postpone the question concerning  {complete} positivity to
the end of this section and to the  {Heisenberg picture} of quantum dynamics (\ref{Schr-Evol}), see subsection A.2.

To see that dynamical semigroup $\{e^{t \, \Phi}\}_{t\geq0}$ with generator (\ref{gen-Damp-Pump1}) enjoy the property
of the
positivity, notice that by (\ref{gen-Damp-Pump1}) one gets for the trace-continuous maps
$\rho \mapsto \rho_{\Phi}(t) := e^{t \, \Phi} (\rho)$:
\begin{equation}\label{Comp-Posit1}
\frac{d}{dt} \rho_{\Phi}(t) = \Phi (\rho_{\Phi} (t)) =
\sum_{\alpha = \downarrow, \uparrow} \sigma_{\alpha} \, V_{\alpha} \, \rho_{\Phi} (t) \, V_{\alpha}^{*} \ .
\end{equation}
Let $\rho \in {\rm{dom}}(\Phi)\subset \mathfrak{C}_{1}(\mathcal{H}_\mathcal{C})$ and  $\rho \geq 0$. Then
(\ref{gen-Damp-Pump1}) implies that $\Phi (\rho) \geq 0$ and that equation (\ref{Comp-Posit1}) is positivity-preserving.
This yields positivity of the solution $\rho_{\Phi} (t)$ for $t\geq0$, if $\rho_{\Phi} (t=0)= \rho$.

To see that semigroup $\{e^{- t \, \Gamma}\}_{t\geq0}$ is also a family of positive mapping on
$\mathfrak{C}_{1}(\mathcal{H}_\mathcal{C})$ note that (\ref{gen-Damp-Pump2}) yields
\begin{equation}\label{Comp-Posit2}
\frac{d}{dt} \rho_{\Gamma}(t) = - \Gamma (\rho_{\Gamma}(t)) = - (\Psi \, \rho_{\Gamma}(t) + \rho_{\Gamma}(t) \, \Psi^* ) =
\frac{d}{dt} \left(e^{- t \, \Psi}\, \rho \, e^{- t \, \Psi^*} \right) \ .
\end{equation}
Then the mapping  $e^{- t \, \Gamma}: \rho \mapsto e^{- t \, \Psi} \rho \, e^{- t \, \Psi^*}$ is positive. For
$\rho \in \mathfrak{C}_{1}(\mathcal{H}_\mathcal{C})$ we denote by
$\rho_{\Gamma}(t):= e^{- t \, \Gamma} (\rho)$ the solution of (\ref{Comp-Posit2}). This operator is positive for
$\rho \geq 0$.

By virtue of (\ref{Comp-Posit1}) and (\ref{Comp-Posit2}) the composition of two maps:
$F(t): \rho \mapsto e^{t \, \Phi} (e^{- t \, \Gamma} (\rho))$, is a positive trace-norm continuous mapping.
Note that the powers of $F(t)$ are also positive maps. Then it is also true for dynamical semigroup with generator
(\ref{K-L-gen}), since by the Trotter product formula one gets
\begin{equation}\label{Trott1}
\mathcal{L}^{t}_{\mathcal{C},\sigma}= e^{t \, L_{\mathcal{C},\sigma}} = \|\cdot\|_{1}-
\lim_{n \rightarrow\infty} \left( e^{t \, \Phi/n} \ e^{- t \, \Gamma/n}\right)^n \ ,
\end{equation}
in the trace-norm topology \cite{Za}. This remark ensures, in particular, the $\mathfrak{C}_{1}$-continuity of
the limit (\ref{Trott1}).
Note also that by (\ref{Damp-Pump-Cav}) and (\ref{Schr-Evol}) the limit (\ref{Trott1}) is unity-preserving
(or {Markov}) semigroup.

Recall that essential in the concept of  {open systems} is a coupling of some  {small} sub-system with
"environment", which is a certain  {large} (even infinite) system. The mathematical description of this concept
involves a tensor product of corresponding configurations spaces and states or algebras of observables. Then the
positivity-preserving evolution due to the master equation for a coupled system, must be robust for forming the tensor
products \cite{Dav1}. Since tensor product of two positive maps might fail to be positive (\cite{AlFa}, Ch.8.3), the
evolution (\ref{Schr-Evol}) has to verify a stronger condition than positivity established in (\ref{Trott1}).

We recall now the  {complete} positivity property and constrains that it implies on dynamical semigroup
(\ref{Schr-Evol}) and on its generator \cite{Dav1,AlFa,AJPII}.

Let $T: \mathfrak{A}^{(1)} \rightarrow \mathfrak{A}^{(2)}$ be a  {positive}
linear map between two $C^*$-algebras, i.e. $T(A)\geq 0$, where $A \in \mathfrak{A}^{(1)}$ and $A\geq0$.
For $k= 1,2$ and for $n\in\mathbb{N}$, let
\begin{equation}\label{A-n}
\mathcal{M}_{n}(\mathfrak{A}^{(k)}) \simeq \mathfrak{A}^{(k)} \otimes \mathcal{M}(\mathbb{C}^n)  \ ,
\end{equation}
be algebra of $n \times n$ matrices with entries in $\mathfrak{A}^{(k)}$. Each of
$\mathcal{M}_{n}(\mathfrak{A}^{(k)})$ is also a $C^*$-algebra. If we denote by  ${\rm{Id}}_{n}$ the {identity}
matrix from $\mathcal{M}(\mathbb{C}^n)$, then
\begin{equation}\label{compl-posit}
T_{n}:= T\otimes {\rm{Id}}_{n}: \mathcal{M}_{n}(\mathfrak{A}^{(1)}) \rightarrow
\mathcal{M}_{n}(\mathfrak{A}^{(2)})\ ,
\end{equation}
defines a linear map by acting with $T$ on  {each} of the matrix element of the operator-valued matrix
${A}_{n}\in \mathcal{M}_{n}(\mathfrak{A}^{(1)})$. The positive map $T$ is called $n$-positive (respectively
 {completely} positive), if operator $T_{n}$ is positive (respectively (\ref{compl-posit}) are positive for all
$n\geq1$). For $n=1$ it obviously reduces to the positive map.
\begin{example}\label{CP-example}
A simple example shows that this property is quite non-trivial. (For more of them we refer to \cite{AlFa}.) Let
$\mathfrak{A}^{(k =1,2)} = \mathcal{M}(\mathbb{C}^2)$ be $C^*$-algebra of square complex matrices. Then the map
of matrices
$\mathcal{M}(\mathbb{C}^2)$ to the \textit{adjoint}, $T_{adj}: A \rightarrow A^* $ is obviously positive. Denote by
$\{E_{ij} \in \mathcal{M}(\mathbb{C}^2)\}_{i,j = 1,2}$ the set of matrices with $1$ in the $ij$-th entry and zeros
elsewhere, i.e. ${\rm{Id}}_{2} = \sum_{i,j = 1,2} E_{ij}$. To verify whether the map $T_{adj}$ is $2$-positive we
consider the algebra (\ref{A-n}):
\begin{equation}\label{A-2}
\mathcal{M}_{2}(\mathcal{M}(\mathbb{C}^2)) \simeq \mathcal{M}(\mathbb{C}^2)\otimes \mathcal{M}(\mathbb{C}^2)
\simeq  \mathcal{M}(\mathbb{C}^4)\ ,
\end{equation}
and the element $E:= \sum_{i,j = 1,2} E_{ij} \bigotimes E_{ij}\in \mathcal{M}(\mathbb{C}^4)$. Since $E =E^*$ and
$E^2 = 2 E$, it is positive $E \geq 0$. On the other hand by definition (\ref{compl-posit}) one gets
\begin{equation}\label{T-2}
T_{{adj},2}:= T_{adj}\otimes {\rm{Id}}_{2}: E \rightarrow \begin{bmatrix}
                                                T_{adj}(E_{11}) & T_{adj}(E_{12}) \\
                                                T_{adj}(E_{21}) & T_{adj}(E_{22}) \\
                                              \end{bmatrix}
= \begin{bmatrix}
    1 & 0 & 0 & 0 \\
    0 & 0 & 1 & 0 \\
    0 & 1 & 0 & 0 \\
    0 & 0 & 0 & 1 \\
  \end{bmatrix}
\ .
\end{equation}
The matrix $T_{{adj},2}(E)$ in (\ref{T-2}) is \textit{not} positive, since its spectrum contains ($-1$). Therefore,
the map $T_{adj}$ is \textit{not} $2$-positive either.
\end{example}

One of the important corollary imposed by demand that quantum Markovian dynamics on the Banach space
$\mathfrak{C}_{1}(\mathcal{H})$:
\begin{equation}\label{QMD}
\mathcal{L}^{t} = e^{t \, L} : \rho \mapsto \rho (t) \ , \  \
\rho (0)= \rho \in \mathfrak{C}_{1}(\mathcal{H}) \ ,
\end{equation}
must be \textit{completely} positive, is certain restrictions on the form of the generator $L$, cf.
(\ref{Schr-Evol}) and (\ref{K-L-gen}). For the case of bounded generator, when semigroup (\ref{QMD}) is
continuous in the operator-norm topology of mappings on $\mathfrak{C}_{1}(\mathcal{H})$, this is just a
celebrated Kossakowski-Lindblad result saying that the most general form of $L$ is
\begin{equation}\label{K-L-gen-I}
L:= - i \, [H , \cdot] + \widehat{L} \ ,
\end{equation}
where non-Hamiltonian part can be presented as
\begin{align}\label{K-L-gen-II}
\widehat{L}(\rho):&=\frac{1}{2}\sum_{\alpha}  ( [\widehat{V}_{\alpha} \rho \, , \, \widehat{V}_{\alpha}^{*}] +
[\widehat{V}_{\alpha}\,, \rho \, \widehat{V}_{\alpha}^{*}] ) \\
&=\sum_{\alpha} \{\widehat{V}_{\alpha} \rho \widehat{V}_{\alpha}^{*} -
\frac{1}{2} (\widehat{V}_{\alpha}^{*}\widehat{V}_{\alpha} \, \rho + \rho \,
\widehat{V}_{\alpha}^{*}\widehat{V}_{\alpha})\} \ ,
\nonumber
\end{align}
see, e.g., \cite{AlFa}, Ch.8. Note that the choice of bounded operators $H = H^*$ and
$\{\widehat{V}_{\alpha}\}_{\alpha}$ in representation (\ref{K-L-gen-II}) is not unique.

The contact between (\ref{K-L-gen-I}), (\ref{K-L-gen-II}) and representations  (\ref{non-Ham-gen1}),(\ref{K-L-gen})
as well as with (\ref{gen-Damp-Pump1}), (\ref{gen-Damp-Pump2}) related to semigroups $\{e^{t \, \Phi}\}_{t\geq0}$,
$\{e^{- t \, \Gamma}\}_{t\geq0}$ follows through verbatim. Then taking into account that
the Stinespring theorem (\cite{AlFa}, Ch.8) implies the complete positivity of these two semigroups and using the
Trotter product formula one can check the complete positivity of (\ref{QMD}).

The present results for unbounded $H$ and/or $\widehat{L}$ are more poor. Certain classes of
strongly continuous on $\mathfrak{C}_{1}(\mathcal{H})$, contraction semigroups (\ref{QMD}) have been studied
in \cite{Dav1} and \cite{EvLe1,EvLe2} under condition that operator
$\sum_{\alpha} \widehat{V}_{\alpha}^{*}\widehat{V}_{\alpha}$ is generator of a strongly continuous contraction
semigroup on $\mathcal{H}$. For further developments see, e.g., \cite{Dav2,Si,Hol,ChFa}.

The Kossakowski-Lindblad representation (\ref{K-L-gen-I}), (\ref{K-L-gen-II}) explains our choice of the right-hand side
in the Markovian Master Equation (\ref{Damp-Pump-Cav}), but appeals for a concrete  verification of the
complete positivity of quantum Markovian dynamics (\ref{Schr-Evol}). Similar to other known cases of the
Kossakowski-Lindblad type generators \cite{FrVe}, this property of $\mathcal{L}^{t}_{\mathcal{C},\sigma}$ follows
directly from explicit calculations.

\smallskip
\noindent \textbf{A.2 Dual Dynamical Map (Heisenberg Picture).}
The equivalent (and often more convenient) is the abstract version of the reduced Markovian dynamics A.1 on the
algebra of observables $\mathfrak{A}(\mathcal{H}_\mathcal{C})$, i.e. the {quantum dynamical semigroup}
$\mathcal{L}^{t \,*}_{\mathcal{C},\sigma}:=(\mathcal{L}^{t}_{\mathcal{C},\sigma})^{*}$  in the \textit{dual} Heisenberg
picture \cite{AlFa,AJPI}.

Recall that in a general setting the $C^*$-dynamical system is a pair $(\mathfrak{A}, \tau^t)$ with two
properties. First, $\mathfrak{A}$ is a unital $C^*$-algebra, i.e. $\idty \in \mathfrak{A}$. Second, $\tau^t$ is a
 {strongly} continuous one-parameter $*$-automorphism of $\mathfrak{A}$, i.e. for any $A \in \mathfrak{A}$ the map:
$t \mapsto \tau^t(A)$, is continuous in the {norm} topology of $\mathfrak{A}$.

Since quantum states belong to trace-class $\mathfrak{C}_{1}(\mathcal{H})$, which is not a unital $C^*$-algebra, this
framework is not satisfactory for Markovian dynamics (\ref{QMD}) on the Banach space $\mathfrak{C}_{1}(\mathcal{H})$,
although (\ref{QMD}) is a strongly continuous $*$-automorphism even for unbounded generator $L$.

To define a dynamics, which is {dual} to the Schr\"{o}dinger picture (\ref{QMD}), we recall that semigroup
$\{\mathcal{L}^{t}\}_{t\geq0}$ serves to calculate evolution of the \textit{normal} states
$\{\omega^{t}(\cdot)\}_{t\geq0}$ on observables  $A \in \mathfrak{B}(\mathcal{H})$:
\begin{equation}\label{Heis-Evol-abstr}
\omega^{t}(A) = \Tr_{\cH}(\mathcal{L}^{t} (\rho) \, A)=: \langle \mathcal{L}^{t} (\rho), A \rangle
 \ , \  \ \rho \in \mathfrak{C}_{1}(\mathcal{H}) \ .
\end{equation}
Here $\mathfrak{B}(\mathcal{H})$ denote the $C^*$-algebra of bounded operators on the Hilbert space $\mathcal{H}$ and
$\rho$ is a density-matrix operator with the trace-norm $\|\rho\|_{\mathfrak{C}_{1}} =1$.

Recall that by virtue of (\ref{Heis-Evol-abstr}) the Banach space of bounded operators on $\mathcal{H}$ is topologically
{dual} of $\mathfrak{C}_{1}(\mathcal{H})$: $\mathfrak{B}(\mathcal{H})= (\mathfrak{C}_{1}(\mathcal{H}))^*$.
This means that the map $A \mapsto \langle \cdot , A\rangle$ is an isometric isomorphism of $\mathfrak{B}(\mathcal{H})$
onto the set of linear continuous functionals $(\mathfrak{C}_{1}(\mathcal{H}))^*$, defined on the space
$\mathfrak{C}_{1}(\mathcal{H})$. Then semi-norms generated by this duality
\begin{equation}\label{semi-norms}
\{\mathfrak{N}_{\rho} (A) : = |\langle \rho , A \rangle|\}_{\rho \in \mathfrak{C}_{1}(\mathcal{H})} \ ,
\ \  A \in \mathfrak{B}(\mathcal{H}) \ ,
\end{equation}
define on the Banach space $\mathfrak{B}(\mathcal{H})$ the weak$^*$-topology, which coincides with the  {operator}
$\sigma$-weak topology, and for the operator norm of $A$ one gets:
\begin{equation}\label{op-norm}
\|A\| = \sup_{\rho \in \mathfrak{C}_{1}(\mathcal{H})} \frac{|\langle \rho , A \rangle|}{\|\rho\|_{\mathfrak{C}_{1}}} \ .
\end{equation}
see e.g. \cite{ReSiI} or \cite{AJPI}. Duality (\ref{Heis-Evol-abstr}) defines also the \textit{adjoint} semigroup
$\{\mathcal{L}^{t \,*}\}_{t\geq0}$ on the dual space $(\mathfrak{C}_{1}(\mathcal{H}))^* = \mathfrak{B}(\mathcal{H})$:
\begin{equation}\label{dual-semigroup}
\langle \mathcal{L}^{t} (\rho) , A \rangle = \langle \rho , \mathcal{L}^{t \,*}(A) \rangle \ , \ \
\rho \in \mathfrak{C}_{1}(\mathcal{H})\ , \   A \in (\mathfrak{C}_{1}(\mathcal{H}))^* \ .
\end{equation}
In general, the adjoint semigroup is  {not} strongly continuous on the dual Banach space
$(\mathfrak{C}_{1}(\mathcal{H}))^*$, although (\ref{dual-semigroup}) and the strong continuity of (\ref{QMD}) trivially
imply the weak$^*$-continuity of $\{\mathcal{L}^{t \,*}\}_{t\geq0}$ on this space.

Even thought $\{\mathcal{L}^{t \,*}\}_{t\geq0}$ is not necessarily strongly continuous on $\mathfrak{B}(\mathcal{H})$,
one can still associate with this semigroup a generator $\widetilde{L}$ in the weak$^*$-topology:
\begin{equation}\label{L-hat}
\widetilde{L} \, A := w^*-\lim_{t\rightarrow +0} \frac{1}{t}(\mathcal{L}^{t \,*} A - A ) \ ,
\end{equation}
with domain
\begin{equation}\label{dom-L-hat}
{\rm{dom}}(\widetilde{L}):= \{A \in \mathfrak{B}(\mathcal{H}):
w^*-\lim_{t\rightarrow +0} \frac{1}{t}(\mathcal{L}^{t \,*} A - A ) \ \  \exists  \ \} \ .
\end{equation}
It turns out that generator $\widetilde{L}$ is weak$^*$-densely defined, closed operator, which coincides,
$\widetilde{L} A = L^* A$, with the adjoint operator $L^*$, i.e.
\begin{align}\label{dom-L-adj}
&{\rm{dom}}(\widetilde{L})= {\rm{dom}}({L}^*):= \\
&\{A \in \mathfrak{B}(\mathcal{H}):  \exists  \ A' \in \mathfrak{B}(\mathcal{H}) \ s.t. \
\langle \rho , A' \rangle = \langle L (\rho) , A \rangle \ {\rm{for \ all}} \  \rho \in {\rm{dom}}({L}) \} \ .
\nonumber
\end{align}
Although the weak$^*$-topology is even  \textit{weaker} than the  {weak} topology on the Banach space
$\mathfrak{B}(\mathcal{H})$ the above arguments make legitime the characterization of semigroup
$\{\mathcal{L}^{t \,*}:= e^{t \, {L}^*}\}_{t\geq0}$ by a  {generator} in a close similarity to  {strongly}
continuous case, see \cite{BrRo1}, Ch.3 for further details. The minimal price is that instead of the  $C^*$-algebra
of bounded operators $\mathfrak{B}(\mathcal{H})$ one must consider this space endowed by a  {weaker}, namely the
weak$^*$-topology.

Recall that \textit{von Neumann algebra} is a $C^*$-algebra acting on $\mathcal{H}$, containing identity operator
and closed in the  {weak operator} topology. It has enough room for the weak$^*$-continuous semigroup
$\{e^{t \, {L}^*}\}_{t\geq0}$. The Banach space $\mathfrak{B}(\mathcal{H})$ is example of a von Neumann algebra, while
the Banach space $\mathfrak{C}_{1}(\mathcal{H})$ is evidently not.

Let $\mathfrak{M} \subseteq \mathfrak{B}(\mathcal{H})$ be a von Neumann algebra and the weak$^*$-continuous for all
$A\in\mathfrak{M}$ map $t \mapsto \tau^{t}(A)$ be a (semi)group of $*$-automorphisms of $\mathfrak{M}$. Then the
pair ($\mathfrak{M}, \tau^{t}$) is called a $W^*$-dynamical system. In our case $\tau^{t} = \mathcal{L}^{t \,* }$.

If the semigroup (\ref{QMD}) is trace-preserving (\ref{Tr-inv}), then the  adjoint semigroup (\ref{dual-semigroup}) is a
unity-preserving ($\mathcal{L}^{t \,* } (\idty) = \idty $) contraction: $\|\mathcal{L}^{t \,* } (A)\| \leq \|A\|$, see
(\ref{op-norm}). This dual map inherits the property to be a  {completely positive} semigroup, which was
established for Markovian dynamics (\ref{QMD}), see \cite{Dav1} and \cite{Dav2}. By consequence, one gets for
generator ${L}^*$ the analogue of the Kossakowski-Lindblad representation  \cite{FrVe,AlFa}:
\begin{equation}\label{K-L-gen-III}
{L}^*(A) = i \, [H , A] + \frac{1}{2}\sum_{\alpha}  ( \widehat{V}_{\alpha}^{*} [A , \widehat{V}_{\alpha}] +
[\widehat{V}_{\alpha}^{*}\,, A ] \widehat{V}_{\alpha}) \ , \ \ A \in \mathfrak{B}(\mathcal{H}) \ ,
\end{equation}
see (\ref{K-L-gen-I}), (\ref{K-L-gen-II}) and (\ref{dom-L-adj}).

Duality (\ref{dual-semigroup}) is useful for control the state evolution $\mathcal{L}^{t} (\rho)$ and,
in particular, for the proof of the $t\rightarrow\infty$ limit $\rho_{\infty}$ by calculation of this limit on
observables. Since for any $\tau \geq 0$ and $\rho \in \mathfrak{C}_{1}(\mathcal{H})$ , $A \in \mathfrak{B}(\mathcal{H})$:
\begin{equation}\label{steady}
\Tr_{\cH}(\rho_{\infty} \, A) = \lim_{t\rightarrow\infty} \Tr_{\cH}(\mathcal{L}^{\tau+t}(\rho), A) =
\Tr_{\cH}(\mathcal{L}^{\tau}(\rho_{\infty}) , A ) \ ,
\end{equation}
we conclude that $\rho_{\infty}$ is $\mathcal{L}^{\tau}$-invariant (\textit{steady}) state. To elucidate topology of the
density matrix convergence, recall that for the Kossakowski-Lindblad generator (\ref{K-L-gen-I}), (\ref{K-L-gen-II})
dynamics of state is trace-preserving (\ref{Tr-inv}): $\|\mathcal{L}^{t}(\rho)\|_{\mathfrak{C}_{1}} = 1$, and that
(\ref{steady}) implies the \textit{weak-operator} convergence $\mathcal{L}^{t}(\rho)\rightarrow \rho_{\infty}$. Then
(\ref{steady}) is equivalent to the trace-norm convergence of density matrices:
\begin{equation}\label{Tr-norm-conv}
\lim_{t\rightarrow\infty} \|\mathcal{L}^{t}(\rho) - \rho_{\infty}\|_{\mathfrak{C}_{1}} = 0 \ ,
\end{equation}
see e.g. \cite{Za}, Ch.2.4.
\smallskip

\noindent \textbf{A.3 $W^*$-Dynamics, Steady States and Return to Equilibrium.}
For open cavity (\ref{Damp-Pump-Cav}) the choice of operators $\{V_\alpha\}_\alpha$ in (\ref{K-L-gen-III}) is defined by
(\ref{non-Ham-gen1}). Since the cavity is a boson system, the natural (Fock) representation of the  {Canonical
Commutation Relations} (CCR) involves  {one-mode} unbounded creation and annihilation operators $b^*$ and $b$ acting
on the Hilbert
space $\mathcal{H} = \mathcal{H}_\mathcal{C}$. It is a boson Fock space $\mathcal{H} = \mathfrak{F}_{B}(\mathbb{C})$ over
the  {one-dimensional} subspace $\mathfrak{h}=\{\zeta \, \phi\}_{\zeta\in\mathbb{C}}$ (of a Hilbert space)
corresponding to this one photon mode $\phi$.

To avoid the problems with unbounded operators and to keep the evolution of observables in the space
$\mathfrak{B}(\mathcal{H})$ one considers the corresponding $*$-algebra of bounded Weyl operators \cite{BrRo2} in the
form:
\begin{equation}\label{Weyl-I}
\mathfrak{W}(\mathbb{C}):= \left\{ W(\zeta)=
\exp \left[ \frac{i}{\sqrt{2}}\ (\overline{\zeta} \, b + \zeta \, b^*) \right] \right\}_{\zeta\in \mathbb{C}} \ ,
\end{equation}
that verify the Weyl CCR-relations
\begin{equation}\label{Weyl-II}
W(\zeta_1) W(\zeta_2) =  e^{- i \, {\rm{Im}}(\overline{\zeta_1} \,  \zeta_2)/2 } \, W(\zeta_1 + \zeta_2)  \ ,
\end{equation}
Note that the family $\{ W(\zeta)\}_{\zeta\in \mathbb{C}}$ is continuous on $\mathcal{H}$ in the  {strong operator}
sense, but it is  {not} continuous in the $C^*$-algebra topology: $\|W(\zeta) - \idty \| = 2$ for any $\zeta \neq 0$.

By virtue of (\ref{non-Ham-gen1}) and (\ref{K-L-gen-II}) one gets for the adjoint generator (\ref{K-L-gen-III}) of the
one-mode open cavity with $H_{\mathcal{C}} = \epsilon \, b^*b$ :
\begin{align}\label{K-L-gen-IV}
&{L}^{*}_{\mathcal{C},\sigma}(A) = i \, [\epsilon \, b^* b, A]  \\
&+\frac{1}{2} \left(\sigma_{-} b^{*} [A , b] + \sigma_{-} [b^{*}\,, A ] b  +
\sigma_{+} b [A , b^{*}] + \sigma_{+} [b\,, A ] b^{*} \right) \ , \ \ A \in \mathfrak{B}(\mathcal{H}) \ .
\nonumber
\end{align}
Then the adjoint semigroup equation:
\begin{equation}\label{cavity-adj-evol}
\partial_{t} \mathcal{L}^{t \,* }_{\mathcal{C},\sigma} \, (A) =
\mathcal{L}^{t \,* }_{\mathcal{C},\sigma} \,({L}^{*}_{\mathcal{C},\sigma}(A)) \ ,
\end{equation}
allows to calculate evolution of the Weyl operators (\ref{Weyl-I}): $A = W(\zeta)$,  explicitly:
\begin{equation}\label{Weyl-III}
\mathcal{L}^{t \,* }_{\mathcal{C},\sigma} \, (W(\zeta)) = e^{- \Omega_{t}(\zeta)} \  W(\zeta (t)) \ .
\end{equation}
Here
\begin{equation}\label{Weyl-IV}
\Omega_{t}(\zeta) := \frac{|\zeta|^2}{4} \, \frac{\sigma_{-} + \sigma_{+}}{\sigma_{-} - \sigma_{+}}
\left\{1 - e^{- (\sigma_{-} - \sigma_{+}) t}\right\}  \ , \
\zeta (t) := \zeta \ e^{i \, \epsilon t - (\sigma_{-} - \sigma_{+}) t /2} \  .
\end{equation}

Since $\|W(\zeta) - \idty \| = 2$, the evolution (\ref{Weyl-III}) is not continuous in the $C^*$-algebra topology, but
it does in the weak$^*$-topology on the von Neumann algebra $\overline{\mathfrak{W}(\mathbb{C})}$ generated by
(\ref{Weyl-I}) and the weak operator closure. Hence, the pair $(\overline{\mathfrak{W}(\mathbb{C})},
\mathcal{L}^{t \,* }_{\mathcal{C},\sigma})$ is $W^*$-dynamical system, see \cite{BrRo1,AJPI}.

Note that by differentiating of $\mathcal{L}^{t \,* }_{\mathcal{C},\sigma} \, (W(\zeta))$ with respect to
$\zeta$ and $\overline{\zeta}$  one can calculate the evolution of polynomials of creation-annihilation operators
in the weak$^*$-topology. For example of the photon number operator
$N(t):=\mathcal{L}^{t \,* }_{\mathcal{C},\sigma}(b^* b)$, the formal evolution equation follows directly from
(\ref{K-L-gen-IV}) for $A= b^* b$, and from (\ref{cavity-adj-evol}):
\begin{equation}\label{Eq-N(t)}
\partial_{t}N(t) = - (\sigma_{-} - \sigma_{+}) N(t) + \sigma_{+} \ , \ N(t=0) = b^* b \ .
\end{equation}
If $\sigma_{-} > \sigma_{+}$ (leaking is stronger then pumping), then one gets
\begin{equation}\label{N-phot}
N(t) = e^{- (\sigma_{-} - \sigma_{+}) t} \ b^* b + \frac{\sigma_{+}}{\sigma_{-} - \sigma_{+}}
\left\{1 - e^{- (\sigma_{-} - \sigma_{+}) t}\right\} \ .
\end{equation}
By consequence, (\ref{N-phot}) formally implies
\begin{equation}\label{Conv-N}
\lim_{t\rightarrow \infty} \mathcal{L}^{t \,* }_{\mathcal{C},\sigma} \, ( b^* b ) =
\idty  \ \frac{\sigma_{+}}{\sigma_{-} - \sigma_{+}} \ ,
\end{equation}
and (\ref{Weyl-III}) and (\ref{Weyl-IV}) yield
\begin{equation}\label{Conv-Equil}
w^*-\lim_{t\rightarrow \infty} \mathcal{L}^{t \,* }_{\mathcal{C},\sigma} \, (W(\zeta)) =
\idty \ \exp\{- \frac{|\zeta|^2}{4}  \ \frac{\sigma_{-} + \sigma_{+}}{\sigma_{-} - \sigma_{+}}\} \ ,
\end{equation}
in the weak$^*$-topology.

Let the  {initial} cavity state $\rho$ be such that $\overline{N}(0):=\Tr_{\cH_\mathcal{C}}(\rho \,
b^* b) < \infty$, and $\sigma_{-} > \sigma_{+} > 0$. Then the limit (\ref{Conv-N}) implies a nontrivial stationary
expectation value of the photon number:
\begin{align}\label{state-I}
&\overline{N}(t) := \lim_{t\rightarrow \infty} \Tr_{\cH_\mathcal{C}}(\rho \, \mathcal{L}^{t \,*} ( b^* b))=
\lim_{t\rightarrow \infty} \Tr_{\cH_\mathcal{C}}(\mathcal{L}^{t}(\rho) \,  b^* b) \\
&= \Tr_{\cH_\mathcal{C}}(\rho_{\infty} \,  b^* b) =  \frac{\sigma_{+}}{\sigma_{-} - \sigma_{+}} \ , \nonumber
\end{align}
in the  {limiting} cavity state $\rho_{\infty} := \lim_{t\rightarrow \infty} \mathcal{L}^{t}(\rho)$. Similarly,
by (\ref{Conv-Equil}) one gets:
\begin{align}\label{state-II}
&\lim_{t\rightarrow \infty} \Tr_{\cH_\mathcal{C}}(\rho \, \mathcal{L}^{t \,*} (W(\zeta))=
\lim_{t\rightarrow \infty} \Tr_{\cH_\mathcal{C}}(\mathcal{L}^{t}(\rho) \, W(\zeta))\\
&= \Tr_{\cH_\mathcal{C}}(\rho_{\infty} \,  W(\zeta)) =
\exp\{- \frac{|\zeta|^2}{4} \ \frac{\sigma_{-} + \sigma_{+}}{\sigma_{-} - \sigma_{+}}\} \ . \nonumber
\end{align}
By (\ref{Weyl-I}) and (\ref{state-I}), (\ref{state-II}) we obtain that the limiting (\textit{steady}) density matrix
$\rho_{\infty}$ of the open cavity corresponds to the one-mode boson equilibrium state:
\begin{equation}\label{lim-state}
\rho_{\infty} =: \rho^{\beta_{cav}} =
(1- e^{- \beta_{{\rm{cav}}} \epsilon} ) \ e^{- \beta_{{\rm{cav}}} \epsilon \ b^* b} \ , \ \ \ \
\beta_{{\rm{cav}}} := \frac{1}{\epsilon} \ln \frac{\sigma_{-}}{\sigma_{+}} \ .
\end{equation}

Therefore, if $\sigma_{-} > \sigma_{+} > 0$ and $\epsilon > 0$, the one-mode pumped leaky cavity evolves from
\textit{any} initial state verifying $\overline{N}(0)< \infty$, to the \textit{equilibrium} Gibbs state with temperature
$\theta_{cav} = 1/\beta_{{cav}}$ entirely defined by the leaking-pumping intensities $\sigma_{\mp}$.
For this limit of {return} to the thermal  {equilibrium} one can distinguish two intuitively clear
extreme cases. The zero-temperature case for  {zero} pumping: $\sigma_{+} = 0$, and the \textit{infinite}
temperature case, when the pumping is not dominated by the leaking: $\sigma_{+} \uparrow \sigma_{-}$. In the first
case the photon-number mean-value (\ref{state-I}) is zero, whereas it is infinite in the second case.

Note that the \textit{time} and the \textit{pumping} limits do {not} commute. Let the initial expectation of
photons in the cavity $\overline{N}(0) < \infty$. Then applying the limit $\sigma_{+} \uparrow \sigma_{-}$ to the
expectation of (\ref{N-phot}) we obtain
\begin{equation}\label{photon-linear}
\lim_{\sigma_{+} \uparrow \sigma_{-}} \overline{N}(t) =  \overline{N}(0) + \sigma_{+} \ t    \ .
\end{equation}
Therefore, one has a \textit{linear} asymptotic increasing of the photon-number mean-value in the limit, when the
leaking and pumping rates coincide: $\sigma_{-} = \sigma_{+} >0$.

\smallskip
\noindent \textbf{A.4 Completely Positive Quasi-Free Dynamics on the $CCR$ Algebra.}
We need a more abstract description of dynamics $\mathcal{L}^{t \,* }_{\mathcal{C},\sigma}$ on the $CCR$ algebra
$\overline{\mathfrak{W}(\mathbb{C})}$ generated by (\ref{Weyl-I}).

Notice that the  {complete positivity} of the map $\mathcal{L}^{t \,* }_{\mathcal{C},\sigma}:
\overline{\mathfrak{W}(\mathbb{C})} \rightarrow \overline{\mathfrak{W}(\mathbb{C})}$ follows from general properties
of the Kossakowski-Lindblad generator (\ref{K-L-gen-III}). For the case of the open cavity this result follows directly
from CCR-relations (\ref{Weyl-II}) and the explicit result (\ref{Weyl-III}). Behind this result there are abstract
observations due to \cite{DVV1,DVV2} and \cite{EvLe2,Van}.

Let $\mathfrak{H}$ be a Hilbert space. Denote by $\Gamma: \mathfrak{H} \rightarrow \mathfrak{H}$ a linear map and by
$\Psi: \mathfrak{H} \rightarrow \mathbb{C}$ a complex function. Consider the Weyl CCR($\mathfrak{H}$) algebra generated by
unitaries:
\begin{equation}\label{Weyl-V}
\mathfrak{W}(\mathfrak{H}):= \left\{ W(f) =
\exp \left[ \frac{i}{\sqrt{2}}\ ( b(f) +  b^{*}(f)) \right] \right\}_{f\in \mathfrak{H}} \ .
\end{equation}
Here linear and anti-linear functions: $f \mapsto b^{*}(f)$ and $f \mapsto b(f)$, are creation and annihilation
operators in the boson Fock space $\mathfrak{F}_{B}(\mathfrak{H})$ over $\mathfrak{H}$. Then the Weyl CCR-relations take
the form
\begin{equation}\label{Weyl-VI}
W(f) W(g) =  e^{- i \, {\rm{Im}}(f,g)_{\mathfrak{H}}/2 } \, W(f + g)  \ , \ f,g\in \mathfrak{H} \ ,
\end{equation}
where $(\cdot , \cdot)_{\mathfrak{H}}$ is the scalar product in $\mathfrak{H}$.

Recall now that a linear map $T: \mathfrak{W}(\mathfrak{H}) \rightarrow \mathfrak{W}(\mathfrak{H})$ on the
CCR($\mathfrak{H}$) algebra (\ref{Weyl-V}) is called \textit{quasi-free}, if it has the form, \cite{DVV1,DVV2}:
\begin{equation}\label{T-quasi-free1}
T(W(f))= \Psi(f) \ W(\Gamma(f)) \ , \ f\in \mathfrak{H} \ .
\end{equation}
Then the unity-preserving \textit{quasi-free} \textit{semigroup} $\{T_{t}\}_{t \geq 0}$ on the CCR($\mathfrak{H}$) algebra
$\mathfrak{W}(\mathfrak{H})$ is defined in a similar way:
\begin{equation}\label{T-quasi-free2}
T_{t}(W(f)):= \Psi_{t}(f) \ W(\Gamma_{t}(f)) \ , \ f\in \mathfrak{H} \ ,
\end{equation}
where $\Psi_{t=0}(f) = 1$, $\Gamma_{t=0} = \idty$ and $\Psi_{t}(f=0) = 1$, $\Gamma_{t}(f=0) = 0$ .
Note that the semigroup property of $\{T_{t}\}_{t \geq 0}$  and (\ref{T-quasi-free2}) imply
\begin{align}
& T_{s+t}(W(f))= T_{s}(T_{t}(W(f)))   \\
&= \Psi_{s}(\Gamma_{t}(f)) \Psi_{t}(f)\ W(\Gamma_{s}(\Gamma_{t}(f)))= \Psi_{s+t}(f) W(\Gamma_{s+t}(f))    \ . \nonumber
\end{align}
Then linear independence of the Weyl operators yields
\begin{equation}\label{T-quasi-free3}
\Gamma_{s+t}(f) = \Gamma_{s} (\Gamma_{t} (f)) \ \ \ {\rm{and}} \ \ \
 \Psi_{s+t}(f)= \Psi_{s}(\Gamma_{t}(f)) \Psi_{t}(f) \ .
\end{equation}
Hence, $\{\Gamma_{t}\}_{t \geq 0}$ is in turn a semigroup on $\mathfrak{H}$.

Let $t\mapsto \Psi_{t}(f)$ be continuous and for each $f\in \mathfrak{H}$ be differentiable at $t=+0$:
\begin{equation}\label{der-Psi}
\Psi_{0}^{\prime}(f) := \lim_{t \rightarrow +0} \partial_{t}\Psi_{t}(f) \ ,
\end{equation}
such that the function $t\mapsto \Psi_{0}^{\prime}(\Gamma_{t}(f))$ be bounded. Since semigroup (\ref{T-quasi-free2}) is
unity-preserving, then (\ref{T-quasi-free3}) and  $W(f=0)= \idty$ imply that for any $f\in \mathfrak{H}$
\begin{equation}\label{Psi-solut}
\Psi_{t}(f) = \exp \left\{\int_{0}^{t} d\tau \Psi_{0}^{\prime}(\Gamma_{\tau}(f))\right\} \ ,
\end{equation}
where $\Psi_{0}^{\prime}(0) = 0$ and $\Psi_{0}^{\prime}(-f) = \overline{\Psi_{0}^{\prime}(f)}$, see \cite{DVV2}.

Note that for the particular case: $\mathfrak{H} =\mathbb{C}$ and for identification of (\ref{Weyl-I}) with (\ref{Weyl-V}),
we reproduce the results (\ref{Weyl-III}) and (\ref{Weyl-IV}) in \textbf{A.3} for
$T_{t}= \mathcal{L}^{t \,* }_{\mathcal{C},\sigma}$. To this end one has to use (\ref{T-quasi-free3}), (\ref{der-Psi})
and to put
\begin{equation}\label{identif}
\Gamma_{t}(\zeta) := \zeta \ e^{i \, \epsilon t - (\sigma_{-} - \sigma_{+}) t /2} \ \ \ \ {\rm{and}}
\ \ \ \ \Psi_{0}^{\prime}(\zeta):= - \, \frac{|\zeta|^2}{4} (\sigma_{-} + \sigma_{+}) \ .
\end{equation}
Then (\ref{Psi-solut}) gives $\Psi_{t}(\zeta) = e^{- \Omega_{t}(\zeta)}$, where  $\Omega_{t}(\zeta)$ is defined by
(\ref{Weyl-IV}).

\smallskip
\noindent \textbf{A.5 Quasi-Free States on the $CCR$ Algebra.}
Since the linear hull of the the Weyl operators is dense in CCR($\mathfrak{H}$), any state $\omega(\cdot)$ is uniquely
determined by its values taken on $\{W(f)\}_{f\in \mathfrak{H}}$. Therefore a state $\omega$ is completely defined
by its characteristic functional
\begin{equation*}
\mathfrak{H}\ni f \mapsto \omega(W(f)).
\end{equation*}
A state $\omega$ is called \textit{regular} if the function $a\mapsto \omega(W(af))$ is continuous for all
$f\in\mathfrak{H}$.
Characteristic functionals of regular states on CCR($\mathfrak{H}$) are characterized by Araki and Segal
\cite{AJPI,BrRo1} in the following theorem.

\begin{theorem}\label{Araki-Segal}
A map $\mathfrak{H}\ni f\mapsto \omega(W(f))\in\bC$ is the characteristic functional of a regular state $\omega$
on CCR($\mathfrak{H}$) if and only if
\begin{enumerate}
\item $\omega(W(0))=1.$
\item The function $a \mapsto \omega(W(af))$ is continuous for all $f\in\mathfrak{H}$.
\item For any integer $n\geq 2,$ all $f_1,...f_n\in\mathfrak{H}$ and all $z_1,...z_n\in\bC$ one has
\begin{equation*}
\sum_{j,k=1}^n\omega(W(f_j-f_k))e^{-i\Im(f_j,f_k)/2}\overline{z_j}z_k\geq 0.
\end{equation*}
\end{enumerate}
\end{theorem}

We remind that the state $\omega_{r,s}(\cdot)$ is called
\textit{quasi-free} \cite{Ver11}, if
\begin{equation}\label{QFree1}
\omega_{r,s}(W(f)) = \exp \{ i \, r(f) - \frac{1}{2} \, s(f,f)\} \ \ , \ \ f\in \mathfrak{H} \ .
\end{equation}
Here $r$ is a linear functional on $\mathfrak{H}$, whereas $s$ is a non-negative (closable) sesquilinear form on
$\mathfrak{H}\times\mathfrak{H}$, that verifies
\begin{equation}\label{QFree2}
\frac{1}{4} \, |{\rm{Im}}(f,g)_{\mathfrak{H}}|^2 \leq s(f,f) \, s(g,g) \ \ , \ f,g \in \mathfrak{H} \ ,
\end{equation}
to ensure the  {positivity} of this state \cite{Ver11}. By the Araki-Segal theorem the quasi-free states are
\textit{regular}
and \textit{analytic}, verifying the equations:
\begin{equation}\label{r-s}
r(f) = \omega_{r,s}(\Phi(f)) \ \ \ {\rm{and}} \ \ \ s(f,f) = \omega_{r,s}(\Phi(f)^2)- \omega_{r,s}(\Phi(f))^2 \ ,
\end{equation}
where $\Phi(f):= (b(f) +  b^{*}(f))/\sqrt{2}$, see, e.g. \cite{Ver11,BrRo2}.

If the state (\ref{QFree1}) is \textit{gauge-invariant}:
$\omega_{r,s}(W(f))= \omega_{r,s}(W(e^{i \, \varphi} \, f ))$, then $r(\cdot) =0$ and we denote this state by
$\omega_{s}(\cdot) := \omega_{r=0,s}(\cdot)$. By virtue of (\ref{state-II}) the limiting ({steady}) state
(\ref{lim-state}) of the open cavity is gauge invariant and quasi-free  with:
\begin{equation}\label{lim-QF}
r(\zeta)=0  \ ,  \ s(\zeta,\zeta) =
\frac{|\zeta|^2}{2} \ \frac{\sigma_{-} + \sigma_{+}}{\sigma_{-} - \sigma_{+}}=
\Tr_{\cH_\mathcal{C}}(\rho^{\beta_{cav}} \, \Phi(\zeta)^2) \ , \ \zeta \in \mathfrak{H} =\mathbb{C} \ ,
\end{equation}
Moreover it is the \textit{equilibrium}, i.e. $(\beta_{cav})$-KMS state with respect to the Hamiltonian
dynamics $\tau^{t}$ generated by $H_{\mathcal{C}}$.

It is worth to note that  {\textit{a priori}} there is no evidence
that the  {steady} state of the open cavity (\ref{Damp-Pump-Cav}) must be either quasi-free, or the Gibbs
equilibrium state. For discussion of the theory of \textit{non-equilibrium} quasi-free steady states we suggest a very
complete account in \cite{AJPI}-\cite{AJPIII}. There one can find curious examples of non-equilibrium ({\textit{non}}-KMS)
quasi-free states, which allow a certain  {informal} Gibbs description via \textit{long-range} many-body interactions
\cite{AschPi}.

Note that by definitions (\ref{T-quasi-free2}) and (\ref{QFree1}) the {quasi-free} dynamics maps the
 {quasi-free states} into the states:
\begin{align}\label{QF-QF}
\omega_{r, s}(T_{t}(W(f))) = \Psi_{t}(f) \ \omega_{r, s}(W(\Gamma_{t}(f))) =
\Psi_{t}(f) \ \omega_{r_{t}, s_{t}}(W(f)) \ ,
\end{align}
where $r_{t}(f) := r(\Gamma_{t}(f))$ and $s_{t}(f,f):= s(\Gamma_{t}(f),\Gamma_{t}(f))$. In general the states
(\ref{QF-QF}) are  {not} quasi-free.

Let us consider the case of \textit{Gaussian} {quasi-free} dynamics \cite{Van}, when the semigroup
$\{\Gamma_{t}\}_{t \geq 0}$ defined on $\mathfrak{H}$ by
\begin{equation}\label{Gauss-Gamma}
\Gamma_{t} := \exp \{i \, t H - \frac{1}{2}\, t \,(\Sigma_{-} - \Sigma_{+}) \} \
\end{equation}
here $H$ is a self-adjoint operator and $\Sigma_{\mp}$ are  {bounded} positive operators on $\mathfrak{H}$
such that $\Sigma_{-} \geq \Sigma_{+} \geq 0$, and (\ref{der-Psi}) is a bilinear form
\begin{equation}\label{der-Psi1}
\Psi_{0}^{\prime}(f) := - \, \frac{1}{4}\, (f, R f) \ , \ f \in \mathfrak{H} \ ,
\end{equation}
defined by a positive  {bounded} operator $R \geq \Sigma_{-}$. Then by virtue of (\ref{Psi-solut}), (\ref{QFree1})
and (\ref{QF-QF}) dynamics $T_{t}$ maps initial quasi-free state $\omega_{r, s}$ into quasi-free state
$\omega_{r_{t}, \widetilde{s}_{t}}$ with
\begin{equation}\label{s-t}
\widetilde{s}_{t}(f,f):= s(\Gamma_{t}(f),\Gamma_{t}(f))+ \, \frac{1}{2}\,
\int_{0}^{t} d\tau \, (\Gamma_{\tau}(f), R \, \Gamma_{\tau}(f)) \ .
\end{equation}

In the particular case of Hamiltonian dynamics ($\Sigma_{\mp} = 0$ and $R=0$) the quasi-free map
(\ref{T-quasi-free2}) is the group of Bogoliubov automorphisms on $\mathfrak{W}(\mathfrak{H})$:
\begin{equation}\label{T-Bog}
T_{t}(W(f))= W( e^{i tH} f) \ , \ \ \ f\in \mathfrak{H} \ .
\end{equation}
Automorphism (\ref{T-Bog}) is the simplest quasi-free dynamics, which corresponds to the unitary one-particle
evolution generated by the Hamiltonian $H$.

For the example of the open cavity, when $\mathfrak{H} =\mathbb{C}$, we use (\ref{identif}) to establish that
in this case one has to put $H=\epsilon$, $\Sigma_{\mp} = \sigma_{\mp}$, $R = \sigma_{-} + \sigma_{+}$ in
(\ref{Gauss-Gamma}) and (\ref{der-Psi1}). Then (\ref{s-t}) yields
\begin{equation}\label{diff-ss}
\widetilde{s}_{t}(\zeta,\zeta) = s(\Gamma_{t}(\zeta),\Gamma_{t}(\zeta)) +
\frac{1}{2}\, |\zeta|^2 \ \frac{\sigma_{-} + \sigma_{+}}{\sigma_{-} - \sigma_{+}} \
(1 - e^{- (\sigma_{-} - \sigma_{+}) t}) \ .
\end{equation}
This means that dynamics generated by (\ref{K-L-gen-IV}) is  {quasi-free}:
$\mathcal{L}^{t \,* }_{\mathcal{C},\sigma} \, (W(\zeta)) = T_{t}(W(\zeta))$, see (\ref{identif}), and that it
 {preserves} the quasi-free states.

 Since $\lim_{t\rightarrow\infty} \Gamma_{t}(\zeta) =0$, by
(\ref{diff-ss}) we obtain that
\begin{equation}\label{lim-ss}
\lim_{t\rightarrow\infty} \omega_{r, s}(\mathcal{L}^{t \,* }_{\mathcal{C},\sigma}(W(\zeta))) =
\exp\{- \frac{|\zeta|^2}{4} \ \frac{\sigma_{-} + \sigma_{+}}{\sigma_{-} - \sigma_{+}}\} \ ,
\end{equation}
for any initial quasi-free state $\omega_{r, s}$.

In fact we established in (\ref{state-II}) that for $t\rightarrow\infty$ the quasi-free dynamics
$\mathcal{L}^{t \,* }_{\mathcal{C},\sigma}$ transforms \textit{any} initial state $\rho$ with finite expectation
of photon number into the \textit{limit state}, which is quasi-free and Gibbs, see (\ref{state-II}), (\ref{lim-state})
and (\ref{lim-ss}).

\vskip 2.0cm

\newpage

\end{document}